\RequirePackage{amsmath}
\documentclass[runningheads]{llncs}
\usepackage{multicol}
\usepackage{amsfonts}
\usepackage{amssymb}
\usepackage{graphicx}
\usepackage{epic}
\usepackage{eepic}
\usepackage{epsfig,float}
\usepackage{verbatim}
\usepackage{pdfsync}
\usepackage{gastex}
\usepackage[lowtilde]{url}
\usepackage{geometry}
\usepackage{enumitem}
\usepackage{url} 

\usepackage{bbm}
\usepackage{xcolor}

\DeclareMathOperator{\lcm}{lcm}

\title{Complexity of Prefix-Convex Regular Languages\thanks{This work was supported by the Natural Sciences and Engineering Research Council of Canada 
grant No.~OGP0000871.}
}

\author{Janusz~Brzozowski and Corwin~Sinnamon}
\authorrunning{J. Brzozowski and C. Sinnamon}  
\titlerunning{Complexity of Prefix-Convex Languages}

\institute{David R. Cheriton School of Computer Science, University of Waterloo, \\
Waterloo, ON, Canada N2L 3G1\\
{\tt brzozo@uwaterloo.ca,sinncore@gmail.com}
}

\begin{document}

\maketitle

\begin{abstract}
A language $L$ over an alphabet $\Sigma$ is prefix-convex if, for any words $x,y,z\in\Sigma^*$, whenever $x$ and $xyz$ are in $L$, then so is $xy$.
Prefix-convex languages include right-ideal, prefix-closed, and prefix-free languages.
We study complexity properties of prefix-convex regular languages.
In particular, we find the quotient/state complexity of boolean operations, product (concatenation), star,    and reversal,  the size of the syntactic semigroup, and the quotient complexity of atoms.
For binary operations we use arguments  with different alphabets when appropriate; this leads to higher tight upper bounds than those obtained with equal alphabets. 
We exhibit most complex prefix-convex languages that meet the complexity bounds for all the measures listed above.
\medskip

 \end{abstract}

\section{Motivation} 
Many formal definitions of key concepts are postponed until Section~\ref{sec:basics}.
\smallskip

We study the complexity of prefix-convex regular languages~\cite{AnBr09,Thi73}. For words $w,x,y$ over an alphabet $\Sigma$, if $w=xy$, then $x$ is a \emph{prefix} of $w$.
A~language $L\subseteq \Sigma^*$ is \emph{prefix-convex} if, whenever $x$ and $xyz$ are in $L$,  then $xy$ is also in $L$.
The class of prefix-convex languages includes three well-known subclasses: right-ideal, prefix-closed, and prefix-free languages.
A language $L$ is a \emph{right ideal} if it is non-empty and satisfies the equation $L=L\Sigma^*$.
Right ideals play a role in pattern matching: If one is searching for all words beginning with words in some language $L$ in a given text (language), then one is looking for words in $L\Sigma^*$. 
Right ideals also constitute a basic concept in semigroup theory.
A language $L$ is \emph{prefix-closed} if, whenever $w$ is in $L$ and $x$ is a prefix of $w$,  then $x$ is also in $L$.
The complement of every right ideal is a prefix-closed language.
The set of allowed input sequences to a digital system is a prefix-closed language.
A language $L$ is \emph{prefix-free} if no word in $L$ is a prefix of another word in $L$.
Prefix-free languages (other than $\{\varepsilon\}$, where $\varepsilon$ is the empty word) 
are prefix codes.
They play an important role in coding theory, and have many applications~\cite{BPR10}. 

The \emph{alphabet of a regular language} $L$ is $\Sigma$ (or \emph{$L$ is a language over $\Sigma$}) if $L \subseteq \Sigma^*$ and every letter of $\Sigma$ appears in a word of $L$.
The \emph{(left) quotient} of $L$ by a word $w\in\Sigma^*$ is $w^{-1}L=\{x\mid wx\in L\}$. 
A language is regular if and only if it has a finite number of distinct quotients. So the number of quotients of $L$ is a natural measure of complexity for $L$; it is called the 
\emph{quotient complexity}~\cite{Brz10} of $L$ and is denoted it by $\kappa(L)$.
An equivalent concept is the \emph{state complexity}~\cite{Yu01} of $L$, which is the number of states in a complete minimal deterministic finite automaton (DFA) recognizing $L$.
If $L_n$ is a regular language of quotient complexity $n$, and $\circ$ is a unary operation, then 
the \emph{quotient/state complexity  of $\circ$} 
is  the maximal value of $\kappa(L_n^\circ)$, expressed as a function $f(n)$,
as $L_n$ ranges over all regular languages of complexity $n$.
If $L'_m$ and $L_n$ are regular languages of quotient complexities $m$ and $n$ respectively, and $\circ$ is a binary operation, then
the \emph{quotient/state complexity of $\circ$} 
is  the maximal value of $\kappa(L'_m \circ L_n)$, expressed as a function $f(m,n)$,
as $L'_m$ and $L_n$ range over all regular languages of complexity $m$ and $n$, respectively.
The quotient/state complexity of an operation gives a worst-case lower bound on the time and space complexities of the operation, and has been studied extensively~\cite{Brz10,Brz13,Yu01};
we refer to quotient/state complexity simply as \emph{complexity}.

In all the past literature on binary operations it has always been assumed that the alphabets of the two operands are restricted to be the same. However, it has been shown recently~\cite{Brz16} that this is an unnecessary restriction: larger complexity bounds can be reached in some cases if the alphabets differ.
In the present paper we examine both \emph{restricted complexity} of binary operations, where the alphabets must be the same, and \emph{unrestricted complexity} where they may differ.

To find the complexity of a unary operation one first finds an upper bound on this complexity, and then exhibits languages that meet this bound. Since we require a language $L_n$ for each $n \ge k$, we need a sequence $(L_k, L_{k+1}, \dots)$; here $k$ is usually a small integer because the bound may not hold for a few small values of $n$. We call such a sequence a \emph{stream} of languages. Usually the languages in a stream have the same basic structure and differ only in the parameter $n$. For example, $((a^n)^* \mid n\ge 2)$ is a stream.
For a binary operation we require two streams. 

There exists a stream $(L_3, L_4, \dots)$ of regular languages $L_n(a,b,c,d)$  which 
 meets the complexity bounds for reversal, (Kleene) star, product (concatenation), and all binary boolean operations~\cite{Brz16}. 
There are reasons, however, why the complexity of languages is not a sufficiently good measure. 
Two languages may have the same complexity $n$ but the syntactic semigroup~\cite{Pin97} of one may have $n-1$ elements, while that of the other has $n^n$ elements~\cite{BrYe11}. For this reason, the size of the syntactic semigroup of a language -- which is the same as the size of the transition semigroup of a minimal DFA accepting  the language~\cite{Pin97} -- has been added as another complexity measure. 
Secondly, \emph{star-free} languages, which constitute a very special subclass of the class of regular languages, meet the complexity bounds of regular languages for all operations except reversal, which only reaches the bound $2^n-1$ instead of $2^n$~\cite{BrLi12}.
While regular languages are the smallest class  containing the finite languages and closed under boolean operations, product  and  star, 
star-free languages are the smallest class containing the finite languages and closed only under boolean operations and product.
Quotient/state complexity does not distinguish these two classes. 

The complexities of the atoms of a regular language have been proposed as an additional measure~\cite{Brz13}.
Atoms are defined by the following left congruence: two words $x$ and $y$ are equivalent if 
 $ux\in L$ if and only if  $uy\in L$ for all $u\in \Sigma^*$. 
 Thus $x$ and $y$ are equivalent if
 $x\in u^{-1}L$ if and only if $y\in u^{-1}L$.
 An equivalence class of this relation is an \emph{atom} of $L$~\cite{BrTa14,Iva16}.
Thus an atom is a non-empty intersection of complemented and uncomplemented quotients of $L$. 
 If $K_0, \dots, K_{n-1}$ are the quotients of $L$, and $S\subseteq Q_n= \{0,\dots,n-1\}$, then atom $A_S$ is the intersection of quotients with subscripts in $S$ and complemented quotients with subscripts in $Q_n\setminus S$.
For more information about atoms see~\cite{BrTa13,BrTa14,Iva16}.

A language stream that meets the  restricted and unrestricted complexity bounds for all boolean operations, product, star, and reversal, and also has the largest syntactic semigroup and most complex atoms~\cite{Brz13,Brz16}, is called a \emph{most complex stream}.
A most complex stream should have the smallest possible alphabet sufficient to meet all the bounds.
Here we present a  regular language stream similar to that of~\cite{Brz16} but one that is better suited to prefix-convex languages.
Most complex streams are useful when one designs a system dealing with regular languages and finite automata. If one would like to know the maximal sizes of automata the system can handle, one can use the one most complex stream  to test all the operations.

\smallskip

In this paper we exhibit most complex language streams for right-ideal, prefix-closed,  prefix-free, and proper prefix-convex languages, where a prefix-convex language is \emph{proper}  if it is not one of the first three types.

\section{Finite Automata, Transformations, and Semigroups}
\label{sec:basics}
A \emph{deterministic finite automaton (DFA)} is a quintuple
$\mathcal{D}=(Q, \Sigma, \delta, q_0,F)$, where
$Q$ is a finite non-empty set of \emph{states},
$\Sigma$ is a finite non-empty \emph{alphabet},
$\delta\colon Q\times \Sigma\to Q$ is the \emph{transition function},
$q_0\in Q$ is the \emph{initial} state, and
$F\subseteq Q$ is the set of \emph{final} states.
We extend $\delta$ to a function $\delta\colon Q\times \Sigma^*\to Q$ as usual.
A~DFA $\mathcal{D}$ \emph{accepts} a word $w \in \Sigma^*$ if ${\delta}(q_0,w)\in F$. The language accepted by $\mathcal{D}$ is denoted by $L(\mathcal{D})$. If $q$ is a state of $\mathcal{D}$, then the language $L^q$ of $q$ is the language accepted by the DFA $(Q,\Sigma,\delta,q,F)$. 
A state is \emph{empty} or \emph{dead} or  \emph{a sink} if its language is empty. Two states $p$ and $q$ of $\mathcal{D}$ are \emph{equivalent} if $L^p = L^q$; otherwise they are \emph{distinguishable}.
A state $q$ is \emph{reachable} if there exists $w\in\Sigma^*$ such that $\delta(q_0,w)=q$.
A DFA is \emph{minimal} if all of its states are reachable and no two states are equivalent.
Usually DFAs are used to establish upper bounds on the complexity of operations and also as witnesses that meet these bounds.

A \emph{nondeterministic finite automaton (NFA)} is a quintuple
$\mathcal{D}=(Q, \Sigma, \delta, I,F)$, where
$Q$,
$\Sigma$ and $F$ are as in a DFA, 
$\delta\colon Q\times \Sigma\to 2^Q$, and
$I\subseteq Q$ is the \emph{set of initial states}. 
An \emph{$\varepsilon$-NFA} is an NFA in which transitions under the empty word $\varepsilon$ are also permitted.

Without loss of generality we use $Q_n=\{0,\dots,n-1\}$ as the set of states of every DFA with $n$ states.
A \emph{transformation} of $Q_n$ is a mapping $t\colon Q_n\to Q_n$.
The \emph{image} of $q\in Q_n$ under $t$ is denoted by $qt$.
In any DFA, each letter $a\in \Sigma$ induces a transformation $\delta_a$ of the set $Q_n$ defined by $q\delta_a=\delta(q,a)$; we denote this by $a\colon \delta_a$. 
By a slight abuse of notation we use the letter $a$ to denote the transformation it induces; thus we write $qa$ instead of $q\delta_a$.
We  extend the notation to sets of states: if $P\subseteq Q_n$, then $Pa=\{pa\mid p\in P\}$.
We also write $P\stackrel{a}{\longrightarrow} Pa$ to mean that the image of $P$ under $a$ is $Pa$.
If $s,t$ are transformations of $Q_n$, their composition is denoted by $s\ast t$ and defined by
$q(s \ast t)=(qs)t$; the $\ast$ is usually omitted.
Let $\mathcal{T}_{Q_n}$ be the set of all $n^n$ transformations of $Q_n$; then $\mathcal{T}_{Q_n}$ is a monoid under composition. 

For $k\ge 2$, a transformation (permutation) $t$ of a set $P=\{q_0,q_1,\ldots,q_{k-1}\} \subseteq  Q_n $ is a \emph{$k$-cycle}
if $q_0t=q_1, q_1t=q_2,\ldots,q_{k-2}t=q_{k-1},q_{k-1}t=q_0$.
This $k$-cycle is denoted by the transformation $(q_0,q_1,\ldots,q_{k-1})$ of $Q_n$, which acts as the identity on the states outside the cycle.
A~2-cycle $(q_0,q_1)$ is called a \emph{transposition}.
A transformation  that sends all the states of $P$ to $q$ and acts as the identity on the remaining states is denoted by $(P \to q)$.  If $P=\{p\}$ we write  $(p\to q)$ for $(\{p\} \to q)$.
 The identity transformation is denoted by $\mathbbm 1$.
 The notation $(_i^j \; q\to q+1)$ denotes a transformation that sends $q$ to $q+1$ for $i\le q\le j$ and is the identity for the remaining states, and  $(_i^j \; q\to q-1)$ is defined similarly.

Let $\mathcal{D} = ( Q_n, \Sigma, \delta, q_0, F)$ be a DFA. For each word $w \in \Sigma^*$, the transition function induces a transformation $\delta_w$ of  $Q_n$ by  $w$: for all $q \in  Q_n $, 
$q\delta_w = \delta(q, w).$ 
The set $T_{\mathcal{D}}$ of all such transformations by non-empty words forms a semigroup of transformations called the \emph{transition semigroup} of $\mathcal{D}$~\cite{Pin97}. 
We can use a set  $\{\delta_a \mid a \in \Sigma\}$ of transformations to define $\delta$, and so the DFA $\mathcal{D}$. 

The \emph{Myhill congruence}~\cite{Myh57} ${\mathbin{\approx_L}}$ of a language $L\subseteq \Sigma^*$ is defined on $\Sigma^+$ as follows:
$$
\mbox{For $x, y \in \Sigma^+$, } x {\mathbin{\approx_L}} y \mbox{ if and only if } wxz\in L  \Leftrightarrow wyz\in L\mbox { for all } w,z \in\Sigma^*.
$$
This congruence is also known as the \emph{syntactic congruence} of $L$.
The quotient set $\Sigma^+/ {\mathbin{\approx_L}}$ of equivalence classes of the relation ${\mathbin{\approx_L}}$ is a semigroup called the \emph{syntactic semigroup} of $L$.
If  $\mathcal{D}$ is a minimal DFA of $L$, then $T_{\mathcal{D}}$ is isomorphic to the syntactic semigroup $T_L$ of $L$~\cite{Pin97}, and we represent elements of $T_L$ by transformations in~$T_{\mathcal{D}}$. 
The size of the syntactic semigroup has been used as a measure of complexity for regular languages~\cite{Brz13,BrYe11,HoKo04,KLS05}.

\smallskip

Recall that binary operations require two language streams to determine the complexity of the operation. 
Sometimes the same stream can be used for both operands, and 
it has been shown in~\cite{Brz13,Brz16} that for all common binary operations on regular languages the second stream can be a ``dialect'' of the first, that is, it can ``differ only slightly'' from the first and all the bounds can still be met. 
Let $\Sigma=\{a_1,\dots,a_k\}$ be an alphabet ordered as shown;
if $L\subseteq \Sigma^*$, we denote it by $L(a_1,\dots,a_k)$ to stress its dependence on $\Sigma$.
A \emph{dialect} of $L$ is a related language obtained by replacing or deleting letters of $\Sigma$ in the words of $L$.
More precisely, for an alphabet $\Sigma'$ and a partial map $\pi \colon \Sigma \mapsto \Sigma'$,
we obtain a dialect of $L$ by replacing each letter $a \in \Sigma$ by $\pi(a)$ in every word of $L$,
or deleting the word entirely if $\pi(a)$ is undefined.
We write $L(\pi(a_1),\dots, \pi(a_k))$ to denote the dialect of $L(a_1,\dots,a_k)$ given by $\pi$,
and we denote undefined values of $\pi$ by  ``$-$''.
For example, if $L(a,b,c)= \{a, ab, ac\}$ then its dialect $L(b,-,d)$ is the language $\{b, bd\}$.
Undefined values for letters at the end of the alphabet are omitted; thus, for example, 
if $\Sigma =\{a,b,c,d,e\}$, $\pi(a)=b$, $\pi(b)=a$, $\pi(c)=c$ and $\pi(d)=\pi(e)=-$, we write $L(b,a,c)$ for $L(b,a,c,-,-)$.

The language stream that meets all the complexity bounds is referred to as the \emph{master} language stream. Every master language stream we present here uses the smallest possible alphabet sufficient to meet all the complexity bounds.
Individual bounds are frequently met by dialects on reduced alphabets, and we prefer to use the smallest alphabet possible for each bound.
For binary operations, we try to minimize the size of the combined alphabet of the two dialects.

As each letter induces a transformation on the states of a DFA (or equivalently, the quotients of a language)
we count the number of distinct transformations induced by letters of the alphabet.
In any language this number is at most the size of the alphabet, but there may be multiple letters which induce the same transformation;
this does not occur in this paper as no language has a repeated transformation.
For binary operations on two dialects of the same master language, we count the number of distinct  transformations of the master language present in either dialect.
For example, suppose $L(a,b,c,-)$ and $L(a,-,b,c)$ are two dialects of a language $L(a,b,c,d)$, which we assume has four distinct transformations.
Each dialect has three letters and three distinct transformations, and between them they have three letters and four distinct transformations.

Although a given complexity bound may be met by many dialects of the master language, we favour dialects, or pairs of dialects, that use small alphabets and few distinct transformations.
In many cases the dialects we present are minimal in these respects, though we do not always justify this.

\section{A Most Complex Regular Stream}
\label{sec:regular}

We now define a DFA stream that we use as a basic component. It is similar to the stream defined in~\cite{Brz13} for the case of equal alphabets, except that there the transformation induced by $c$ is $(n-1\to 0)$.
It is also similar to the DFA of~\cite{Brz16}, except that there 
the transformation induced by $c$ is $(n-1\to 0)$ and an additional input $d$ is used.

\begin{definition}
\label{def:regular}
For $n\ge 3$, let $\mathcal{D}_n=\mathcal{D}_n(a,b,c)=(Q_n,\Sigma,\delta_n, 0, \{n-1\})$, where 
$\Sigma=\{a,b,c\}$, 
and $\delta_n$ is defined by the transformations
$a\colon (0,\dots,n-1)$,
$b\colon(0,1)$, and
$c\colon(1 \rightarrow 0)$.
Let $L_n=L_n(a,b,c)$ be the language accepted by~$\mathcal{D}_n$.
The structure of $\mathcal{D}_n(a,b,c)$ is shown in Figure~\ref{fig:RegWit}. 
\end{definition}

\begin{figure}[ht]
\unitlength 8.5pt
\begin{center}\begin{picture}(37,7)(0,3)
\gasset{Nh=1.8,Nw=3.5,Nmr=1.25,ELdist=0.4,loopdiam=1.5}
	{\scriptsize
\node(0)(1,7){0}\imark(0)
\node(1)(8,7){1}
\node(2)(15,7){2}
}
\node[Nframe=n](3dots)(22,7){$\dots$}
	{\scriptsize
\node(n-2)(29,7){$n-2$}
	}
	{\scriptsize
\node(n-1)(36,7){$n-1$}\rmark(n-1)
	}
\drawloop(0){$c$}
\drawedge[curvedepth= .8,ELdist=.1](0,1){$a,b$}
\drawedge[curvedepth= .8,ELdist=-1.2](1,0){$b,c$}
\drawedge(1,2){$a$}
\drawloop(2){$b,c$}
\drawedge(2,3dots){$a$}
\drawedge(3dots,n-2){$a$}
\drawloop(n-2){$b,c$}
\drawedge(n-2,n-1){$a$}
\drawedge[curvedepth= 4.0,ELdist=-1.0](n-1,0){$a$}
\drawloop(n-1){$b,c$}
\end{picture}\end{center}
\caption{Minimal DFA  of  a most complex regular language.}
\label{fig:RegWit}
\end{figure}
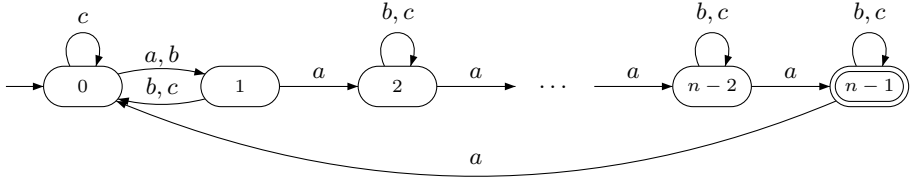

\begin{theorem}[Most Complex Regular Languages]
\label{thm:regular}
For each $n\ge 3$, the DFA of Definition~\ref{def:regular} is minimal and its 
language $L_n(a,b,c)$ has complexity $n$.
The stream $(L_m(a,b,c) \mid m \ge 3)$  with 
some dialect streams is most complex in the class of regular languages.
In particular, it meets all the complexity bounds below.
At least three letters are required in any witness meeting all these bounds and a total of four distinct letters is required for any two witnesses for  unrestricted  union and symmetric difference.
In several cases the bounds can be met with a smaller alphabet as shown below.
\begin{enumerate}
\item
The syntactic semigroup of $L_n(a,b,c)$ has cardinality $n^n$.  
\item
Each quotient of $L_n(a)$ has complexity $n$.
\item
The reverse of $L_n(a,b,c)$ has complexity $2^n$, and $L_n(a,b,c)$ has $2^n$ atoms.
\item
Each atom  $A_S$  of $L_n(a,b,c)$ has maximal complexity:
\begin{equation*}
	\kappa(A_S) =
	\begin{cases}
		2^n-1, 			& \text{if $S\in \{\emptyset,Q_n\}$;}\\
		1+ \sum_{x=1}^{|S|}\sum_{y=1}^{n-|S|} \binom{n}{x}\binom{n-x}{y},
		 			& \text{if $\emptyset \subsetneq S \subsetneq Q_n$.}
		\end{cases}
\end{equation*}
\item
The star of $L_n(a,b)$ has complexity $2^{n-1}+2^{n-2}$.
\item
\begin{enumerate}
\item 
	Restricted Complexity:\\
	The product $L'_m(a,b) L_n(a,-,b)$ has complexity $m2^n-2^{n-1}$.
\item
	Unrestricted Complexity:\\
	The product $L'_m(a,b) L_n(a,c,b)$ has complexity $m2^n+2^{n-1}$.
\end{enumerate}

\item
\begin{enumerate}
\item
	Restricted Complexity:\\
The complexity of $L'_m(a,b) \circ L_n(b,a)$ is $mn$ for $\circ\in \{\cup,\oplus, \setminus, \cap\}$.
\item
	Unrestricted Complexity:\\
The complexity of union and symmetric difference is $mn+m+n+1$ and this bound is met by 
$L'_m(a,b,c)$ and $L_n(b,a,d)$, 
that of difference is  $mn+m$ and this bound is met by 
$L'_m(a,b,c)$ and $L_n(b,a)$, 
and that of intersection is 
$mn$ and this bound is met by  
$L'_m(a,b)$ and $L_n(b,a)$.
 A total of four letters is required to meet the bounds for union and symmetric difference.
\end{enumerate}
\end{enumerate}
\end{theorem}
\begin{proof} Clearly $L_n(a)$ has complexity $n$ as the DFA of Definition~\ref{def:regular} is minimal.
\begin{enumerate}
\item {\bf Syntactic Semigroup} The transformations 
$a\colon (0, \dots, n-1)$, $b\colon (0,1)$, and $c\colon (n-1 \to 0)$ were used in~\cite{Brz13}. It is well known that these transformations as well as $a$, $b$, and $c \colon (1\to 0)$ generate the semigroup of all transformations of $Q_n$. 
\item {\bf Quotients} Obvious.
\item {\bf Reversal} This follows from a theorem in~\cite{SWY04} which states that if the transition semigroup has $n^n$ elements, then the complexity of reversal is $2^n$. It was shown in~\cite{BrTa14} that the number of atoms is the same as the complexity of the reverse.
\item {\bf Atoms} Proved in~\cite{BrDa14}.
\item {\bf Star} Proved in~\cite{Brz13}.
\item {\bf Product}
Let $\mathcal{D}'= (Q'_m,\Sigma', \delta', 0', F')$ and $\mathcal{D}= (Q_n,\Sigma, \delta, 0, F)$ be minimal DFAs of languages $L'$ and $L$, respectively.
We use the standard  construction of the $\varepsilon$-NFA $\mathcal{N}$ for the product $L'L$:
the final states of $\mathcal{D}'$ becomes non-final, and  an $\varepsilon$-transition is added from each state of $F'$ to the initial state $0$ of $\mathcal{D}$.

The subset construction on this NFA yields sets $\{p'\} \cup S$ where $p' \in Q'_m \setminus F'$ and $S \subseteq Q_n$ and sets $\{p',0\} \cup S$ where $p' \in F'$ and $S \subseteq Q_n\setminus \{0\}$,  as well as sets $S \subseteq Q_n$ which can only be reached by letters in $\Sigma\setminus \Sigma'$. 
Hence the restricted complexity of $L'L$ is bounded by $(m-|F'|)2^n + |F'|2^{n-1} \le m2^n -2^{n-1}$,
and the unrestricted complexity of $L'L$ is bounded by $(m-|F'|)2^n + |F'| 2^{n-1} + 2^n \le m2^n + 2^{n-1}$. 

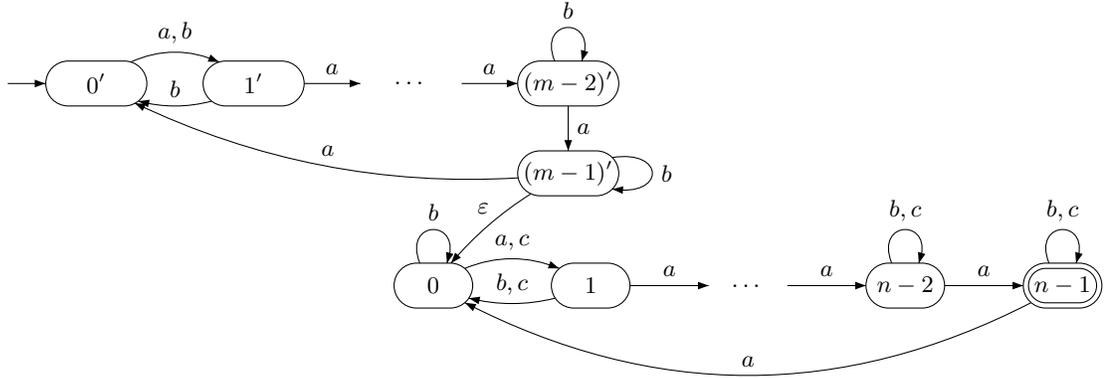
\begin{figure}[th]
\unitlength 8.5pt
\begin{center}\begin{picture}(38,15)(-3,2.8)
\gasset{Nh=2.0,Nw=4.5,Nmr=1.25,ELdist=0.4,loopdiam=1.5}
\node(0')(-4,15){$0'$}\imark(0')
\node(1')(3,15){$1'$}
\node[Nframe=n](3dots')(10,15){$\dots$}
\node(m-2')(17,15){$(m-2)'$}
\drawloop(m-2'){$b$}
\node(m-1')(17,11){$(m-1)'$}
\drawloop[loopangle=0](m-1'){$b$}

\drawedge[curvedepth= 1.4](0',1'){$a,b$}
\drawedge[curvedepth= 1,ELside=r](1',0'){$b$}
\drawedge(1',3dots'){$a$}
\drawedge(3dots',m-2'){$a$}
\drawedge(m-2',m-1'){$a$}
\drawedge[curvedepth=1.7, ELside=r](m-1',0'){$a$}

\gasset{Nh=2.0,Nw=3.5,Nmr=1.25,ELdist=0.4,loopdiam=1.5}
\node(0)(11,6){0}
\drawloop(0){$b$}
\node(1)(18,6){1}
\node[Nframe=n](3dots)(25,6){$\dots$}
	{\small
\node(n-2)(32,6){$n-2$}
\drawloop(n-2){$b,c$}
\node(n-1)(39,6){$n-1$}\rmark(n-1)
\drawloop(n-1){$b,c$}
	}
\drawedge[curvedepth= 1.2](0,1){$a,c$}
\drawedge[curvedepth= .8,ELside=r](1,0){$b,c$}
\drawedge(1,3dots){$a$}
\drawedge(3dots,n-2){$a$}
\drawedge[curvedepth=4, ELside=r](n-1,0){$a$}
\drawedge(n-2,n-1){$a$}
\drawedge[ELside=r, curvedepth=-0.5](m-1',0){$\varepsilon$}
\end{picture}\end{center}
\caption{An NFA for the product of $L'_m(a,b)$ and $L_n(a,c,b)$. The NFA for the product of $L'_m(a,b)$ and $L_n(a,-,b)$ is the same except $c$ is omitted.}
\label{fig:RegularProduct}
\end{figure}
\emph{Restricted Complexity:}
Consider $L'_m(a,b)$ and $L_n(a,-,b)$ of Definition~\ref{def:regular}; we show that their product meets the upper bound for restricted complexity.
As before, we construct an NFA recognizing $L'_m(a,b)L_n(a,-,b)$ and then apply the subset construction to obtain a DFA.
Figure~\ref{fig:RegularProduct} shows the NFA for the unrestricted product $L'_m(a,b)L_n(a,c,b)$;
the product $L'_m(a,b)L_n(a,-,b)$ is the same except $c$ is omitted.

The initial state is $\{0'\}$ and each state $\{p'\}$ for $0 \le p \le m-2$ is reached by $a^p$.
Consider $\{0'\} \cup S$, where $S = \{q_1, q_2, \dots, q_k\}$ with $0 \le q_1 <q_2 < \dots < q_k \le n-1$.
If $q_1 \ge 1$ then 
$\{(m-2)', q_2 -q_1 -1, \dots , q_k-q_1-1\} \xrightarrow{a^2} \{0', 1, q_2 -q_1+1, \dots , q_k-q_1+1\} \xrightarrow{(ab)^{q_1-1}}\{0'\} \cup S$.
If $q_1=0$ and $k \ge 2$, then
$\{(m-2)', n-2, q_3-q_2-1, \dots , q_k-q_2-1\} \xrightarrow{a^2} \{0', 0, 1, q_3-q_2+1, \dots , q_k-q_2+1\} \xrightarrow{(ab)^{q_2-1}} \{0'\} \cup S$.
State $\{0', 0\}$ is reached by $a^mb^2$.
Hence for any non-empty $S \subseteq Q_n$, state $\{0'\} \cup S$ is reachable from $\{(m-2)'\} \cup T$ for some $T \subseteq Q_n$ of size $|S|-1$.
We reach $\{p'\} \cup S$ from $\{0'\} \cup (S-p)$ by $a^p$, where $S-p$ denotes $\{q-p \mid q \in S\}$ taken mod $n$.
By induction, $\{p'\} \cup S$ is always reachable and thus all $m2^n-2^{n-1}$ states are reachable.

We check that all states are pairwise distinguishable.
\begin{enumerate}
\item Any two sets which differ by $q \in Q_n$ are distinguished by $a^{n-1-q}$.
\item States $\{p'_1\}$ and $\{p'_2\}$ with $p'_1 < p'_2$ are distinguished by $a^{m-1-p_2}a^{n-1}$.
\item States $\{0', 0\}$ and $\{p', 0\}$ are distinguished by $(ab)^{m-2-p}aa^{n-1}$ if $p' \not= (m-1)'$; otherwise 
 apply $ab$ to simplify to this case.
\item States $\{p'_1, 0\}$ and $\{p'_2, 0\}$ reduce to Case (c) by $(ab^2)^{m-p_2}$.
\item States $\{p'_1\} \cup S$ and $\{p'_2\} \cup S$, where $S \not= \emptyset$, reduce to Case (d) by $(ab)^n$ since $S \xrightarrow{(ab)^n} \{0\}$ and $(ab)^n$ permutes $Q'_m$.
\end{enumerate}
We can distinguish any pair of states; hence the complexity of $L'_m(a,b)L_n(a,-,b)$ is $m2^n-2^{n-1}$ for all $m,n \ge 3$.

\medskip

\emph{Unrestricted Complexity:}
The NFA for the product of $L'_m(a,b)L_n(a,c,b)$ is illustrated in Figure~\ref{fig:RegularProduct}.
The NFA is the same as the restricted case except it has the additional transformation $c \colon (0,1)(Q'_m \to \emptyset)$.
Hence the subset construction yields the $m2^n - 2^{n-1}$ sets of the restricted case,
as well as all sets $S \subseteq Q_n$ since $S$ is reachable from $\{0'\} \cup S$ by $c^2$.
We check that these sets are distinguishable from all previously reached sets.
\begin{enumerate}
\item Any two sets which differ by $q \in Q_n$ are distinguished by $a^{n-1-q}$.
\item State $\{p'\}$ is distinguishable from $\emptyset$ by $a^{m-1-p}a^{n-1}$.
\item States $\{0', 0\}$ and $\{0\}$ are distinguished by $a^{m-1}a^{n-1}$ if $m-1$ is not a multiple of $n$, and by $ba^{m-2}a^{n-1}$ otherwise.
\item States $\{p', 0\}$ and $\{0\}$ reduce to Case (c)  by $(ab^2)^{m-p}$.
\item States $\{p'\} \cup S$ and $S$, where $S \not= \emptyset$, reduce to Case (d) by $(ab)^n$ since $S \xrightarrow{(ab)^n} \{0\}$ and $(ab)^n$ permutes $Q'_m$.
\end{enumerate}
Hence $L'_m(a,b)L_n(a,c,b)$ has complexity $m2^n + 2^{n-1}$.

\item {\bf Boolean Operations}

\emph{Restricted Complexity:}
All boolean operations have complexity at most $mn$~\cite{Brz10}.
Applying the standard construction for boolean operations we consider the direct product of $\mathcal{D}'_m(a,b)$ and $\mathcal{D}_n(b,a)$ which has states $Q'_m \times Q_n$;
the final states vary depending on the operation.
By~\cite[Theorem 1]{BBMR14} and computation for the cases $(m,n) \in \{(3,4),(4,3),(4,4)\}$, the states of $Q'_m \times Q_n$ are reachable and pairwise distinguishable for each operation $\circ \in \{\cup, \oplus, \setminus, \cap\}$; hence each operation has complexity $mn$.

Note that two letters are required to meet these bounds:
 To a contradiction suppose a single letter $\ell$ is sufficient to reach $Q'_m\times Q_n$ in the direct product, where $m, n \ge 2$ are not coprime.
Letter $\ell$ must induce an $m$-element permutation on $Q'_m$; otherwise there is an unreachable state in $Q'_m$ or the sequence $0', 0'\ell, 0'\ell^2, \dots, 0'\ell^k, \dots$ never returns to $0'$. Similarly $\ell$ must induce an $n$-cycle in $Q_n$.
Hence $\ell$ has order $\lcm(mn)$ in the direct product.
It must have order $mn$ in the direct product, which only occurs when $m$ and $n$ are coprime. 

\medskip

\emph{Unrestricted Complexity:}
The upper bounds on the unrestricted complexity of boolean operations are derived in~\cite{Brz16}.
To compute $L'_m(a,b,c) \circ L_n(b,a,d)$, where $\circ$ is a boolean operation, add an empty state $\emptyset'$ to $\mathcal{D}'_m(a,b,c)$, and send all the transitions from any state of $Q'_m$ under $d$ to $\emptyset'$.
Similarly, add an empty state  $\emptyset$ to  $\mathcal{D}_n(b,a,d)$ together with appropriate transitions; now the alphabets of the resulting DFAs are the same.
We consider the direct product of $\mathcal{D}'_{m,\emptyset'}$ and $\mathcal{D}_{n,\emptyset}$ which has states $\{(p',q) \mid p' \in Q'_m \cup \{\emptyset'\}, q \in Q_n \cup \{\emptyset\}\}$.
A DFA recognizing $L'_m(a,b,c) \cup L_n(b,a,d)$ is shown in Figure~\ref{fig:regularcross} for $m=3$ and $n=4$.

\begin{figure}[th]
\unitlength 8pt
\begin{center}\begin{picture}(35,21)(0,-3)
\gasset{Nh=2.6,Nw=2.6,Nmr=1.2,ELdist=0.3,loopdiam=1.2}
	{\scriptsize
\node(0'0)(2,15){$0',0$}\imark(0'0)
\node(1'0)(2,10){$1',0$}
\node(2'0)(2,5){$2',0$}\rmark(2'0)
\node(3'0)(2,0){$\emptyset',0$}

\node(0'1)(10,15){$0',1$}
\node(1'1)(10,10){$1',1$}
\node(2'1)(10,5){$2',1$}\rmark(2'1)
\node(3'1)(10,0){$\emptyset',1$}

\node(0'2)(18,15){$0',2$}
\node(1'2)(18,10){$1',2$}
\node(2'2)(18,5){$2',2$}\rmark(2'2)
\node(3'2)(18,0){$\emptyset',2$}

\node(0'3)(26,15){$0',3$}\rmark(0'3)
\node(1'3)(26,10){$1',3$}\rmark(1'3)
\node(2'3)(26,5){$2',3$}\rmark(2'3)
\node(3'3)(26,0){$\emptyset',3$}\rmark(3'3)

\node(0'4)(34,15){$0',\emptyset$}
\node(1'4)(34,10){$1',\emptyset$}
\node(2'4)(34,5){$2',\emptyset$}\rmark(2'4)
\node(3'4)(34,0){$\emptyset',\emptyset$}
	}

\drawedge(0'4,1'4){$a$}
\drawedge(1'4,2'4){$a$}
\drawedge[curvedepth=-3](2'4,0'4){$a$}

\drawedge(2'0,3'0){$d$}
\drawedge[ELside=r](2'1,3'0){$d$}
\drawedge(2'2,3'2){$d$}
\drawedge(2'3,3'3){$d$}
\drawedge(2'4,3'4){$d$}

\drawedge(3'0,3'1){$b$}
\drawedge(3'1,3'2){$b$}
\drawedge(3'2,3'3){$b$}
\drawedge[curvedepth=3](3'3,3'0){$b$}

\drawedge(0'3,0'4){$c$}
\drawedge[ELside=r](1'3,0'4){$c$}
\drawedge(2'3,2'4){$c$}
\drawedge(3'3,3'4){$c$}

\end{picture}\end{center}
\caption{Direct product for union of $\mathcal{D}'_3(a,b,c)$ and $\mathcal{D}_4(b,a,d)$ of Definition~\ref{def:regular} shown partially.}
\label{fig:regularcross}
\end{figure}
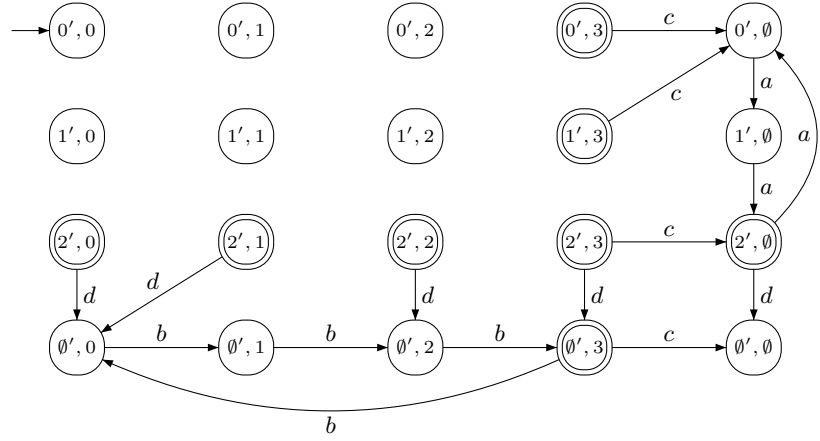

As in the restricted case all the states of $Q'_m \times Q_n$ are reachable by words in $\{a,b\}^*$.
Since $a$ and $b$ induce permutations on the direct product, it follows that every state in $Q'_m \times Q_n$ is reachable from every other.
The remaining states in $C = \{(p', \emptyset) \mid p' \in Q'_m \cup \{\emptyset'\}\}$ and $R=\{(\emptyset',q) \mid q \in Q_n \cup \{\emptyset\}\}$ are easily reachable using $c$ and $d$ in addition to $a$ and $b$. Hence all $(m+1)(n+1)$ states are reachable.

For union and symmetric difference, the states of $C$ are pairwise distinguishable by words in $a^*$ and
they are distinguished from all other states by words in $b^*d$.
Similarly the states of $R$ are distinguishable from each other and all other states; hence all $mn+m+n+1$ states are distinguishable.

For difference, the final states are $((m-1)', q)$ for $q \not= n-1$.
The states of $R$ are all empty, and they are only reachable by $d$.
As the words of $L'_m(a,b,c) \setminus L_n(b,a,d)$ do not contain $d$, the alphabet is $\{a,b,c\}$;
hence we can omit $d$ and delete the states of $R$, and be left with a DFA recognizing the same language.
We check that the remaining $mn+m$ states are pairwise distinguishable.
Any states $(p'_1, \emptyset)$ and $(p'_2,q)$ where $p'_1 \not= p'_2$ and $q \in Q_n \cup \{\emptyset\}$ are distinguished by words in $a^*$.
State $(p', \emptyset)$ is distinguished from $(p',q)$ by some $w \in \{a,b\}^*$ that maps $(p',q)$ to $((m-1)', n-1)$, since $w$ must send $(p',\emptyset)$ to the final state $((m-1)', \emptyset)$.

For intersection the only final state is $((m-1)', n-1)$.
The alphabet of $L'_m(a,b,c) \cap L_n(b,a,d)$ is $\{a,b\}$; hence we can omit $c$ and $d$ and delete the states of $R \cup C$, and be left with a DFA recognizing the same language.
The remaining $mn$ states are pairwise distinguishable as in the restricted case.

Note that a total of four letters between the alphabets $\Sigma'$ of $\mathcal{D}'_m$ and $\Sigma$ of $\mathcal{D}_n$ is required for union and symmetric difference.
As in the restricted case, two letters in $\Sigma' \cap \Sigma$ are required to reach the states of $Q'_m \times Q_n$ for general values of $m$ and $n$.
Letters in both alphabets cannot be used to reach states $(p', \emptyset)$ and $(\emptyset', q)$ as the empty states in each coordinate are only reached by letters outside the corresponding alphabet. Thus two additional letters are required, one in $\Sigma'\setminus \Sigma$ and one in $\Sigma\setminus \Sigma'$. 
Hence each alphabet must contain at least three letters, and $\Sigma' \cup \Sigma$ must contain at least four.
In contrast, the bound for difference is met by $L'_m(a,b,c)$ and $L_n(b,a)$,
and the bound for intersection is met by $L'_m(a,b)$ and $L_n(b,a)$. 
\end{enumerate}
\end{proof}

\section{Right Ideals}

The results in this section are based on~\cite{BrDa15,BDL15,BrYe11}; however, the stream  below is different from  that of~\cite{BrYe11}, where 
$c\colon (n-2\to 0)$ and $d\colon (n-2 \to n-1)$. 
\begin{definition}
\label{def:ideal}
For $n\ge 4$, let $\mathcal D_n=\mathcal D_n(a,b,c,d)=(Q_n,\Sigma,\delta_n, 0, \{n-1\})$, where
$\Sigma=\{a,b,c,d\}$
and $\delta_n$ is defined by the transformations
$a\colon (0,\dots,n-2)$,
$b\colon(0,1)$,
$c\colon(1 \rightarrow 0)$, and
$d\colon (_0^{n-2} q\to q+1)$.
Let $L_n=L_n(a,b,c,d)$ be the language accepted by~$\mathcal D_n$.
For the structure of $\mathcal D_n(a,b,c,d)$ see Figure~\ref{fig:RIdealWit}.
\end{definition}

\begin{figure}[ht]
\unitlength 8.5pt
\begin{center}\begin{picture}(37,8)(0,-4)
\gasset{Nh=2.0,Nw=4.1,Nmr=1.25,ELdist=0.4,loopdiam=1.5}
	{\scriptsize
\node(0)(1,-2){0}\imark(0)
\node(1)(8,-2){1}
\node(2)(15,-2){2}
\node[Nframe=n](3dots)(22,-2){$\dots$}
\node(n-2)(29,-2){$n-2$}
\node(n-1)(36,-2){$n-1$}\rmark(n-1)
	}
\drawloop(0){$c$}
\drawedge[curvedepth= .8,ELdist=.1](0,1){$a,b,d$}
\drawedge[curvedepth= .8,ELdist=-1.2](1,0){$b,c$}
\drawedge(1,2){$a,d$}
\drawloop(2){$b,c$}
\drawedge(2,3dots){$a,d$}
\drawedge(3dots,n-2){$a,d$}
\drawloop(n-2){$b,c$}
\drawedge(n-2,n-1){$d$}
\drawedge[curvedepth= 3.6,ELdist=-1.0](n-2,0){$a$}
\drawloop(n-1){$a,b,c,d$}
\end{picture}\end{center}
\caption{Minimal DFA of a most complex right ideal.}
\label{fig:RIdealWit}
\end{figure}
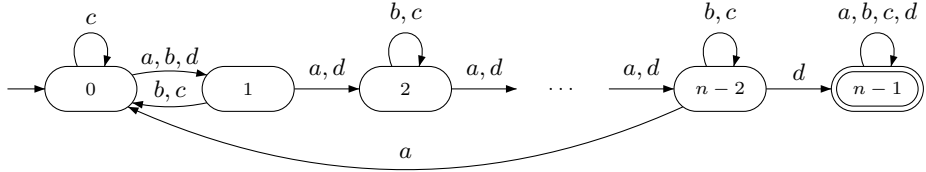

\begin{theorem}[Most Complex Right Ideals]
\label{thm:rightideals}
For each $n\ge 4$, the DFA of Definition~\ref{def:ideal} is minimal and  $L_n(a,b,c,d)$ is a right ideal of complexity $n$.
The stream $(L_m(a,b,c,d) \mid m \ge 4)$  with 
\ some dialect streams 
is most complex in the class of  right ideals.
In particular, it meets all the  bounds  below.
At least  four letters are required to meet these bounds.
\begin{enumerate}
\item
The syntactic semigroup of $L_n(a,b,c,d)$ has cardinality $n^{n-1}$.  
There is only one maximal transition semigroup of a minimal DFA accepting a right ideal, since it consists of all the transformations of $Q_n$ that fix $n-1$. 
At least four letters are needed for this bound.
\item
The quotients of $L_n(a,-,-,d)$ have complexity $n$, except that  $\kappa(\{a, d\}^*)=1$.
\item
The reverse of $L_n(a,-,-,d)$ has complexity $2^{n-1}$, and $L_n(a,-,-,d)$ has $2^{n-1}$ atoms.
\item
Each atom $A_S$\ of $L_n(a,b,c,d)$ has maximal complexity:
\begin{equation*}
	\kappa(A_S) =
	\begin{cases}
		2^{n-1}, 			& \text{if $S=Q_n$;}\\
		1 + \sum_{x=1}^{|S|}\sum_{y=1}^{n-|S|}\binom{n-1}{x-1}\binom{n-x}{y},
		 			& \text{if $\emptyset \subsetneq S \subsetneq Q_n$.}
		\end{cases}
\end{equation*}
\item
The star of $L_n(a,-,-,d)$ has complexity $n+1$.
\item
\begin{enumerate}
\item 
	Restricted Complexity:\\
The product $L'_m(a,-,c,d) L_n(a,-,c,d)$ has complexity $m+2^{n-2}$.
\item
	Unrestricted Complexity:\\
The product $L'_m(a,-,c,d) L_n(b,-,c,d)$ has complexity $m+2^{n-1}+2^{n-2}+1$.  At least three letters for each language and four letters in total are required to meet this bound.\end{enumerate}
\item
\begin{enumerate}
\item
	Restricted Complexity:\\
The complexity of $\circ$
is $mn$ if $\circ\in \{\cap,\oplus\}$, 
 $mn-(m-1)$ if $\circ=\setminus$, and $mn-(m+n-2)$ if $\circ=\cup$, and 
these bounds are met by $L'_m(a,-,-,d) \circ L_n(-,-,d,a)$.
At least two letters are required to meet these bounds.
\item
	Unrestricted Complexity:\\
The complexity of $L'_m(a,-,c,d) \circ L_n(b,-,d,a)$ is the same as  for arbitrary regular languages:
$mn+m+n+1$ if $\circ\in \{\cup,\oplus\}$, 
$mn+m$ if $\circ=\setminus$, and $mn$ if $\circ=\cap$.
At least three letters in each language and four letters in total are required to meet the bounds for intersection and symmetric difference.
The bound for difference is also met by $L'_m(a,-,c,d) \setminus L_n(-,-,d,a)$ and the bound for intersection is met by $L'_m(a,-,-,d) \cap L_n(-,-,d,a)$.

\end{enumerate}
\end{enumerate}
\label{thm:ideal}
\end{theorem}

\begin{proof}
DFA $\mathcal{D}_n(-,-,-,d)$ is minimal because  the shortest word in $d^*$ accepted by state $q$ is $d^{n-1-q}$, and
$L_n(a,b,c,d)$ is a right ideal because it has only one final state and that state accepts $\Sigma^*$.
\begin{enumerate}
\item {\bf Semigroup}
The transformations induced by $a$, $b$, and $c$ generate all transformations of $Q_{n-1}$. Also, since the transformation induced by $da^{n-2}$ is $(n-2 \to n-1)$, the transition semigroup of 
$\mathcal{D}_n(a,b,c,d)$ contains the one in~\cite{BrYe11}, which is maximal for right ideals. Hence the syntactic semigroup of $L_n(a,b,c,d)$ has size $n^{n-1}$ as well. The fact that at least four letters are needed was proved in~\cite{BSY15}.
\item {\bf Quotients}
If the initial state of $\mathcal{D}_n(a,-,-,d)$ is changed to $q$ with $0\le q\le n-2$, the new DFA accepts a quotient of $L_n$ and is still minimal; hence the complexity of that quotient is $n$.
\item {\bf Reversal}
It was proved in~\cite{BJL13} that the reverse has complexity $2^{n-1}$, and  in~\cite{BrTa14} that the number of atoms is the same as the complexity of the reverse.
\item {\bf Atoms}
The proof in~\cite{BrDa15} applies since the DFA has all the transformations that fix $n-1$.

\item {\bf Star}
If $L_n$ is a right ideal, then $L_n^*=L_n\cup \{\varepsilon\}$. If we add a new initial state $0'$ to the DFA of Definition~\ref{def:ideal} with the same transitions as those from $0$ and make $0'$ final, the new DFA accepts $L_n^*$ and is minimal -- $0'$ does not accept $a$, and so is not equivalent to $n-1$.

\item {\bf Product} 
Let $\mathcal{D}'= (Q'_m,\Sigma', \delta', 0', \{(m-1)'\})$ and $\mathcal{D}= (Q_n,\Sigma, \delta, 0, \{n-1\})$ be minimal DFAs of $L'$ and $L$, respectively, 
where $L'$ and $L$ are right ideals.
We use the standard  construction of the NFA for the product $L'L$:
the final state $(m-1)'$ of $\mathcal{D}'$ becomes non-final, and an $\varepsilon$-transition is added from that state to the initial state $0$ of $\mathcal{D}$.
We bound the complexity of the product by counting the reachable states in the subset construction on this NFA.
The $m-1$ non-final states $\{p'\}$ of $\mathcal{D}'$ may be reachable, as well as $\{(m-1)',0\}$.
From  $\{(m-1)',0\}$ we may reach all $2^{n-2}$ subsets of $Q_n$ which contain 0 but not $n-1$, and $2^{n-2}$ states that contain both 0 and $n-1$; however, the latter $2^{n-2}$ states all accept $\Sigma^*$ and are therefore equivalent. 
So far, we have $m-1+2^{n-2}+1=m+2^{n-2}$ states; these are the only reachable sets if the witnesses are restricted to the same alphabet.

For the unrestricted case, suppose that $\ell'\in \Sigma'\setminus \Sigma$ and $\ell \in \Sigma \setminus \Sigma'$.
By applying $\ell$ to $\{(m-1)',0\} \cup S$, $S\subseteq Q_n\setminus \{0\}$, we may reach 
all $2^n-1$ non-empty subsets of $Q_n$, and then by applying $\ell'$ we reach the empty subset. 
However, the $2^{n-1}$ subsets of $Q_n$ that contain $n-1$ all accept $\Sigma^*$. 
Hence there are at most $2^{n-1} +1$ additional sets, for a total of $m+2^{n-2}+2^{n-1}+1$ reachable sets.

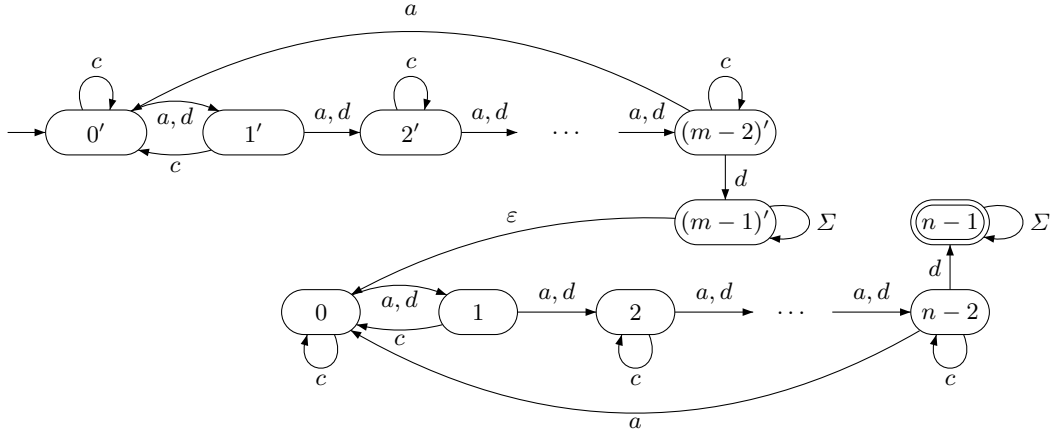
\begin{figure}[th]
\unitlength 8.5pt
\begin{center}\begin{picture}(38,17)(-4,2.5)
\gasset{Nh=2.0,Nw=4.5,Nmr=1.25,ELdist=0.4,loopdiam=1.5}
\node(0')(-4,14){$0'$}\imark(0')
\drawloop(0'){$c$}
\node(1')(3,14){$1'$}
\node(2')(10,14){$2'$}
\drawloop(2'){$c$}
\node[Nframe=n](3dots')(17,14){$\dots$}
\node(m-2')(24,14){$(m-2)'$}
\drawloop(m-2'){$c$}
\node(m-1')(24,10){$(m-1)'$}

\drawloop[loopangle=360,ELdist=.3](m-1'){$\Sigma$}
\drawedge[curvedepth= 1.4,ELdist=-1.2](0',1'){$a,d$}
\drawedge[curvedepth= 1,ELdist=.3](1',0'){$c$}
\drawedge(1',2'){$a,d$}
\drawedge(2',3dots'){$a,d$}
\drawedge(3dots',m-2'){$a,d$}
\drawedge[curvedepth= -4.5,ELdist=-1.2](m-2',0'){$a$}
\drawedge(m-2',m-1'){$d$}

\gasset{Nh=2.0,Nw=3.5,Nmr=1.25,ELdist=0.4,loopdiam=1.5}
\node(0)(6,6){0}
\drawloop[loopangle=270,ELdist=.3](0){$c$}
\node(1)(13,6){1}
\node(2)(20,6){2}
\drawloop[loopangle=270,ELdist=.3](2){$c$}
\node[Nframe=n](3dots)(27,6){$\dots$}
	{\small
\node(n-2)(34,6){$n-2$}
\drawloop[loopangle=270,ELdist=.3](n-2){$c$}
\node(n-1)(34,10){$n-1$}\rmark(n-1)
\drawloop[loopangle=360,ELdist=.3](n-1){$\Sigma$}
	}
\drawedge[curvedepth= 1.2,ELdist=-1.2](0,1){$a,d$}
\drawedge[curvedepth= .8,ELdist=.2](1,0){$c$}
\drawedge(1,2){$a,d$}
\drawedge(2,3dots){$a,d$}
\drawedge(3dots,n-2){$a,d$}
\drawedge[curvedepth= 4.5,ELdist=0.2](n-2,0){$a$}
\drawedge(n-2,n-1){$d$}
\drawedge[curvedepth= -1.6,ELdist=-1](m-1',0){$\varepsilon$}
\end{picture}\end{center}
\caption{An NFA  for product of right ideal  $L'_m(a,-,c,d)$ and its dialect $L_n(a,-,c,d)$.}
\label{fig:RIdealProduct1}
\end{figure}

\medskip
 
\emph{Restricted Complexity:} 
To prove the bound is tight, consider the two dialects of the DFA of Definition~\ref{def:ideal} shown in Figure~\ref{fig:RIdealProduct1}.
The $m-1$ sets $\{p'\}$ for $p' \in Q'_{m-1}$ are reachable in $\mathcal{D}_m'$ by words in $d^*$, and $\{(m-1)',0\}$ is reached by $d^{m-1}$. 
The $2^{n-2}$ sets of the form 
$\{(m-1)',0\} \cup S$, where $S\subseteq Q_n\setminus \{0\}$, are reachable using words in $\{c,d\}^*$ as follows:
To reach $\{(m-1)',0\} \cup S$, where $S= \{q_1,\dots,q_k\}$, $1\le q_1<q_2< \dots < q_k \le n-1$, we have first
$\{(m-1)',0\}d=\{(m-1)',0,1\}$. State 1 will then be moved to the right by applying either $d$ or $dc$ repeatedly: If $q_{k-1}=q_k-1$, use $d$; otherwise use $dc$.
Repeating this process $q_k$ times we eventually construct $S$. For example, to reach 
$\{(m-1)',0\} \cup \{2,5,7,8\}$ use $dd(dc)d(dc)(dc)d(dc)$.
The $2^{n-2}$ sets $\{(m-1)', 0\} \cup S$ that contain $n-1$ all accept $\{a,c,d\}^*$; hence they are all equivalent.

The remaining states are pairwise distinguishable:
States $\{p'\}$ and $\{q'\}$ with $0\le p<q \le m-2$ are distinguished by $d^{m-1-q}d^{n-1}$, and 
 $\{p'\}$ is distinguished from $\{(m-1)',0\} \cup S$  by $d^{n-1}$.
Two non-final states $\{(m-1)',0\} \cup S$ and
$\{(m-1)',0\} \cup T$ with $q \in S\oplus T$ are distinguished by $a^{n-2-q}d$.
Thus the product has complexity $m+2^{n-2}$.

\begin{figure}[th]
\unitlength 8.5pt
\begin{center}\begin{picture}(38,17)(-4,2.5)
\gasset{Nh=2.0,Nw=4.5,Nmr=1.25,ELdist=0.4,loopdiam=1.5}
\node(0')(-4,14){$0'$}\imark(0')
\drawloop(0'){$c$}
\node(1')(3,14){$1'$}
\node(2')(10,14){$2'$}
\drawloop(2'){$c$}
\node[Nframe=n](3dots')(17,14){$\dots$}
\node(m-2')(24,14){$(m-2)'$}
\drawloop(m-2'){$c$}
\node(m-1')(24,10){$(m-1)'$}

\drawloop[loopangle=360,ELdist=.3](m-1'){$\Sigma'$}
\drawedge[curvedepth= 1.4,ELdist=-1.2](0',1'){$a,d$}
\drawedge[curvedepth= 1,ELdist=.3](1',0'){$c$}
\drawedge(1',2'){$a,d$}
\drawedge(2',3dots'){$a,d$}
\drawedge(3dots',m-2'){$a,d$}
\drawedge[curvedepth= -4.5,ELdist=-1.2](m-2',0'){$a$}
\drawedge(m-2',m-1'){$d$}

\gasset{Nh=2.0,Nw=3.5,Nmr=1.25,ELdist=0.4,loopdiam=1.5}
\node(0)(6,6){0}
\drawloop[loopangle=270,ELdist=.3](0){$c$}
\node(1)(13,6){1}
\node(2)(20,6){2}
\drawloop[loopangle=270,ELdist=.3](2){$c$}
\node[Nframe=n](3dots)(27,6){$\dots$}
	{\small
\node(n-2)(34,6){$n-2$}
\drawloop[loopangle=270,ELdist=.3](n-2){$c$}
\node(n-1)(34,10){$n-1$}\rmark(n-1)
\drawloop[loopangle=360,ELdist=.3](n-1){$\Sigma$}
	}
\drawedge[curvedepth= 1.2,ELdist=-1.2](0,1){$b,d$}
\drawedge[curvedepth= .8,ELdist=.2](1,0){$c$}
\drawedge(1,2){$b,d$}
\drawedge(2,3dots){$b,d$}
\drawedge(3dots,n-2){$b,d$}
\drawedge[curvedepth= 4.5,ELdist=0.2](n-2,0){$b$}
\drawedge(n-2,n-1){$d$}
\drawedge[curvedepth= -1.6,ELdist=-1](m-1',0){$\varepsilon$}
\end{picture}\end{center}
\caption{An NFA  for product of right ideal  $L'_m(a,-,c,d)$ and its dialect $L_n(b,-,c,d)$.}
\label{fig:RIdealProduct2}
\end{figure}
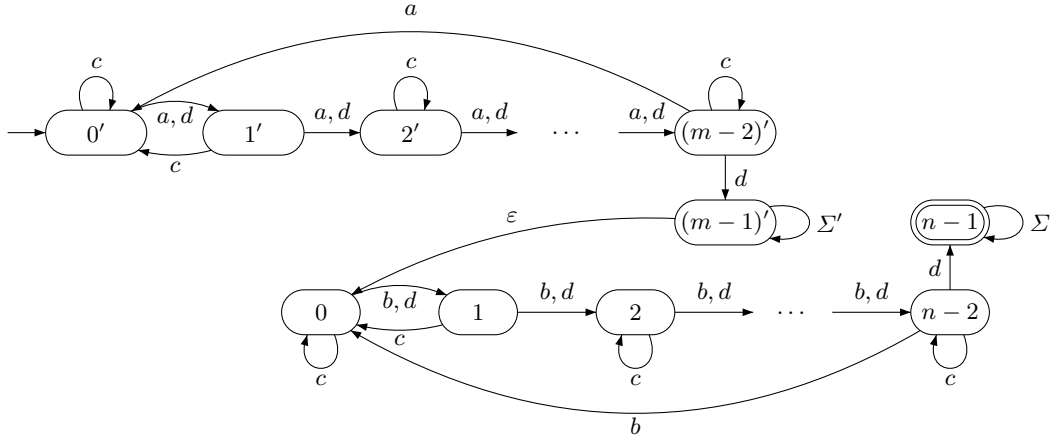

\emph{Unrestricted Complexity:}
Consider two dialects of the DFA of Definition~\ref{def:ideal} shown in Figure~\ref{fig:RIdealProduct2}.
Here $\Sigma' = \{a, c,d\}$ and $\Sigma = \{b,c,d\}$.
By the restricted case, all states $\{p'\}$ for $p' \in Q'_{m-1}$ and $\{(m-1)', 0\} \cup S$ for $S \subseteq Q_n \setminus \{0\}$ are reachable by words in $\{c,d\}^*$.
Apply $b$ from $\{0'\}$ to reach the empty subset.
By applying $b$ to $\{(m-1)',0\} \cup S$, $S\subseteq Q_n\setminus \{0\}$, we reach all $2^n-1$ non-empty subsets of $Q_n$;
hence all states are reachable.
However, the $2^{n-1}$ sets $S \subseteq Q_n$ that contain $n-1$ all accept $\{b,c,d\}^*$ and are sent to the empty state by $a$; hence they are all equivalent.
Similarly, the $2^{n-2}$ sets $\{(m-1)', 0\} \cup S$ that contain $n-1$ all accept $\{b,c,d\}^*$ and are sent to $\{(m-1)', 0\}$ by $a$; hence they are also equivalent.

The remaining states are pairwise distinguishable.
States $\{p'\}$ and $\{q'\}$ with $0\le p<q \le m-2$ are distinguished by $d^{m-1-q}d^{n-1}$, and 
 $\{p'\}$ is distinguished from $\{(m-1)',0\} \cup S$  by $d^{n-1}$
 or from $S$, where $\emptyset \subsetneq S\subseteq Q_n$ by $d^{n-1}$.
Two states $\{(m-1)',0\} \cup S$ and
$\{(m-1)',0\} \cup T$ with $q\in S\oplus T$ are distinguished by $b^{n-2-q}d$.
Two states $S$ and $T$ with $q\in S\oplus T$ are distinguished by $b^{n-2-q}d$.
A state $\{(m-1)',0\} \cup S$ is distinguishable from $T$ where $S,T \subseteq Q_n$ by $ad^{n-1}$.
Thus all $m+2^{n-2}+2^{n-1}+1$ states are pairwise distinguishable.

At least three inputs to each DFA are required to achieve the bound in the unrestricted case:
There must be a letter in $\Sigma$ (like $d$) with a transition to $n-1$ to reach sets containing $n-1$, and this letter must be in $\Sigma'$ in order to reach the sets that contain both $(m-1)'$ and $n-1$. However no single letter in $\Sigma' \cap \Sigma$ is sufficient to reach every set of the form $\{(m-1)',0\} \cup S$, regardless of its behaviour on $Q_n$. For example, if the letter maps $0 \to q_1$ and $q_1 \to q_2$ then it is impossible to reach the state $\{(m-1)',0,q_2\}$ by repeatedly applying the letter from $\{(m-1)',0\}$, as it can never delete $q_1$.
Hence there must be at least two letters in $\Sigma' \cap \Sigma$.
Furthermore there must be some $\ell \in \Sigma \setminus \Sigma'$ to reach the empty state, and there must be some $\ell' \in \Sigma' \setminus \Sigma$ to distinguish $\{(m-1)', 0, n-1\}$ from $\{n-1\}$. Thus each alphabet must contain at least three letters to meet the bound.

\item {\bf Boolean Operations}

\emph{Restricted Complexity:}
The bounds for right ideals were derived in~\cite{BJL13}.
We show that the DFAs  $\mathcal{D}'_m(a,-,-,d)$ and $\mathcal{D}_n(-,-,d,a)$ of the right ideals of Definition~\ref{def:ideal} meet the bounds.

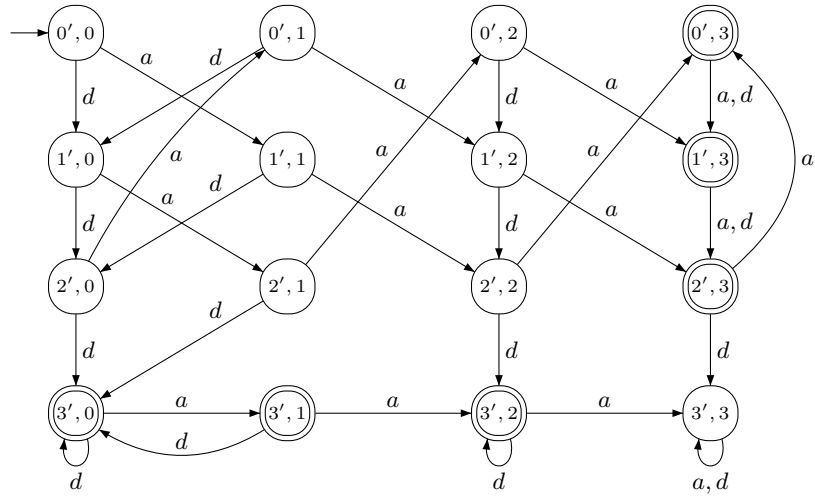
\begin{figure}[th]
\unitlength 8pt
\begin{center}\begin{picture}(35,22)(0,-2)
\gasset{Nh=2.6,Nw=2.6,Nmr=1.2,ELdist=0.3,loopdiam=1.2}
	{\scriptsize
\node(0'0)(2,18){$0',0$}\imark(0'0)
\node(1'0)(2,12){$1',0$}
\node(2'0)(2,6){$2',0$}
\node(3'0)(2,0){$3',0$}\rmark(3'0)

\node(0'1)(12,18){$0',1$}
\node(1'1)(12,12){$1',1$}
\node(2'1)(12,6){$2',1$}
\node(3'1)(12,0){$3',1$}\rmark(3'1)

\node(0'2)(22,18){$0',2$}
\node(1'2)(22,12){$1',2$}
\node(2'2)(22,6){$2',2$}
\node(3'2)(22,0){$3',2$}\rmark(3'2)

\node(0'3)(32,18){$0',3$}\rmark(0'3)
\node(1'3)(32,12){$1',3$}\rmark(1'3)
\node(2'3)(32,6){$2',3$}\rmark(2'3)
\node(3'3)(32,0){$3',3$}
	}
	
\drawedge(0'0,1'0){$d$}
\drawedge[ELpos=30, ELside=r](0'1,1'0){$d$}
\drawedge(0'2,1'2){$d$}
\drawedge(0'3,1'3){$a,d$}

\drawedge(1'0,2'0){$d$}
\drawedge[ELpos=30, ELside=r](1'1,2'0){$d$}
\drawedge(1'2,2'2){$d$}
\drawedge(1'3,2'3){$a,d$}
\drawedge[curvedepth=-4,ELside=r](2'3,0'3){$a$}

\drawedge[curvedepth= 1,ELside=r](2'0,0'1){$a$}
\drawedge(2'0,3'0){$d$}
\drawedge[ELpos=30, ELside=r](2'1,3'0){$d$}
\drawedge(2'2,3'2){$d$}
\drawedge(2'3,3'3){$d$}

\drawedge(3'0,3'1){$a$}
\drawedge(3'1,3'2){$a$}
\drawedge(3'2,3'3){$a$}

\drawedge[ELpos=30](0'0,1'1){$a$}
\drawedge[ELpos=40](1'0,2'1){$a$}
\drawedge(0'1,1'2){$a$}
\drawedge(1'1,2'2){$a$}
\drawedge(2'1,0'2){$a$}
\drawedge(0'2,1'3){$a$}
\drawedge(1'2,2'3){$a$}
\drawedge(2'2,0'3){$a$}

\drawloop[loopangle=270,ELdist=.3](3'0){$d$}
\drawedge[curvedepth= 2,ELside=r](3'1,3'0){$d$}
\drawloop[loopangle=270,ELdist=.3](3'2){$d$}
\drawloop[loopangle=270,ELdist=.3](3'3){$a,d$} 

\end{picture}\end{center}
\caption{The direct product for the symmetric difference of right ideals $L'_4(a,-,-,d)$ and $L_4(-,-,d,a)$.}
\label{fig:idealcross1}
\end{figure}

Consider the direct product of $L'_m(a,-,-,d)$ and $L_n(-,-,d,a)$, illustrated in Figure~\ref{fig:idealcross1} for $m=n=4$.
For $p' \in Q'_{m-1}$ state $(p',0)$ is reached by $d^p$.
Since the first column of $Q_{m-1} \times Q_n$ is reachable and $(p',q) \xrightarrow{a} ((p+1)',(q+1))$, where $p+1$ is taken mod $m-1$,
we can reach every state in $Q_{m-1} \times Q_n$.
State $((m-1)', q)$ is reached by $d^{m-1}a^q$; hence the states of $Q'_m \times Q_n$ are reachable.

We now check distinguishability, which depends on the final states of the DFA.
The direct product is made to recognize $L'_m(a,-,-,d) \circ L_n(-,-,d,a)$ by setting the final states to be $(\{(m-1)'\} \times Q_n) \circ (Q'_m \times \{n-1\})$.

For intersection and symmetric difference, all states are pairwise distinguishable. States that differ in the first coordinate are distinguished by words in $d^*a^*$ and states that differ in the second coordinate are distinguished by words in $a^*d^*$. Hence the complexity is $mn$.

For difference, the states $\{(p',n-1) \mid p' \in Q'_m\}$ are all empty, and therefore equivalent.
The remaining states are non-empty, and they are distinguished by words in $d^*$ if they differ in the first coordinate or by words in $a^*d^*$ if they differ in the second coordinate.
Hence the complexity is $mn-m+1$.

For union, the states $\{(p',n-1) \mid p' \in Q'_m\} \cup \{((m-1)', q) \mid q \in Q_n\}$ are all final and equivalent as they accept $\{a,d\}^*$. The remaining states are distinguished by words in $d^*$ if they differ in the first coordinate or by words in $a^*$ if they differ in the second coordinate.
Hence the complexity is $mn-(m+n-2)$.

As in regular languages, one letter in $\Sigma' \cap \Sigma$ is not sufficient to reach all the states of $Q'_{m-1} \times Q_{n-1}$  for all values of $m$ and $n$;
hence two letters are required to meet any of the bounds.

\medskip

\emph{Unrestricted Complexity:}
The unrestricted bounds for right ideals are the same as those for arbitrary regular languages~\cite{Brz16}.
We show that the DFAs  $\mathcal{D}'_m(a,-,c,d)$ and $\mathcal{D}_n(b,-,d,a)$ of Definition~\ref{def:ideal} meet the bounds.

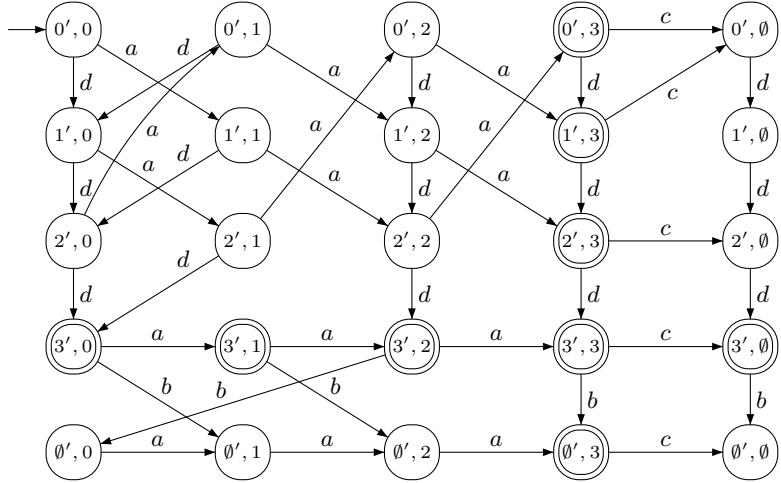
\begin{figure}[th]
\unitlength 8pt
\begin{center}\begin{picture}(35,23)(0,-6)
\gasset{Nh=2.6,Nw=2.6,Nmr=1.2,ELdist=0.3,loopdiam=1.2}
	{\scriptsize
\node(0'0)(2,15){$0',0$}\imark(0'0)
\node(1'0)(2,10){$1',0$}
\node(2'0)(2,5){$2',0$}
\node(3'0)(2,0){$3',0$}\rmark(3'0)
\node(4'0)(2,-5){$\emptyset',0$}

\node(0'1)(10,15){$0',1$}
\node(1'1)(10,10){$1',1$}
\node(2'1)(10,5){$2',1$}
\node(3'1)(10,0){$3',1$}\rmark(3'1)
\node(4'1)(10,-5){$\emptyset',1$}

\node(0'2)(18,15){$0',2$}
\node(1'2)(18,10){$1',2$}
\node(2'2)(18,5){$2',2$}
\node(3'2)(18,0){$3',2$}\rmark(3'2)
\node(4'2)(18,-5){$\emptyset',2$}

\node(0'3)(26,15){$0',3$}\rmark(0'3)
\node(1'3)(26,10){$1',3$}\rmark(1'3)
\node(2'3)(26,5){$2',3$}\rmark(2'3)
\node(3'3)(26,0){$3',3$}\rmark(3'3)
\node(4'3)(26,-5){$\emptyset',3$}\rmark(4'3)

\node(0'4)(34,15){$0',\emptyset$}
\node(1'4)(34,10){$1',\emptyset$}
\node(2'4)(34,5){$2',\emptyset$}
\node(3'4)(34,0){$3',\emptyset$}\rmark(3'4)
\node(4'4)(34,-5){$\emptyset',\emptyset$}
	}
	
\drawedge(0'0,1'0){$d$}
\drawedge[ELpos=30, ELside=r](0'1,1'0){$d$}
\drawedge(0'2,1'2){$d$}
\drawedge(0'3,1'3){$d$}

\drawedge(1'0,2'0){$d$}
\drawedge[ELpos=30, ELside=r](1'1,2'0){$d$}
\drawedge(1'2,2'2){$d$}
\drawedge(1'3,2'3){$d$}

\drawedge[curvedepth= 1,ELside=r](2'0,0'1){$a$}
\drawedge(2'0,3'0){$d$}
\drawedge[ELpos=30, ELside=r](2'1,3'0){$d$}
\drawedge(2'2,3'2){$d$}
\drawedge(2'3,3'3){$d$}

\drawedge(0'4,1'4){$d$}
\drawedge(1'4,2'4){$d$}
\drawedge(2'4,3'4){$d$}
\drawedge(3'4,4'4){$b$}

\drawedge(3'0,3'1){$a$}
\drawedge(3'1,3'2){$a$}
\drawedge(3'2,3'3){$a$}

\drawedge[ELpos=30](0'0,1'1){$a$}
\drawedge[ELpos=40](1'0,2'1){$a$}
\drawedge(0'1,1'2){$a$}
\drawedge(1'1,2'2){$a$}
\drawedge(2'1,0'2){$a$}
\drawedge(0'2,1'3){$a$}
\drawedge(1'2,2'3){$a$}
\drawedge(2'2,0'3){$a$}

\drawedge(4'0,4'1){$a$}
\drawedge(4'1,4'2){$a$}
\drawedge(4'2,4'3){$a$}
\drawedge(4'3,4'4){$c$}

\drawedge(0'3,0'4){$c$}
\drawedge[ELside=r](1'3,0'4){$c$}
\drawedge(2'3,2'4){$c$}
\drawedge(3'3,3'4){$c$}

\drawedge(3'0,4'1){$b$}
\drawedge(3'1,4'2){$b$}
\drawedge[ELside=r, ELpos=55](3'2,4'0){$b$}
\drawedge(3'3,4'3){$b$}
\end{picture}\end{center}
\caption{Partial illustration of the direct product for  $L'_4(a,-,c,d) \cup L_4(b,-,d,a)$.}
\label{fig:idealcross2}
\end{figure}

To compute $L'_m(a,-,c,d) \circ L_n(b,-,d,a)$, where $\circ$ is a boolean operation, add an empty state $\emptyset'$ to $\mathcal{D}'_m(a,-,c,d)$, and send all the transitions from any state of $Q'_m$ under $b$ to $\emptyset'$.
Similarly, add an empty state  $\emptyset$ to  $\mathcal{D}_n(b,a,d)$  together with appropriate transitions; now the alphabets of the resulting DFAs are the same.
The direct product of $L'_m(a,-,c,d)$ and $L_n(b,-,d,a)$ is illustrated in Figure~\ref{fig:idealcross2} for $m=n=4$.

As in the restricted case, the $mn$ states of $Q'_m \times Q_n$ are reachable by words in $\{a,d\}^*$.
The remaining states $(p', \emptyset)$ and $(\emptyset', q)$ are easily seen to be reachable using $b$ and $c$.

We now check distinguishability, which depends on the final states of the DFA.
The direct product is made to recognize $L'_m(a,-,c,d) \circ L_n(b,-,d,a)$ by setting the final states to be $(\{(m-1)'\} \times Q_n \cup \{\emptyset\}) \circ (Q'_m \cup \{\emptyset'\}\times \{n-1\})$.

For union and symmetric difference, all states are pairwise distinguishable: States that differ in the first coordinate are distinguished by words in $d^*c$ and states that differ in the second coordinate are distinguished by words in $a^*b$.

For difference, the final states are $((m-1)', q)$ for $q \not= n-1$.
The alphabet of $L'_m(a,-,c,d) \setminus L_n(b,-,a,d)$ is $\{a,c,d\}$; hence we can omit $b$ and delete all states $(\emptyset', q)$ and be left with a DFA recognizing the same language.
The remaining states are distinguished by words in $d^*c$ if they differ in the first coordinate or by words in $a^* d^*$ if they differ in the second coordinate.

For intersection, the only final state is $((m-1)', n-1)$.
The alphabet of $L'_m(a,-,c,d) \cap L_n(b,-,d,a)$ is $\{a,b\}$; hence we can omit $b$ and $c$ and delete all states $(p', \emptyset)$ and $(\emptyset', q)$.
The remaining $mn$ states are pairwise distinguishable as in the restricted case.

Note that the bound for difference is met by $L'_m(a,-,c,d) \setminus L_n(-,-,d,a)$, and that of intersection is met by $L'_m(a,-,-,d) \cap L_n(-,-,d,a)$.
However the bounds for union and symmetric difference all require three letters in each dialect:
There must be a letter in $\Sigma' \setminus \Sigma$ to reach states of the form $(p', \emptyset)$, and 
there must a letter in $\Sigma \setminus \Sigma'$ to reach states of the form $(\emptyset', q)$.  As in regular languages, one letter in $\Sigma' \cap \Sigma$ is not sufficient to reach all the states of $Q'_m \times Q_n$  for all values of $m$ and $n$; hence $|\Sigma' \cap \Sigma| \ge 2$ and so both $\Sigma'$ and $\Sigma$ must contain at least three letters.
\end{enumerate}

It has been shown in~\cite{BJL13} that at least two letters are needed for each right ideal that meets the bounds for star or reversal. Hence  almost all our witnesses in Theorem~\ref{thm:rightideals} that meet the bounds for the common operations use minimal alphabets.

\end{proof}

\section{Prefix-Closed Languages}
The complexity of operations on prefix-closed languages was studied in~\cite{BJZ14}, but most complex prefix-closed languages were not considered.
As every prefix-closed language has an empty quotient, the restricted and unrestricted complexities are the same for binary operations.
\begin{definition}
\label{def:ClosedWit}
For $n\ge 4$, let $\mathcal{D}_n=\mathcal{D}_n(a,b,c,d)=(Q_n,\Sigma,\delta_n, 0, Q_n\setminus \{n-1\})$, 
where 
$\Sigma=\{a,b,c,d\}$, 
and $\delta_n$ is defined by the transformations $a\colon (0,\dots,n-2)$,
$b\colon(0,1)$,
${c\colon(1\rightarrow 0)}$, and
$d\colon \left(_{n-2}^{0} \; q\to q-1 \pmod n \right)$.
Let $L_n=L_n(a,b,c,d)$ be the language accepted by~$\mathcal{D}_n$.
The structure of $\mathcal{D}_n(a,b,c,d)$ is shown in Figure~\ref{fig:ClosedWit}. 
\end{definition}

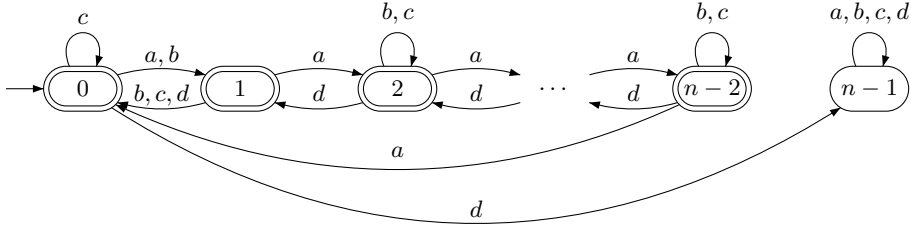
\begin{figure}[th]
\unitlength 8.5pt
\begin{center}\begin{picture}(37,8)(0,3)
\gasset{Nh=2.0,Nw=3.5,Nmr=1.25,ELdist=0.4,loopdiam=1.5}
\node(0)(1,7){0}\imark(0)\rmark(0)
\node(1)(8,7){1}\rmark(1)
\node(2)(15,7){2}\rmark(2)
\node[Nframe=n](3dots)(22,7){$\dots$}
	{\small
\node(n-2)(29,7){$n-2$}\rmark(n-2)
	}
	{\small
\node(n-1)(36,7){$n-1$}
	}
\drawedge[curvedepth= .9,ELdist=.1](0,1){$a,b$}
\drawedge[curvedepth= -6,ELdist=.3](0,n-1){$d$}
\drawedge[curvedepth= .9,ELdist=-1.1](1,0){$b,c,d$}
\drawedge[curvedepth= .9,ELdist=.3](1,2){$a$}
\drawedge[curvedepth= .9,ELdist=-1.1](2,1){$d$}
\drawedge[curvedepth= .9,ELdist=.3](2,3dots){$a$}
\drawedge[curvedepth= .9,ELdist=-1.1](3dots,2){$d$}
\drawedge[curvedepth= .9,ELdist=.3](3dots,n-2){$a$}
\drawedge[curvedepth= .9,ELdist=-1.1](n-2,3dots){$d$}
\drawedge[curvedepth= 3.6,ELdist=-1.0](n-2,0){$a$}

\drawloop(0){$c$}
\drawloop(2){$b,c$}
\drawloop(n-2){$b,c$}
\drawloop(n-1){$a,b,c,d$}
\end{picture}\end{center}
\caption{DFA of a most complex prefix-closed language.}
\label{fig:ClosedWit}
\end{figure}

\begin{theorem}[Most Complex Prefix-Closed Languages]
\label{thm:closed}
For each $n\ge 4$, the DFA of Definition~\ref{def:ClosedWit} is minimal and  $L_n(a,b,c,d)$ is a prefix-closed language of complexity $n$.
The stream $(L_m(a,b,c,d) \mid m \ge 4)$ with some dialect streams
is most complex in the class of prefix-closed languages.
At least four letters are required to meet the bounds below.
\begin{enumerate}
\item
The syntactic semigroup of $L_n(a,b,c,d)$ has cardinality $n^{n-1}$.  
\item
The quotients of $L_n(a,-,-,d)$ have complexity $n$, except for $\emptyset$, which has complexity 1.
\item
The reverse of $L_n(a,-,-,d)$ has complexity $2^{n-1}$, and $L_n(a,-,-,d)$ has $2^{n-1}$ atoms.
\item
Each atom $A_S$  of $L_n(a,b,c,d)$ has maximal complexity:
\begin{equation*}
\kappa(A_S) =
\begin{cases}
2^{n-1}, & \text{if $S=\emptyset$;}\\
1 + \sum_{x=1}^{n-|S|}\sum_{y=1}^{|S|}\binom{n-1}{x-1}\binom{n-x}{y},
& \text{if $\emptyset \subsetneq S \subsetneq Q_n$.}
\end{cases}
\end{equation*}

\item
The star of  $L_n(a,-,c,d)$  has complexity $2^{n-2}+1$.
\item
The product $L'_m(a,b,c,d) L_n(a,d,b,c)$ has complexity $(m+1)2^{n-2}$.
\item
For any proper binary boolean function $\circ$, the complexity of $L'_m(a,b,-,d) \circ L_n(b,a,-,d)$
is maximal. In particular, the complexity is $mn$ if $\circ\in \{\cup,\oplus\}$, 
 $mn-(n-1)$ if $\circ=\setminus$, and $mn-(m+n-2)$ if $\circ=\cap$.
\end{enumerate}

\end{theorem}
\begin{proof}
The DFA is minimal since state $p$ rejects $d^{q}$ if and only if $p < q$. It is prefix-closed because all non-empty states are final. 
\begin{enumerate}
\item {\bf Semigroup}
Let $d'$ induce the transformation $(_0^{n-2} q\to q+1)$ (this was called $d$ in the right-ideal section).
Since $ada =d'$, the transition semigroup of the DFA of Figure~\ref{fig:ClosedWit} is the same as that of the DFA of the right ideal of Figure~\ref{fig:RIdealWit}.

\item {\bf Quotients}
Obvious.

\item {\bf Reversal}
Since reversal commutes with complementation, we consider the complement of the language accepted by the DFA of Figure~\ref{fig:ClosedWit} restricted to the alphabet $\{a,d\}$. 
It was proved in~\cite{BJL13} that the reverse of a right ideal has complexity at most $2^{n-1}$, and  in~\cite{BrTa14} that the number of atoms is the same as the complexity of the reverse.
It remains to prove that all $2^{n-1}$ states of the DFA $\mathcal{D}^R$ obtained by the subset construction from the NFA 
$\mathcal{N}$ obtained by reversal of the DFA of the right ideal $\mathcal{D}$
are reachable and distinguishable. The proof is similar to that of~\cite{BJL13}. 
 Subset $\{n-1\}$ is the initial state of $\mathcal{N}$, and $n-1$ appears in every reachable state of $\mathcal{D}^R$.
 Every subset $\{n-1,q_{2},q_{3} \ldots,q_{k}\}$ of size $k$, 
 where $1\le k\le n-2$ and
 $0\le q_2<q_3<\dots <q_k\le n-2$,
 is reached from the subset $\{n-1,q_3-(q_2+1),\ldots,q_k-(q_2+1)\}$ of size $k-1$
 by $da^{n-(q_2+1)}$.
Since only state $q$, $0\le q\le n-2$, accepts $a^q$, any two subsets differing by $q$ are distinguishable by $a^q$.

\item {\bf Atoms}
Let $L$  be a prefix-closed language with quotients $K_0, \dots, K_{n-1}$, $n \ge 4$.
Recall that $\overline{L}$ is a right ideal with quotients $\overline{K_0}, \dots, \overline{K_{n-1}}$.
For  $S \subseteq \{0, \dots, n-1\}$, the atom of $L$ corresponding to $S$ is
$A_S = \bigcap_{i \in S} K_i \cap \bigcap_{i \in \overline{S}} \overline{K_i}$. This can be rewritten as $\bigcap_{i \in \overline{S}} \overline{K_i} \cap \bigcap_{i \in \overline{\overline{S}}} \overline{\overline{K_i}}$, which is the atom of $\overline{L}$ corresponding to $\overline{S}$; hence the sets of atoms of $L$ and $\overline{L}$ are the same.
The claim follows from the theorem for right ideals. The proof in~\cite{BrDa15} applies since the DFA that accepts the complement of the prefix-closed language 
of Figure~\ref{fig:ClosedWit} 
has all the transformations that fix $n-1$.

\item {\bf Star}
It was proved in~\cite{BJZ14} that $2^{n-2}+1$ is the maximal complexity of the star of a prefix-closed language. 
We now show that $L_n(a,-,c,d)$ meets this bound.
Since $L_n(a,-,c,d)$ accepts $\varepsilon$, no new initial state is needed and it suffices to delete the empty state and add an $\varepsilon$-transition from each final state to the initial state to get an NFA $\mathcal{N}$ for $L_n^*$. 
In this NFA all  $2^{n-2}$ subsets of $Q_{n-1}$ containing 0 are reachable and pairwise distinguishable. 
Any non-empty set $\{0, q_2, q_3, \dots, q_k\}$ of size $k$ with $0 < q_2 < q_3 < \dots < q_k \le n-2$ is reached from $\{0, q_3-q_2, \dots, q_k-q_2\}$ of size $k-1$ by $a(ac)^{q_2-1}$.
Moreover, the empty set is reached from $\{0\}$ by $d$, giving the required bound.
Sets  $\{0\} \cup S$ and $\{0\} \cup T$ with $q\in S\oplus T$ are distinguished by $a^{n- 2-q}d^{n- 2}$.

\item {\bf Product}
It was shown in~\cite{BJZ14} that the complexity of the product of prefix-closed languages is $(m+1)2^{n-2}$.
We now prove that our witness $L_m'(a,b,c,d)$ with minimal DFA $\mathcal{D}'_m(a,b,c,d)$ together with the dialect $L_n(a,d,b,c)$ 
with minimal DFA $\mathcal{D}_n(a,d,b,c)$ meets this bound. 
Construct the following NFA $\mathcal{N}$ for the product. Start with $\mathcal{D}'_m(a,b,c,d)$, but make all of its states non-final.
Delete the empty state from $\mathcal{D}_n(a,d,b,c)$ and all the transitions to the empty state, add an $\varepsilon$-transition from each state $p'\in Q'_{m-1}$ to the initial state $0$ of $\mathcal{D}_n(a,d,b,c)$.
We will show that $(m-1)2^{n-2}$ states of the form $\{p',0\}\cup S$, where $S\subseteq Q_{n-1}\setminus\{0\}$, and $2^{n-1}$ states of the form $\{(m-1)'\} \cup S$, where $S\subseteq Q_{n-1}$ are reachable and pairwise distinguishable.

The initial state of the subset automaton of $\mathcal{N}$ is $\{0',0\}$.
State $\{1', 0\}$ is reachable by $b$ and $\{p', 0\}$ for $2 \le p \le m-2$ is reachable from $\{1', 0\}$ by $(ab)^{p-1}$.
State $\{p', 0\} \cup S$ where $p' \in Q_{m-1}'$ and $S = \{q_1, \dots, q_k\}$ is reachable from $\{r', 0, q_2 - q_1, \dots, q_k -q_1\}$ by $a(ab)^{q_1-1}$ for some $r \in Q_{m-1}'$.
By induction, all $(m-1)2^{n-2}$ states $\{p', 0\} \cup S$ are reachable.
From $\{0', 0\} \cup S$ by $d^2$ we reach $\{ (m-1)', 0\} \cup S$. Further apply $ca$ to reach $\{(m-1)'\} \cup S$.
Hence all $2^{n-1}$ subsets of the form $\{(m-1)'\} \cup S$ are reachable.

We check that the states are pairwise distinguishable in four cases.
\begin{enumerate}
\item $\{(m-1)'\} \cup S$ and $\{(m-1)'\} \cup T$ with $r \in S \oplus T$ are distinguished by $a^{n-2-r}c^{n-2}$.
\item $\{p'\} \cup S$ and $\{p'\} \cup T$ with $r \in S \oplus T$ reduces to Case (a) by $a^{n-2-r}d^m$.
\item $\{p'\} \cup S$ and $\{(m-1)'\} \cup T$ with $p \in Q_{m-1}'$ are distinguished by $c^n$.
\item $\{p'\} \cup S$ and $\{q'\} \cup T$ with $p < q  \le m-2$ reduces to Case (c) by $d^{p+1}$.
\end{enumerate}

\item {\bf Boolean Operations}
It is again convenient to consider the ideal languages defined  by the complements of the prefix-closed languages of Figure~\ref{fig:ClosedWit} restricted to the alphabet $\{a,b,d\}$ and then use De Morgan's laws. Since every prefix-closed language has an empty quotient, it is sufficient to consider boolean operations on languages over the same alphabet. 
The problems are the same as those in~\cite{BDL15}, except that there the transformation induced by $d$ is $d:(n-2\to n-1)$.

Let $\mathcal{D}_n(a,b,c,d)$ denote the DFA for the complement  of the prefix-closed language of Definition~\ref{def:ClosedWit} of complexity $n$ and let $L_n$ be the language accepted by $\mathcal{D}_n$. We consider boolean operations on right ideals $L'_m$ and $L_n$.

\begin{figure}[ht]
\unitlength 8.5pt
\begin{center}\begin{picture}(35,20)(0,-2)
\gasset{Nh=2.6,Nw=2.6,Nmr=1.2,ELdist=0.3,loopdiam=1.2}
	{\scriptsize
\node(4'4)(2,15){$4',4$}
\node(0'4)(2,10){$0',4$}\rmark(0'4)
\node(1'4)(2,5){$1',4$}\rmark(1'4)
\node(2'4)(2,0){$2',4$}\rmark(2'4)

\node(4'0)(10,15){$4',0$}\rmark(4'0)
\node(0'0)(10,10){$0',0$}\imark(0'0)
\node(1'0)(10,5){$1',0$}
\node(2'0)(10,0){$2',0$}

\node(4'1)(18,15){$4',1$}\rmark(4'1)
\node(0'1)(18,10){$0',1$}
\node(1'1)(18,5){$1',1$}
\node(2'1)(18,0){$2',1$}

\node(4'2)(26,15){$4',2$}\rmark(4'2)
\node(0'2)(26,10){$0',2$}
\node(1'2)(26,5){$1',2$}
\node(2'2)(26,0){$2',2$}

\node(4'3)(34,15){$4',3$}\rmark(4'3)
\node(0'3)(34,10){$0',3$}
\node(1'3)(34,5){$1',3$}
\node(2'3)(34,0){$2',3$}
	}

\drawedge(0'0,4'4){$d$}
\drawedge(0'1,4'0){$d$}
\drawedge(0'2,4'1){$d$}
\drawedge(0'3,4'2){$d$}

\drawedge(1'0,0'4){$d$}
\drawedge(2'0,1'4){$d$}

\drawedge(4'0,4'4){$d$}
\drawedge(0'4,4'4){$d$}

\drawedge[curvedepth=1,ELdist=0.2](4'3,4'2){$d$}
\drawedge[curvedepth=1,ELdist=0.2](4'2,4'1){$d$}
\drawedge[curvedepth=1,ELdist=0.2](4'1,4'0){$d$}

\drawedge[curvedepth=1,ELdist=0.2](4'2,4'3){$b$}
\drawedge[curvedepth=1,ELdist=0.2](4'1,4'2){$b$}
\drawedge[curvedepth=1,ELdist=0.2](4'0,4'1){$b$}
\drawedge[curvedepth=-4,ELdist=0.5](4'3,4'0){$b$}

\drawedge[curvedepth=-1,ELdist=-0.7](2'4,1'4){$d$}
\drawedge[curvedepth=-1,ELdist=-0.7](1'4,0'4){$d$}

\drawedge[curvedepth=-1,ELdist=-0.7](1'4,2'4){$a$}
\drawedge[curvedepth=-1,ELdist=-0.7](0'4,1'4){$a$}
\drawedge[curvedepth=3,ELdist=0.5](2'4,0'4){$a$}

\drawedge(1'3,2'3){$a$}
\drawedge(0'3,1'3){$a$}
\drawedge[curvedepth=-3,ELdist=0.5](2'3,0'3){$a$}

\drawedge(2'2,2'3){$b$}
\drawedge(2'1,2'2){$b$}
\drawedge(2'0,2'1){$b$}
\drawedge[curvedepth=2,ELdist=0.5](2'3,2'0){$b$}
\end{picture}\end{center}
\caption{Direct product for symmetric difference of right ideals with DFAs $\mathcal{D}'_4(a,b,-,d)$ and $\mathcal{D}_5(b,a,-,d)$ shown partially.}
\label{fig:closedcross}
\end{figure}

The direct product is illustrated in Figure~\ref{fig:closedcross}.
The states in $Q'_{m-1} \times Q_{n-1}$ are reachable from the initial state $(0',0)$ by~\cite[Theorem 1]{BBMR14}.
Then $((m-1)', 0)$ is reached from $(0',1)$ by $d$ and states of the form $((m-1)', q)$, $0\le q\le n-2$, are then reached by words in $b^*$.
Similarly,  $(0',n-1)$ is reached from $(1',0)$ by $d$ and states of the form $(p',n-1)$, $0\le p\le m-2$, are then reached by words in $a^*$. Finally, $((m-1)',n-1)$ is reached from $((m-1)',0)$ by $d$. Hence all $mn$ states are reachable.

Let $S = Q'_{m-1} \times Q_{n-1}$, $R = \{(m-1)'\} \times Q_n$, and $C = Q'_m \times \{n-1\}$.
The final states of the direct product to recognize $L'_m(a,b,-,d) \circ L_n(b,a,-,d)$ are $R \circ C$.

Consider the following DFAs: $D'_{m-1}(a,b) = (Q'_{m-1},\{a,b\}, \delta, 0', \{0'\})$ and $D_{n-1}(b,a) = (Q_{n-1},$ 
$\{a,b\}, \delta, 0, \{0\})$.
By~\cite[Theorem 1]{BBMR14}, the states of $S$ are pairwise distinguishable with respect to the states
$\left(\{0'\} \times Q_{n-1}\right) \circ \left(Q'_{m-1} \times \{0\}\right)$ for any $\circ \in \{\cup, \oplus, \setminus, \cap\}$.
One can verify that if $w$ distinguished two states of $S$ with respect to $\left(\{0'\} \times Q_{n-1}\right) \circ \left(Q'_{m-1} \times \{0\}\right)$, then $wd$ distinguishes them with respect to $R \circ C$ for each $\circ = \{\cup, \oplus, \setminus, \cap\}$. The rest of the argument depends on the operation $\circ \in \{\cup, \oplus, \setminus, \cap\}$.

$\cap, \oplus$: All $mn$ states are pairwise distinguishable. The states of $R$ are distinguished by words in $d^*$. The states of $C$ are similarly distinguishable. The states of $R$ are distinguished from the states of $C$ by words in $\{a,d\}^*$. Every state of $S$ is sent to a state of $R$ by a word in $\{a,d\}^*$, and similarly to a state of $C$ by a word in $\{b,d\}^*$; thus the states of $S$ are distinguishable from the states of $R$ or $C$.

$\setminus$: The states of $C$ are all empty, leaving $m(n-1) + 1$ distinguishable states. The states of $R$ are distinguished by words in $d^*$.

$\cup$: The states of $R$ and $C$ are equivalent final states accepting all words, leaving $(m-1)(n-1) + 1$ distinguishable states.

By De Morgan's laws we have $\kappa(L'_m \cup L_n) =\kappa(\overline{L'_m} \cap \overline{L_n})$, $\kappa(L'_m \oplus L_n) =\kappa(\overline{L'_m} \oplus \overline{L_n})$, $\kappa(L'_m \setminus L_n) =\kappa(\overline{L_n} \setminus \overline{L'_m})$, and $\kappa(L'_m \cap L_n) = \kappa(\overline{L'_m} \cup \overline{L_n})$. Thus the prefix-closed witness meets the bounds for boolean operations.

\end{enumerate}
Since the semigroup of a prefix-closed language is the same as that of its complement, which is a right ideal, at least four letters are required to meet all the bounds.
\end{proof}

\section{Prefix-Free Languages}
The complexity of operations on prefix-free languages was studied in~\cite{HSW09,JiKr10,Kra11}, but most complex prefix-free languages were not considered.
As every prefix-free language has an empty quotient, the restricted and unrestricted complexities are the same for binary operations. 
\begin{definition}
\label{def:prefix-free}
For $n\ge 4$, 
let $\Sigma_n=\{a,b,c,d,e_0,\dots,e_{n-3}\}$ and define the DFA 
$\mathcal{D}_n(\Sigma_n) = (Q_n,\Sigma_n,$
$\delta_n,0,\{n-2\}),$ 
where  
$\delta_n$ is defined by the transformations
$ a \colon (n-2 \to n-1) (0,\dots,n-3) $,
$ b \colon (n-2 \to n-1) (0,1) $, 
$ c \colon (n-2 \to n-1) (1\to 0)  $,
$ d \colon (0\to n-2) (Q_n\setminus \{0\} \to n-1)$,
$e_q\colon (n-2\to n-1)(q \to n-2)$ for $q=0,\dots,n-3$.
The transformations induced by $a$ and $b$ coincide when $n=4$.
Let $L_n(\Sigma_n)$ be the language accepted by $\mathcal{D}_n(\Sigma_n)$.
The structure of $\mathcal{D}_n(\Sigma_n)$ is shown in Figure~\ref{fig:free}.
\end{definition}

\begin{figure}[ht]
\unitlength 8pt
\begin{center}\begin{picture}(40,16)(0,0)
\gasset{Nh=2.0,Nw=3.5,Nmr=1.25,ELdist=0.4,loopdiam=1.5}
\node(0)(1,10){0}\imark(0)
\node(1)(9,10){1}
\node(2)(17,10){2}
\node[Nframe=n](3dots)(25,10){$\dots$}
\node(n-4)(33,10){$n-4$}
\node(n-3)(41,10){$n-3$}
\node(n-2)(17,1){$n-2$}\rmark(n-2)
\node(n-1)(33,1){$n-1$}

\drawedge[curvedepth=1,ELdist=.3](0,1){$a,b$}
\drawedge(1,2){$a$}
\drawedge(2,3dots){$a$}
\drawedge(3dots,n-4){$a$}
\drawedge(n-4,n-3){$a$}
\drawedge[curvedepth=-7.0,ELdist=.6](n-3,0){$a$}

\drawedge[ELdist=-2.3](0,n-2){$e_0$}
\drawedge(1,n-2){$e_1$}
\drawedge(2,n-2){$e_2$}
\drawedge[ELdist=-1.6](n-4,n-2){$e_{n-4}$}
\drawedge[ELdist=.7](n-3,n-2){$e_{n-3}$}
\drawedge[ELdist=-1.5](n-2,n-1){$\Sigma_n$}

\drawedge[curvedepth=1,ELdist=-1.25](1,0){$b,c$}

\end{picture}\end{center}
\caption{DFA of a most complex prefix-free language. 
Input $d$ not shown; other missing transitions are self-loops.}
\label{fig:free}
\end{figure}
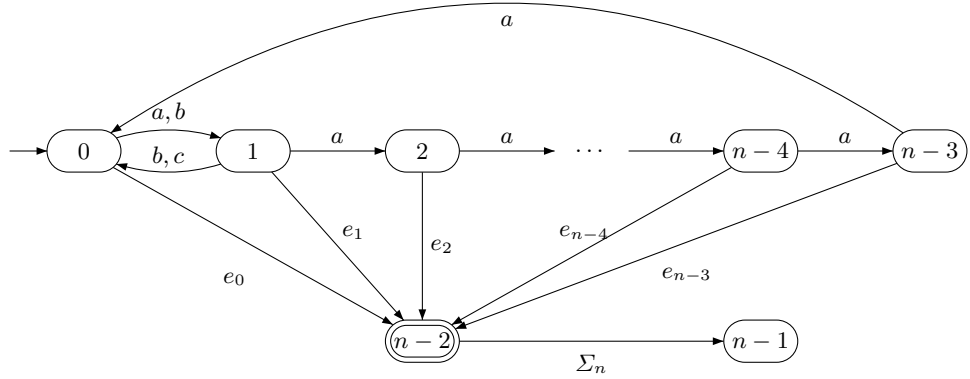

\begin{theorem}
\label{thm:Prefix-free_witness}
For $n\ge 4$, the DFA of Definition~\ref{def:prefix-free} is minimal and  $L_n(\Sigma)$ is a prefix-free language of complexity $n$.
The stream $(L_m(a,b,c,d,e_0,\dots,e_{m-3}) \mid m \ge 4)$  with dialect stream 
$(L_n(b,a,-,-,e_0,$ 
$e_{m-3}) \mid n \ge 4)$
is most complex in the class of  prefix-free languages.
At least $n+2$ inputs are required to meet all the bounds below.

\begin{enumerate}
\item
The syntactic semigroup of $L_n(a,b,c,-,e_0,\dots,e_{n-3})$   has cardinality $n^{n-2}$.  
There is only one maximal transition semigroup of minimal DFAs accepting prefix-free languages. 
Moreover, fewer than $n+1$ inputs do not suffice to meet this bound. 
\item 
The quotients of $L_n(a,-,-,d)$ have complexity $n$, except for  $\varepsilon$ and $\emptyset$, which have complexity 2 and 1, respectively.
\item 
The reverse of $L_n(a, -,c,-,e_0)$ has complexity $2^{n-2}+1$, and $L_n(a, -,c,-,e_0)$ has $2^{n-2}+1$ atoms.
\item
Each atom  $A_S$ of  $L_n(a,b,c,-,e_0)$ has maximal complexity:
\begin{equation*}
	\kappa(A_S) =
	 \begin{cases}
		2, 				& \text{if $S=\{n-2\}$;}\\
		2^{n-1}, 			& \text{if $S= \emptyset$;}\\
		2^{n-2}+1, 		& \text{if $S= Q_{n-2}$;}\\
		2 + \sum_{x=1}^{|S|}\sum_{y=1}^{n-2-|S|}\binom{n-2}{x}\binom{n-2-x}{y},
		 			& \text{if $\emptyset \subsetneq S \subsetneq Q_{n-2}$.}
		\end{cases}
\end{equation*}
\item 
The star of  $L_n(a,-,-,d)$  has complexity $n$.
\item 
The product  $L'_m(a,-,-,d) L_n(a,-,-,d)$ has complexity $m+n-2$.
\item 
For $m,n\ge 4$ but $(m,n)\neq (4,4)$, 
and for any proper binary boolean function $\circ$, the complexity of $L_m(a,b,-,-,e_0,e_{m-3}) \circ L_n(b,a,-,-,e_0,e_{m-3})$
is maximal. In particular,
 these languages meet the bounds
$mn-2$ for union and symmetric difference, $mn-2(m+n-3)$ for intersection,  and $mn-(m+2n-4)$ for difference.
\end{enumerate}
\end{theorem}

\begin{proof}
Since  only state $q$ accepts $a^{n-2-q}d$ for $0\le q\le n-3$, DFA $\mathcal{D}_n(a,-,-,d)$ is minimal.
Since it has only one final state and that state accepts $\{\varepsilon\}$, $L_n(a,-,-,d)$ is prefix-free.
\begin{enumerate}
\item {\bf Semigroup}
The proof that the size of the semigroup is $n^{n-2}$ is very similar to that in~\cite{BLY12}. Inputs $a$, $b$, and $c$ generate all transformations of $Q_{n-2}$.
Moreover, any state $q\in Q_{n-2}$ can be sent to $n-2$ by $e_q$ and to $n-1$ by $e_qe_q$. 
Hence we have all $n^{n-2}$ transformations of $Q_n$ that fix $n-1$ and send $n-2$ to $n-1$.
The maximal transition semigroup  is unique, since it must contain all these transformations.

To prove that at least $n+1$ inputs are necessary, 
we see that 
$e_q\colon (n-2\to n-1) (q\to n-2)$ is in the transition semigroup of $\mathcal{D}_n$. 
There are two types of states in $q \in Q_{n-2}$: those of Type 1, for which $e_q$ is a generator (that is, the transformation of $e_q$ is induced by a single letter), and those of Type 2, for which it is not.
If $e_q$ and $e_p$ are generators, then clearly $e_p\neq e_q$.

If $e_q$ is not a generator, then it must be a composition, $e_q=u_qv_q$, where $u_q$ is in the semigroup and $v_q$ is a generator.  No state can be mapped by $u_q$ to $n-2$ because then $v_q$ would map $n-2$ to $n-1$. Hence $u_q$ must be a permutation of $Q_{n-2}$.
If $q\neq q'$ and $e_q$ and $e_{q'}$ are not generators, then there exist $u_q,v_q$ and $u_{q'},v_{q'}$ as above, such that $e_q=u_qv_q$ and $e_{q'}=u_{q'}v_{q'}$. Then we must have $qu_q\neq q'u_{q'}$; otherwise both $q$ and $q'$ would be mapped to $n-2$. Hence $v_q\neq v_{q'}$ and all the generators of this type are distinct. 

Finally, if $e_q$ is a generator and $v_{q'}$ is as above, then $e_q\neq v_{q'}$, for otherwise $u_{q'}$ would be the identity and  $q'$ would be of Type 1. Therefore, $n-2$ generators are required in addition to those induced by $a$, $b$ and $c$.
\item {\bf Quotients}
This is clear from the definition.

\item {\bf Reversal}
This was proved in~\cite{BLY12}.

\item {\bf Atoms} 
First we establish an upper bound on the complexity of atoms of prefix-free languages. Consider the intersection
$A_S=\bigcap_{i\in S} K_i \cap \bigcap_{i\in \overline{S}} \overline{K_i}$, where $S\subseteq Q_n$, and $\overline{S}=Q_{n}\setminus S$. 
Clearly $n-1$ must be in $\overline{S}$ if $A_S$ is an atom, for $K_{n-1}=\emptyset$.
Since a prefix-free language has only one final state and that state accepts $\varepsilon$, if $n-2\in S$,  no other quotient is in $S$, for then $A_S$ would not be an atom. Hence if $S=\{n-2\}$ then $A_S=\{\varepsilon\}$, and $\kappa(A_S)=2$.

Now suppose $S=\emptyset$; then $A_S=\bigcap_{i\in \overline{S}}\overline{K_i}$. Since $K_{n-1}$ appears in every quotient of $A_S$, there are at most $2^{n-1}$ subsets of $Q_{n-1}$ that can be reached from $A_S$ together with $n-1$. Hence $\kappa(A_S)\le 2^{n-1}$.

If $S=Q_{n-2}$, then $\overline{S}=\{n-2,n-1\}$ and $\bigcap_{i \in \overline{S} }\overline{K_i}=\Sigma^+$. If we reach $\overline{K_{n-1}}=\Sigma^*$, then 
any intersection which has $\{n-2,n-1\}$ in the complemented part is equivalent to one that has only $\{n-1\}$, since no quotient other than $K_{n-2}$ contains $\varepsilon$.
Hence we can reach at most $2^{n-2}-1$ subsets of $Q_{n-2}$, along with the intersection $K_{n-2}\cap \overline{K_{n-1}}= \varepsilon$, and the empty quotient, for a total of $2^{n-2}+1$ states.

Finally, consider the case where $\emptyset \subsetneq S \subsetneq Q_{n-2}$. Then we have from 1 to $|S|$ uncomplemented quotients $K_i$ with $i\in Q_{n-2}$, and from 1 to $n-2 -|S|$ quotients $K_i$ with $i\in Q_{n-2}$ in the complemented part; this leads to the formula given in the theorem.

It remains to be proved that the atoms of $L_n(a,b,c,-,e_0)$ meet these bounds.
Atom $A_{\{n-2\}}$ is equal to $\{\varepsilon\}$ and thus has two quotients as required; assume now that $S \subseteq Q_{n-2}$.
We are interested in the number of distinct quotients of $A_S=\bigcap_{i\in S} K_i \cap \bigcap_{i\in \overline{S}} \overline{K_i}$, where $S\subseteq Q_n \setminus \{n-1\}$.
The quotients $w^{-1}A_S$ have the form $J_{X,Y}= \bigcap_{i\in X} K_i \cap \bigcap_{i\in Y} \overline{K_i}$ where $X = \{i \mid K_i = w^{-1}K_j \text{ for some $j \in S$}\}$ and $Y = \{i \mid K_i = w^{-1}K_j \text{ for some $j \in \overline{S}$}\}$. For brevity, we write $S \xrightarrow{w} X$ and $\overline{S} \xrightarrow{w} Y$; this notation is in agreement with the action of $w$ on the states of $\mathcal{D}_n$ corresponding to $S$ and $\overline{S}$.

Notice $J_{X,Y} = J_{X, Y \cup \{n-2\}}$ for all $X$ and $Y$, except for the case $X = \{n-2\}$ in which $J_{X,Y} \in \{\{\varepsilon\}, \emptyset\}$.
Thus it is sufficient to assume $n-2 \not\in Y$ from now on, as $\{J_{X,Y} \mid n-2 \not\in Y, n-1 \in Y\}$ contains every quotient of $A_S$.
We show that whenever $|X| \le |S|$, $|Y| \le |\overline{S}|$, $n-2 \not\in Y$, and $n-1 \in Y$, there is a word $w \in \{a,b,c,e_0\}^*$ such that $S \xrightarrow{w} X$ and $\overline{S} \xrightarrow{w} Y$ and hence $J_{X,Y}$ is a quotient of $A_S$.
When $S = Q_{n-2}$ we reach all quotients $J_{X,\{n-1\}}$ where $\emptyset \subsetneq X \subseteq Q_{n-2}$ by words in $\{a,b,c\}^*$, we reach $J_{\{n-2\}, \{n-1\}}$ from $J_{\{0\}, \{n-1\}}$ by $e_0$, and from there we reach the empty quotient by $e_0$.
Similarly, when $\emptyset \subseteq S \subsetneq Q_{n-2}$, we reach $J_{X, Y}$ for $X \subseteq Q_{n-2}$ and $Y \cap Q_{n-2} \not= \emptyset$ by words in $\{a,b,c\}^*$,
and the remaining quotients are easily reached using $e_0$.

It remains to show that non-empty quotients $J_{X,Y}$ and $J_{X',Y'}$ are distinct whenever $X \not=X'$ or $Y \not= Y'$.
Notice $J_{X,Y} = \emptyset$ if either $X \cap Y \not= \emptyset$ or $\{n-2\} \subsetneq X$, and $J_{X,Y} = \{\varepsilon\}$ if and only if $X= \{n-2\}$.
Apart from these special cases, every $J_{X,Y}$ is non-empty and does not contain $\varepsilon$.

For any $X \subseteq Q_{n-2}$, let $w_X$ denote a word that maps $X \to \{n-2\}$ and $Q_n\setminus X \to \{n-1\}$;
there is such a word in $\{a,b,c,e_0\}^*$ because $\{a,b,c\}^*$ contains $u \colon (n-2 \rightarrow n-1)(X \rightarrow n-3)(Q_{n-2}\setminus X \rightarrow 0)$, and then $w_X = ue_0ae_0$.
Observe that $w_X \in K_i$ for all $i \in X$ and $w_X \not\in K_j$ for all $j \not\in X$.
Hence, if $X \subseteq Q_{n-2}$ and $Y \subseteq Q_n \setminus X$, then $w_X \in J_{X,Y}$ and $w_{\overline{Y} \cap Q_{n-2}} \in J_{X,Y}$.

Let $X'$ and $Y'$ be any disjoint subsets of $Q_n$ where $n-1 \in Y'$ and $J_{X',Y'} \not= \emptyset$.
If $X' \not= X$ then either $w_X \not\in J_{X',Y'}$ or  $w_{X'} \not\in J_{X,Y}$.
Similarly, if $Y' \not= Y$ (and $Y \oplus Y' \not= \{n-2\}$) then either $w_{\overline{Y}\cap Q_{n-2}} \not\in J_{X',Y'}$ or $w_{\overline{Y'}\cap Q_{n-2}} \not\in J_{X,Y}$.
Thus, any two quotients $J_{X,Y}$ and $J_{X',Y'}$, where $(X,Y) \not= (X',Y')$, are distinct.

When we established the upper bound on $\kappa(A_S)$, we counted the number of reachable, potentially distinct quotients $J_{X,Y}$ of each $A_S$. We have now shown that every reachable $J_{X,Y}$ is a quotient of $A_S$ and determined all the cases when $J_{X,Y} = J_{X',Y'}$. It follows that every bound is met by $L_n(a,b,c,-,e_0)$.

\item {\bf Star}
Proved in~\cite{HSW09}. For the purpose of proving that $n+2$ inputs are required for a most complex prefix-free witness, an outline of the proof is repeated here.

Suppose that $L$ is a prefix-free language with $n$ quotients whose syntactic semigroup is maximal, and $L^*$ has maximal complexity.
We show that $L$ requires an alphabet of size $n+2$.
Towards a contradiction, let $\mathcal{D}=(Q_n, \Sigma, \delta, 0, n-2)$ be a DFA for $L$ where $|\Sigma| = n+1$. Assume $0, 1, \dots, n-3$ are non-final, non-empty states, $n-2$ is the unique final state, and $n-1$ is the empty state. By~\cite{HSW09}, $\mathcal{D}$ must have this structure and $\delta(n-2, w) = n-1$ for any $w \in \Sigma^+$.

Since the syntactic semigroup of $L$ is maximal, each letter of $\Sigma$ has a specific role in $\mathcal{D}$ as described in {\bf 1} of this theorem.
Three letters $a'$, $b'$, and $c'$ are required to induce the transformations on $Q_{n-2}$; these letters cannot map any state of $Q_{n-2}$ to $n-2$ or to $n-1$.
An additional $n-2$ letters $v_0, v_1, \dots, v_{n-3}$ are required to generate $e_q\colon (n-2\to n-1) (q\to n-2)$ for each $q \in Q_{n-2}$, where the action of $e_q$ is induced by a word in $\{a',b',c'\}^*v_q$. Notice $v_q$ cannot map any state of $Q_{n-2}$ to $n-1$, since $e_q$ does not.
In summary, $\Sigma = \{a', b', c', v_0, \dots, v_{n-3}\}$ and for all $\ell \in \Sigma$ and $q \in Q_{n-2}$, $\delta(q, \ell) \not= n-1$.

An NFA for $L^*$ is produced by adding to $\mathcal{D}$ a new initial state $0'$, which is final, adding an $\varepsilon$-transition from $n-2$ to $0$, and deleting the empty state $n-1$.
The transitions from $0'$ are exactly the same as the transitions from $0$. Perform the subset construction on this NFA.
The $n-1$ states $\{0'\}, \{0\}, \{1\}, \dots, \{n-3\}$ are all reachable and distinguishable by words in $\{a',b',c', v_0\}$.
The only way to reach a set containing more than one state is by moving to $n-2$ and using the $\varepsilon$-transition.
This leads to the state $\{0, n-2\}$, but applying any word $w \in \Sigma^+$ deletes $n-2$; thus, $\{0, n-2\}$ is the only reachable set with two or more states.
However, $\{0'\}$ and $\{0, n-2\}$ are indistinguishable, since both are final and $\delta(\{0'\}, w) = \delta(\{0\}, w) = \delta(\{0, n-2\}, w)$ for $w \in \Sigma^+$.

So far, there are only $n-1$ reachable, distinguishable states in the subset construction. The remaining state is $\emptyset$, which can only be reached if there is a letter $\ell$ that moves from $q \in Q_{n-2}$ to $n-1$ in $\mathcal{D}$; a transition from $n-2$ to $n-1$ is not sufficient to reach the empty state. We showed that in our witness no $\ell \in \Sigma$ has $\delta(q, \ell) = n-1$. Therefore, $\kappa(L^*) \le n-1$, a contradiction. To achieve $\kappa(L^*) = n$, an additional letter is required. Therefore, any most complex prefix-free language stream requires $n+2$ inputs.

\item {\bf Product}
Proved in~\cite{HSW09}.

\item {\bf Boolean Operations}
Let $S = Q'_{m-2}\times Q_{n-2}$.  For $0 \le p \le m-1$, let $R_p = \{(p',q) \mid q \in Q_n\}$, and for $0 \le q \le n-1$ let $C_q = \{(p',q) \mid p' \in Q'_m\}$. These are the sets of states in the rows and columns of Figure~\ref{fig:freecross}.
The states in $S$ are reachable from the initial state $(0',0)$ by~\cite[Theorem 1]{BBMR14}.
Every other state in the direct product is reachable from some state in $S$, as illustrated in Figure~\ref{fig:freecross}.

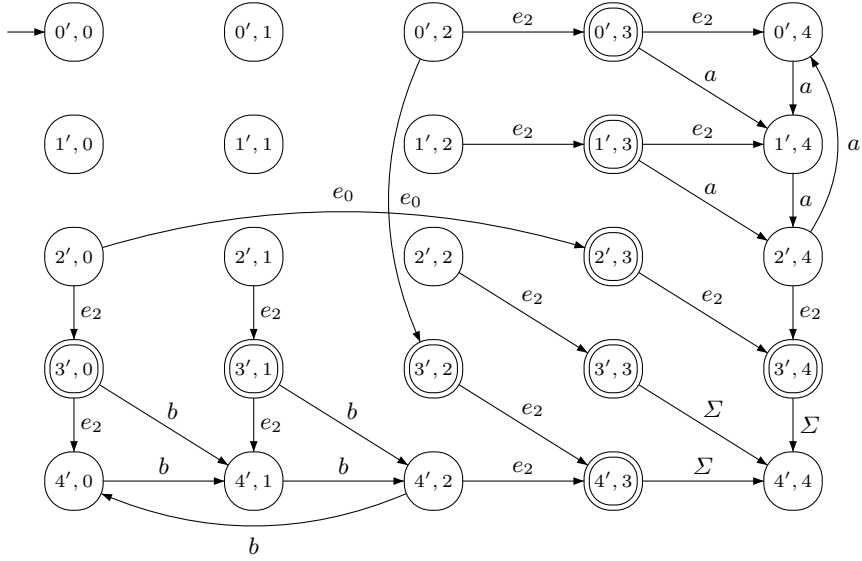
\begin{figure}[ht]
\unitlength 8.5pt
\begin{center}\begin{picture}(35,25)(0,-9)
\gasset{Nh=2.6,Nw=2.6,Nmr=1.2,ELdist=0.3,loopdiam=1.2}
	{\scriptsize
\node(0'0)(2,15){$0',0$}\imark(0'0)
\node(1'0)(2,10){$1',0$}
\node(2'0)(2,5){$2',0$}
\node(3'0)(2,0){$3',0$}\rmark(3'0)
\node(4'0)(2,-5){$4',0$}

\node(0'1)(10,15){$0',1$}
\node(1'1)(10,10){$1',1$}
\node(2'1)(10,5){$2',1$}
\node(3'1)(10,0){$3',1$}\rmark(3'1)
\node(4'1)(10,-5){$4',1$}

\node(0'2)(18,15){$0',2$}
\node(1'2)(18,10){$1',2$}
\node(2'2)(18,5){$2',2$}
\node(3'2)(18,0){$3',2$}\rmark(3'2)
\node(4'2)(18,-5){$4',2$}

\node(0'3)(26,15){$0',3$}\rmark(0'3)
\node(1'3)(26,10){$1',3$}\rmark(1'3)
\node(2'3)(26,5){$2',3$}\rmark(2'3)
\node(3'3)(26,0){$3',3$}\rmark(3'3)
\node(4'3)(26,-5){$4',3$}\rmark(4'3)

\node(0'4)(34,15){$0',4$}
\node(1'4)(34,10){$1',4$}
\node(2'4)(34,5){$2',4$}
\node(3'4)(34,0){$3',4$}\rmark(3'4)
\node(4'4)(34,-5){$4',4$}
	}

\drawedge[curvedepth=-2,ELdist=.5](0'2,3'2){$e_0$}
\drawedge[curvedepth=2,ELdist=0.5](2'0,2'3){$e_0$}

\drawedge(0'2,0'3){$e_2$}
\drawedge(1'2,1'3){$e_2$}
\drawedge(2'2,3'3){$e_2$}
\drawedge(2'0,3'0){$e_2$}
\drawedge(2'1,3'1){$e_2$}

\drawedge(3'0,4'0){$e_2$}
\drawedge(3'1,4'1){$e_2$}
\drawedge(0'3,0'4){$e_2$}
\drawedge(1'3,1'4){$e_2$}

\drawedge(3'0,4'1){$b$}
\drawedge(3'1,4'2){$b$}
\drawedge(3'2,4'3){$e_2$}
\drawedge(4'0,4'1){$b$}
\drawedge(4'1,4'2){$b$}
\drawedge(4'2,4'3){$e_2$}
\drawedge[curvedepth=2,ELdist=0.5](4'2,4'0){$b$}

\drawedge(0'3,1'4){$a$}
\drawedge(1'3,2'4){$a$}
\drawedge(2'3,3'4){$e_2$}
\drawedge(0'4,1'4){$a$}
\drawedge(1'4,2'4){$a$}
\drawedge(2'4,3'4){$e_2$}
\drawedge[curvedepth=-2,ELdist=-1](2'4,0'4){$a$}

\drawedge(3'3,4'4){$\Sigma$}

\drawedge(4'3,4'4){$\Sigma$}

\drawedge(3'4,4'4){$\Sigma$}
\end{picture}\end{center}
\caption{Direct product for union of prefix-free languages with DFAs $\mathcal{D}'_5(a,b,-,-, e_0, e_2)$
and $\mathcal{D}_5(b,a,-, -,e_0, e_2)$ shown partially.}
\label{fig:freecross}
\end{figure}

For $\circ \in \{\cup, \oplus, \setminus, \cap\}$, the direct product recognizes $L'_m \circ L_n$ if the final states are set to be $R_{m-2} \circ C_{n-2}$.
Now we must determine which states are distinguishable with respect to $R_{m-2} \circ C_{n-2}$ for each value of $\circ$.
Consider the DFAs $D'_m = (Q'_{m-2}, \{a,b\}, \delta, 0', \{(m-3)'\})$ and $D_n = (Q_{n-2}, \{b, a\}, \delta, 0, \{n-3\})$.
By~\cite[Theorem 1]{BBMR14}, the states of $S$ are pairwise distinguishable with respect to $(R_{m-3} \circ C_{n-3}) \cap S$.
For any pair of states in $S$, let $w$ be a word that distinguishes them in $(R_{m-3} \circ C_{n-3}) \cap S$; one verifies that further applying $e_{m-3}$ distinguishes them with respect to $R_{m-2} \circ C_{n-2}$. The rest of the distinguishability argument depends on $\circ \in \{\cup, \oplus, \setminus, \cap\}$.

$\cup$: States $((m-1)',n-2)$, $((m-2)',n-1)$, and $((m-2)',n-2)$ are equivalent, since all three are final and any letter sends them to $((m-1)',n-1)$.

States of $R_{m-1}$ are distinguished by words in $b^*e_{m-3}$. States of $C_{n-1}$ are distinguished by words in $a^*e_{m-3}$.
Excluding $((m-1)',n-2)$ and $((m-2)',n-1)$, which are equivalent, states of $R_{m-1}$ are distinguished from states of $C_{n-1}$ by words in $a^*e_{m-3}$.

States of $R_{m-2} \cup C_{n-2}$ are moved to states of $R_{m-1} \cup C_{n-1}$ by applying $e_{m-3}$. Excluding $((m-1)',n-2)$, $((m-2)',n-1)$, and $((m-2)',n-2)$, which are equivalent, every state is mapped by $e_{m-3}$ to a different state of $R_{m-1} \cup C_{n-1}$; hence they are distinguishable.

Finally, we must show that states of $S$ are distinguishable from the states of $R_{m-1} \cup C_{n-1}$.
For any $(p',q) \in S$, there exists $w \in \{a,b\}^*$ such that $(p',q) \xrightarrow{w} (0', n-3)$,
 since both $(p', q)$ and $(0', n-3)$ are reached from $(0',0)$ by words in $\{a,b\}^*$, and $a$ and $b$ permute $S$.
Then $(0', n-3)  \xrightarrow{e_{m-3}} (0', n-2)$
and we have already shown that $(0', n-2)$ is distinguishable from all states in $R_{m-1} \cup C_{n-1}$. Thus, the $mn-2$ remaining states are pairwise distinguishable.

$\oplus$: States $((m-1)',n-2)$ and $((m-2)',n-1)$ are equivalent, and states $((m-2)',n-2)$ and $((m-1)',n-1)$ are equivalent.
The rest of the states are distinguishable by an argument similar to that of union.

$\cap$: State $((m-2)', n-2)$ is the only final state. The remaining non-final states of $R_{m-2} \cup R_{m-1} \cup C_{n-2} \cup C_{n-1}$ are all empty.
Clearly the states of $S$ are non-empty, since $((m-3)', n-3) \xrightarrow{e_{m-3}} ((m-2)', n-2)$. Thus, the remaining $mn - 2(m+n-3)$ states are pairwise distinguishable.

$\setminus$:
The states of $R_{m-1}$ and $((m-2)',n-2)$ are all equivalent.
States $((m-1)',q)$ and $((m-2)',q)$ are equivalent for $0 \le q \le m-3$.
The final states ($R_{m-2} \setminus \{((m-2)', n-2)\}$) are all equivalent.

The states of $C_{n-1}$ are distinguished by words in $a^*e_{m-3}$.
It remains to show that states of $S$ are distinguishable from the states of $C_{n-1}$.
Notice $((m-3)', n-3)$ is distinguished from $((m-3)', n-1)$ by $e_{m-3}$, and from every other state of $C_{n-1}$ by $be_{m-3}$.
For any state of $S$, there exists $w \in \{a,b\}^*$ that sends that state to $((m-3)', n-3)$, and notice $C_{n-1}w \subseteq C_{n-1}$.
So all $mn - (m+2n-4)$ remaining states are pairwise distinguishable.
\end{enumerate}
Note that the  stream  $L_m(a,b,c,d,e_0,e_{m-3})$ with dialect stream $L_n(b,a,c,d,e_0,e_{m-3})$ meets  the bounds for quotients, reversal, atomic complexity, star, product and boolean operations. 
\end{proof}

Using some results from~\cite{JiKr10,Kra11} we define another prefix-free witness stream that meets all the bounds except those for syntactic complexity and atom complexity.
\emph{Moreover, all the bounds are met by dialects over minimal alphabets.}
\begin{definition}
\label{def:prefix-free_small}
For $n\ge 4$,
let
$\mathcal{D}_n(a,c,d,e,f,g) =(Q_n,\Sigma,\delta_n,0,\{n-2\}),$
where
$\Sigma=\{a,c,d,e,f,g\}$,
and $\delta_n$ is defined by the transformations
\begin{itemize}
\item
$ a \colon (n-2 \to n-1) (0,\dots,n-3) $,
\item
$ c \colon (n-2 \to n-1) (1\to 0)  $.
\item
$ d \colon (0\to n-2) (Q_n\setminus \{0\} \to n-1)$,
\item
$e\colon (n-2\to n-1)(n-3 \to n-2)$,
\item
$f\colon (n-2\to n-1)(_0^{n-2}\; q\to q+1)$,
\item
$g\colon (n-2\to n-1)$.
\end{itemize}
Note that $b$ is not used, $a$,  $c$, $d$, and $e$ induce the same transformations as $a$, $c$, $d$, and $e_{n-3}$ in Definition~\ref{def:prefix-free}.
DFA $\mathcal{D}_n(\Sigma)$ is shown in Figure~\ref{fig:small}.
Let $L_n(\Sigma)$ be the language accepted by $\mathcal{D}_n(\Sigma)$.
\end{definition}

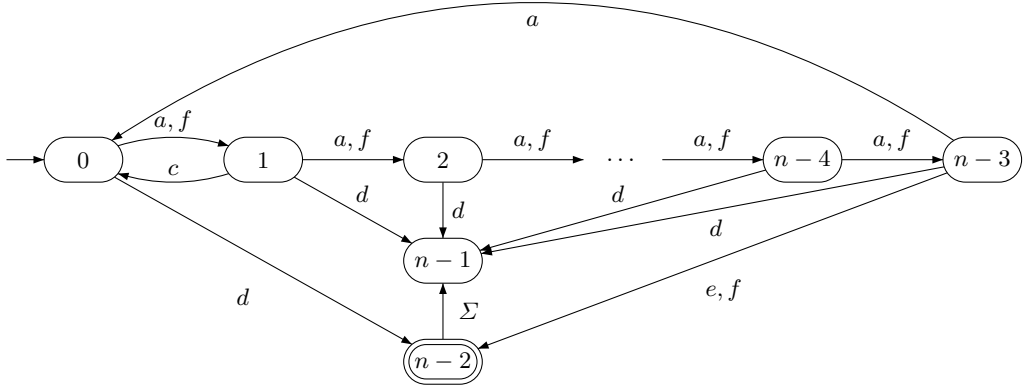
\begin{figure}[ht]
\unitlength 8.5pt
\begin{center}\begin{picture}(40,16)(0,0)
\gasset{Nh=2.0,Nw=3.5,Nmr=1.25,ELdist=0.4,loopdiam=1.5}
\node(0)(1,10){0}\imark(0)
\node(1)(9,10){1}
\node(2)(17,10){2}
\node[Nframe=n](3dots)(25,10){$\dots$}
\node(n-4)(33,10){$n-4$}
\node(n-3)(41,10){$n-3$}
\node(n-2)(17,1){$n-2$}\rmark(n-2)
\node(n-1)(17,5.5){$n-1$}

\drawedge[curvedepth=1,ELdist=.27](0,1){$a,f$}
\drawedge(1,2){$a,f$}
\drawedge(2,3dots){$a,f$}
\drawedge(3dots,n-4){$a,f$}
\drawedge(n-4,n-3){$a,f$}
\drawedge[curvedepth=-7.0,ELdist=.6](n-3,0){$a$}

\drawedge[ELdist=-2.3](0,n-2){$d$}
\drawedge(1,n-1){$d$}
\drawedge(2,n-1){$d$}
\drawedge[ELdist=-1.2](n-4,n-1){$d$}
\drawedge(n-3,n-1){$d$}
\drawedge[ELdist=.7](n-3,n-2){$e,f$}
\drawedge[ELdist=-1.6](n-2,n-1){$\Sigma$}

\drawedge[curvedepth=1,ELdist=-.8](1,0){$c$}
\end{picture}\end{center}
\caption{DFA $\mathcal{D}_n(\Sigma)$  of Definition~\ref{def:prefix-free_small}; missing transitions are self-loops.}
\label{fig:small}
\end{figure}

\begin{proposition}
\label{prop:small}
For $n\ge 4$, the DFA of Definition~\ref{def:prefix-free_small} is minimal and  $L_n(\Sigma)$ is a prefix-free language of complexity $n$. Moreover, all the witnesses for individual operations have minimal alphabets.
\begin{enumerate}
\item
The quotients of $L_n(a,-,-,-,f)$ have complexity $n$, except for the quotient $\varepsilon$ and the empty quotient, which have complexity 2 and 1 respectively.
\item
The reverse of $L_n(a,c,-,e)$ has complexity $2^{n-2}+1$, and $L_n(a,-,c,d)$ has $2^{n-2}+1$ atoms.
\item
The star of  $L_n(a,-,-,d)$  has complexity $n$.
\item
For  $m,n\ge 4$, the product of $L_m(-,-,-,-,f)$ and $ L_n(-,-,-,-,f)$ has complexity $m+n-2$.
\item
\begin{enumerate}
        \item
         $L_m(-,-,-,-,f,g) \cup  L_n(-,-,-,-,g,f)$  and $L_m(-,-,-,-,f,g) \oplus  L_n(-,-,-,-,g,f)$ have complexity  $mn-2$.
        \item
        $L_m(a,-,-,e,-,-) \setminus L_n(-,-,-,-,e,a)$ has complexity $mn-(m+2n-4)$.
        \item
        $L_m(a,-,-,e,-,-) \cap L_n(-,-,-,-,e,a)$ has complexity $mn-2(m+n-3)$.
        \end{enumerate}

\end{enumerate}
\end{proposition}
\begin{proof}
The first claim is obvious.
The second and third claims were proved in Theorem~\ref{thm:Prefix-free_witness}.
(A ternary witness was also used in~\cite{JiKr10} for the reverse, but it had more complicated transitions than our witness.)
The fourth claim is from~\cite{JiKr10}.
The results for union, symmetric difference and intersection were proved in~\cite{JiKr10}, and that for difference in~\cite{Kra11}.
\end{proof}

\section{Prefix-Convex Languages}\label{convexsection}

\begin{lemma}\label{idealorempty}
Let $L$ be a prefix-convex language over $\Sigma$. Either $L$ is a right ideal or $L$ has an empty quotient.
\end{lemma}

\begin{proof}
Suppose that $L$ is not a right ideal. If $L = \emptyset$, then $\varepsilon^{-1}L=L$ is an empty quotient of $L$. 
If $L \not= \emptyset$, we cannot have  $w^{-1}L = \Sigma^*$ for all $w \in L$, because then $L$ would be a right ideal. 
Hence there exists some $w \in L$ such that $w^{-1}L \not= \Sigma^*$. 
Pick any $x \in \Sigma^* \setminus w^{-1}L$; then $w \in L$, but $wx \not\in L$. There cannot be a word $y \in \Sigma^*$ such that $wxy \in L$, because then $wx$ would be in $L$ by prefix convexity. Therefore, $(wx)^{-1}L$ is an empty quotient.
\end{proof}

\begin{proposition}\label{prefixchar}
Let $L$ be a regular language having $n$ quotients and $k$ final quotients, and let $\mathcal D = (Q_n, \Sigma, \delta, 0, F)$ be a minimal DFA recognizing $L$. The following statements are equivalent:
\begin{enumerate}
\item $L$ is prefix-convex.
\item For all $p,q,r \in Q_n$, if $p$ and $r$ are final, $q$ is reachable from $p$, and $r$ is reachable from $q$,  then $q$ is final.
\item Every state reachable in $\mathcal D$ from any final state is either final or empty.
\end{enumerate}
\end{proposition}
\begin{proof} 
({\bf 1} $\implies$ {\bf 2}) Assume {\bf 1} is true. Suppose there exist $p, r \in F$ and $q \in Q_n$ such that $q$ is reachable from $p$ and $r$ is reachable from $q$. Let $w,x,y \in \Sigma^*$ be such that $0 \xrightarrow{w} p$, $p \xrightarrow{x} q$, and $q \xrightarrow{y} r$. It follows that $w$ and $wxy$ are both in $L$, and thus $wx$ is  in $L$ by prefix convexity. Since $\delta(0, wx) = q$,  state $q$ is final.\\
({\bf 2} $\implies$ {\bf 3}) Assume {\bf 2} is true. Take any $p \in F$, $q \in Q_n$, and $x \in \Sigma^*$ such that $\delta(p,x) = q$. If a final state $r$ is reachable from $q$, then $q$ is final by {\bf 2}. Otherwise, $q$ is the empty state.\\
({\bf 3} $\implies$ {\bf 1}) Assume {\bf 3} is true. Let $w, x$, and $y$ be words in $\Sigma^*$ such that $w \in L$ and $wxy \in L$. There are states $p,q$, and $r$ in $Q_n$ such that $\delta(0, w) = p \in F$, $\delta(0, wx) = q$, and $\delta(0, wxy) = r \in F$. 
State $q$ cannot be empty, since  final state $r$ is reachable from $q$. Since $q$ is reachable from final state $p$, it follows from {\bf 3} that $q$ is final. Thus, $wx \in L$. Therefore $L$ is prefix-convex.
\end{proof}

\begin{proposition}\label{convexclasses}
Let $L$ be a non-empty prefix-convex language having $n$ quotients and $k$ final quotients, and let $\mathcal D = (Q_n, \Sigma, \delta, 0, F)$ be a minimal DFA recognizing $L$.
\begin{enumerate}
\item  $L$ is prefix-closed if and only if $0 \in F$.
\item $L$ is prefix-free if and only if $\mathcal D$ has a unique final state $p$, an empty state $p'$, and $\delta(p, a) = p'$ for all $a \in \Sigma$.
\item $L$ is a right ideal if and only if $\mathcal D$ has a unique final state $p$ and $\delta(p, a) = p$ for all $a \in \Sigma$.
\end{enumerate}
\end{proposition}

\begin{proof}
Note that $|F| = k$ and $k \ge 1$ since $L$ is non-empty.
\begin{enumerate}

\item Suppose $L$ is prefix-closed. Clearly, $\varepsilon$ is a prefix of some word in $L$, since $L$ is non-empty. Thus $\varepsilon \in L$, and so $0 \in F$.
Conversely, suppose $0 \in F$.
For any $wx \in L$, there are states $q, r \in Q_n$ such that $0 \xrightarrow{w} q \xrightarrow{x} r$, and $r$ is final.
By Proposition \ref{prefixchar}, since $0,r \in F$, $q$ is reachable from $0$, and $r$ is reachable from $q$,  we have $q \in F$.
Hence $w \in L$, and therefore $L$ is prefix-closed.

\item Suppose $L$ is prefix-free. If $q \in Q_n$ and $p \in F$ are distinct and $q$ is reachable from $p$, then $q$ cannot be final as that would imply $p \not\in F$. In particular, for any $p \in F$ and $a \in \Sigma$, $\delta(p,a) \not\in F$. By Proposition \ref{prefixchar}, $\delta(p,a)$ must be the empty state for all $a \in \Sigma$. Thus, the transitions from all final states are identical, and hence all final states are equivalent. By minimality, $\mathcal D$ has a unique final state $p$, an empty state $p'$, and $\delta(p, a) = p'$ for all $a \in \Sigma$.

For the converse, suppose $F = \{p\}$, $p' \in Q_n$ is an empty state, and $\delta(p,a) = p'$ for all $a \in \Sigma$. Then $w \in L$ if and only if $\delta(0,w) = p$. For all $w \in L$ and $a \in \Sigma$, we have $\delta(0, wa) = p'$. Thus, whenever $w \in L$ and $wx \in L$, we have $x = \varepsilon$. Therefore, $L$ is prefix-free.

\item Suppose $L$ is a right ideal. 
For all $w \in L$ we have $L\supseteq w \Sigma^*$, and hence $w^{-1}L \supseteq \Sigma^*$, meaning that 
$w^{-1}L = \Sigma^*$.
Hence, for any final state $q \in F$ and $x \in \Sigma^*$, $\delta(q, x) \in F$. This implies that all final states are equivalent. By minimality, there is a unique final state $p$. Since $\delta(p, a) \in F$ for all $a \in \Sigma$, it follows that $\delta(p, a) = p$ for all $a \in \Sigma$.
For the converse, suppose $F = \{p\}$ and $\delta(p, a) = p$ for all $a \in \Sigma$. Then $w \in L$ if and only if $\delta(0, w) = p$. Hence, for all $w \in L$ and $x \in \Sigma^*$, we have $\delta(0, wx) = p$. Thus, $w\Sigma^* \subseteq L$ for all $w \in L$, and so $L = L \Sigma^*$. Therefore, $L$ is a right ideal.
\end{enumerate}
\end{proof}

\section{Proper Prefix-Convex Languages}\label{sec:prefix-convex}
Recall that a prefix-convex language $L$ is \textit{proper} if it is not a right ideal and it is neither prefix-closed nor prefix-free. Moreover, it is \emph{$k$-proper} if it has $k$ final states, $1\le k\le n-2$.
Every minimal DFA for a $k$-proper language with complexity $n$ has the same general structure: there are $n-1-k$ non-final, non-empty states, $k$ final states, and one empty state. Every letter fixes the empty state and, by Proposition~\ref{prefixchar}, no letter sends a final state to a non-final, non-empty state. 
As every proper language has an empty quotient, the restricted and unrestricted complexities are the same for binary operations. 

Next we define a stream of particular $k$-proper DFAs and languages.

\begin{definition}\label{def:proper}
For $n \ge 3$ and $1 \le k \le n-2$, let $\mathcal D_{n,k}(\Sigma) = (Q_n, \Sigma, \delta_{n,k}, 0, F_{n,k})$ where $\Sigma = \{a, b, c_1, c_2, d_1,d_2,e\}$, $F_{n,k} = \{n-1-k, \dots , n-2\}$, and $\delta_{n,k}$ is defined by the transformations  \begin{align*}
a &\colon \begin{cases}
(1, \dots, n-2-k)(n-1-k, n-k), &\text{ \emph{if} $n-1-k$ \emph{is even and} $k \ge 2$;} \\
(0, \dots, n-2-k)(n-1-k, n-k), &\text{ \emph{if} $n-1-k$ \emph{is odd and} $k \ge 2$;} \\
(1, \dots, n-2-k), &\text{ \emph{if} $n-1-k$ \emph{is even and} $k = 1$;} \\
(0, \dots, n-2-k), &\text{ \emph{if} $n-1-k$ \emph{is odd and} $k = 1$.} \\
\end{cases}\\ 
b &\colon \begin{cases}
(n-k, \dots, n-2)(0, 1), &\text{ \emph{if} $k$ \emph{is even and} $n-1-k \ge 2$;} \\
(n-1-k, \dots, n-2)(0, 1), &\text{ \emph{if} $k$ \emph{is odd and} $n-1-k \ge 2$;} \\
(n-k, \dots, n-2), &\text{ \emph{if} $k$ \emph{is even and} $n-1-k = 1$;} \\
(n-1-k, \dots, n-2), &\text{ \emph{if} $k$ \emph{is odd and} $n-1-k = 1$.} \\
\end{cases}\\
c_1 &\colon \begin{cases}(1 \to 0), &\text{\emph{if} $n-1-k \ge 2$;} \\ \quad~ \mathbbm{1}, &\text{\emph{if} $n-1-k = 1$.} \end{cases}\\
c_2 &  \colon \begin{cases}(n-k \to n-1-k), &\text{\emph{if} $k \ge 2$;} \\ \quad\quad\quad~ \mathbbm{1}, &\text{\emph{if} $k = 1$.} \end{cases} \\
d_1 &\colon (n-2-k \rightarrow n-1)(_0^{n-3-k} \;\; q\rightarrow q+1).\\
d_2 &\colon (_{n-1-k}^{n-2} \;\; q\rightarrow q+1).\\
e &\colon (0 \to n-1-k).
\end{align*}
Additionally define $E_{n,k} = \{0, \dots, n-2-k\}$; it is often useful to partition $Q_n$ into $E_{n,k}$, $F_{n,k}$, and $\{n-1\}$.
Note that letters $a$ and $b$ have complementary behaviours on $E_{n,k}$ and $F_{n,k}$, depending on the parities of $n$ and $k$.
Letters $c_1$, and $d_1$ act on $E_{n,k}$ in exactly in the same way as $c_2$, and $d_2$ act on $F_{n,k}$.
In addition, $d_1$ and $d_2$ send states $n-2-k$ and $n-2$, respectively, to state $n-1$, and letter $e$ connects the two parts of the DFA.
The structure of $\mathcal{D}_{n,k}(\Sigma)$ is shown in Figures~\ref{fig:ProperWit1} and \ref{fig:ProperWit2} for certain parities of $n-1-k$ and $k$. 
Let $L_{n,k}(\Sigma)$ be the language recognized by $\mathcal D_{n,k}(\Sigma)$.
\end{definition}

\begin{figure}[th]
\unitlength 8.5pt
\begin{center}\begin{picture}(40,12)(0,-1.6)
\gasset{Nh=2.0,Nw=5.6,Nmr=1.25,ELdist=0.4,loopdiam=1.5}
\node(0)(-1,7){0}\imark(0)
\node(1)(9.5,7){1}
\node(2)(18,7){2}
\node[Nframe=n](3dots)(26.5,7){$\dots$}
	{\small
\node(n-2-k)(35,7){$n-2-k$}
	}
	{\small
\node(n-1)(41,3.5){$n-1$}
	}

	{\small
\node(n-1-k)(-1,0){$n-1-k$}\rmark(n-1-k)
	}
	{\small
\node(n-k)(9.5,0){$n-k$}\rmark(n-k)
	}
	{\small
\node(n-k+1)(18,0){$n-k+1$}\rmark(n-k+1)
	}
\node[Nframe=n](3dots1)(26.5,0){$\dots$}
\node(n-2)(35,0){$n-2$}\rmark(n-2)
\drawedge[curvedepth= 1.2,ELdist=-1.3](0,1){$a, b, d_1$}
\drawedge(n-2-k,n-1){$d_1$}
\drawedge(0,n-1-k){$e$}
\drawedge[curvedepth= 1.2,ELdist=.2](1,0){$b,c_1$}
\drawedge[curvedepth= -4.5,ELdist=.3](n-2-k,0){$a$}
\drawedge(1,2){$a, d_1$}
\drawedge(2,3dots){$a, d_1$}
\drawedge(3dots,n-2-k){$a, d_1$}
\drawedge[curvedepth= 1.2,ELdist=-1.3](n-1-k,n-k){$a, d_2$}
\drawedge[curvedepth= 1.2,ELdist=.25](n-k,n-1-k){$a,c_2$}
\drawedge(n-k,n-k+1){$b, d_2$}
\drawedge(n-k+1,3dots1){$b, d_2$}
\drawedge(3dots1,n-2){$b, d_2$}
\drawedge(n-2,n-1){$d_2$}
\drawedge[curvedepth= -4.5,ELdist=.3](n-2,n-k){$b$}
\end{picture}\end{center}
\caption{DFA $\mathcal{D}_{n,k}(a,b,c_1,c_2, d_1, d_2,e)$ of Definition~\ref{def:proper} when $n-1-k$ is odd, $k$ is even, and both are at least $2$; missing transitions are self-loops.}
\label{fig:ProperWit1}
\end{figure}
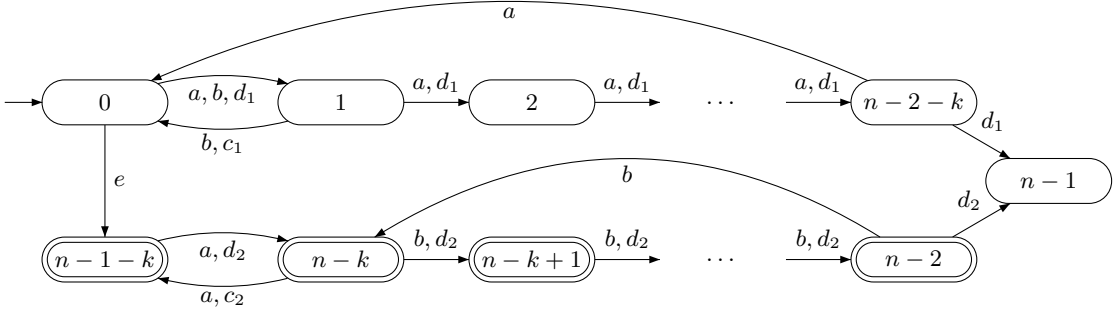

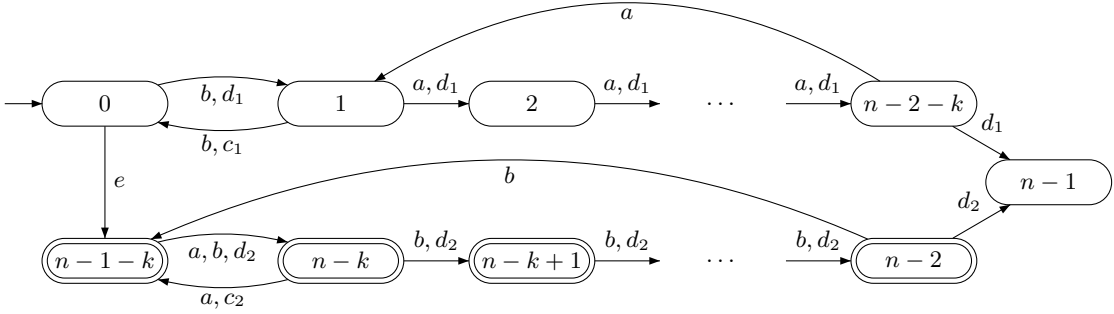
\begin{figure}[th]
\unitlength 8.5pt
\begin{center}\begin{picture}(40,14)(0,-1.6)
\gasset{Nh=2.0,Nw=5.6,Nmr=1.25,ELdist=0.4,loopdiam=1.5}
\node(0)(-1,7){0}\imark(0)
\node(1)(9.5,7){1}
\node(2)(18,7){2}
\node[Nframe=n](3dots)(26.5,7){$\dots$}
	{\small
\node(n-2-k)(35,7){$n-2-k$}
	}
	{\small
\node(n-1)(41,3.5){$n-1$}
	}

	{\small
\node(n-1-k)(-1,0){$n-1-k$}\rmark(n-1-k)
	}
	{\small
\node(n-k)(9.5,0){$n-k$}\rmark(n-k)
	}
	{\small
\node(n-k+1)(18,0){$n-k+1$}\rmark(n-k+1)
	}
\node[Nframe=n](3dots1)(26.5,0){$\dots$}
\node(n-2)(35,0){$n-2$}\rmark(n-2)
\drawedge[curvedepth= 1.2,ELdist=-1.3](0,1){$b, d_1$}
\drawedge(n-2-k,n-1){$d_1$}
\drawedge(0,n-1-k){$e$}
\drawedge[curvedepth= 1.2,ELdist=.2](1,0){$b,c_1$}
\drawedge[curvedepth= -4.5,ELdist=.3](n-2-k,1){$a$}
\drawedge(1,2){$a, d_1$}
\drawedge(2,3dots){$a, d_1$}
\drawedge(3dots,n-2-k){$a, d_1$}
\drawedge[curvedepth= 1.2,ELdist=-1.3](n-1-k,n-k){$a, b, d_2$}
\drawedge[curvedepth= 1.2,ELdist=.25](n-k,n-1-k){$a,c_2$}
\drawedge(n-k,n-k+1){$b, d_2$}
\drawedge(n-k+1,3dots1){$b, d_2$}
\drawedge(3dots1,n-2){$b, d_2$}
\drawedge(n-2,n-1){$d_2$}
\drawedge[curvedepth= -4.5,ELdist=.3](n-2,n-1-k){$b$}
\end{picture}\end{center}
\caption{DFA $\mathcal{D}_{n,k}(a,b,c_1,c_2, d_1, d_2,e)$ of Definition~\ref{def:proper} when $n-1-k$ is even, $k$ is odd, and both are at least $2$; missing transitions are self-loops.}
\label{fig:ProperWit2}
\end{figure}

\begin{theorem}[Proper Prefix-Convex Languages]
\label{thm:propermain}
For $n\ge 3$ and $1 \le k \le n-2$, the DFA $\mathcal{D}_{n,k}(\Sigma)$ of Definition~\ref{def:proper} is minimal and  $L_{n,k}(\Sigma)$ is a $k$-proper language of complexity $n$.  
This language is most complex in the class of $k$-proper prefix-convex languages; in particular, it meets all the complexity bounds below.
At least 7 letters are required to meet these bounds.
\begin{enumerate}
\item
The syntactic semigroup of $L_{n,k}(\Sigma)$ has cardinality $n^{n-1-k}(k+1)^k$; the maximal value $n(n-1)^{n-2}$ is reached only when $k=n-2$. At least 7 letters are required to meet this bound.
\item
The non-empty, non-final quotients of $L_{n,k}(a, b, -, -, -, d_2, e)$ have complexity $n$, the final quotients have complexity $k+1$, and the empty quotient has complexity 1.
\item
The reverse of $L_{n,k}(a,b,-,-,-,d_2,e)$ has complexity $2^{n-1}$; moreover, the language has $2^{n-1}$ atoms for all $k$.
\item
For each atom $A_S$ of $L_{n,k}(\Sigma)$, write $S = X_1 \cup X_2$, where $X_1 \subseteq E_{n,k}$ and $X_2 \subseteq F_{n,k}$.
Let $\overline{X_1} = E_{n,k}\setminus X_1$ and $\overline{X_2} = F_{n,k}\setminus X_2$.
If $X_2 \not= \emptyset$, then $\kappa(A_S) =$
$$1 + \sum_{x_1=0}^{|X_1|}\sum_{x_2=1}^{|X_1| + |X_2| - x_1}\sum_{y_1=0}^{|\overline{X_1}|}\sum_{y_2=0}^{|\overline{X_1}|+|\overline{X_2}| -y_1}\binom{n-1-k}{x_1}\binom{k}{x_2}\binom{n-1-k-x_1}{y_1}\binom{k-x_2}{y_2}.$$
If $X_1 \not= \emptyset$ and $X_2 = \emptyset$, then $\kappa(A_S) =$
\begin{multline*}
1 + \sum_{x_1=0}^{|X_1|}\sum_{x_2=0}^{|X_1| - x_1}\sum_{y_1=0}^{|\overline{X_1}|}\sum_{y_2=0}^{k}\binom{n-1-k}{x_1}\binom{k}{x_2}\binom{n-1-k-x_1}{y_1}\binom{k-x_2}{y_2} \\
-2^k\sum_{y=0}^{|\overline{X_1}|}\binom{n-1-k}{y}.
\end{multline*}
Otherwise, $S = \emptyset$ and $\kappa(A_S) = 2^{n-1}$.
\item
The star of $L_{n,k}(a,b,-,-,d_1,d_2, e)$  has complexity $2^{n-2}+2^{n-2-k}+1$. The maximal value $2^{n-2}+2^{n-3}+1$ is reached only when $k=1$.
\item
The product of $L'_{m,j}(a,b,c_1,-, d_1, d_2, e)$ and  $ L_{n,k}(a,d_2, c_1,-,d_1, b, e)$ has complexity $m-1-j +j2^{n-2}+2^{n-1}$. The maximal value $m 2^{n-2} + 1$ is reached only when $j=m-2$.
\item
For $m,n\ge 3$, $1 \le j \le m-2$, and $1 \le k \le n-2$, let $L'_{m,j} = L'_{m,j}(a, b, c_1, -, d_1, d_2, e)$ and $L_{n,k} = L_{n,k}(a, b, e, -, d_2, d_1, c_1)$. For any proper binary boolean function $\circ$, the complexity of $L'_{m,j} \circ L_{n,k}$ is maximal. In particular,
	\begin{enumerate}
	\item
	$L'_{m,j} \cup L_{n,k}$ and $L'_{m,j} \oplus L_{n,k}$ have complexity  $mn$.
	\item
	$L'_{m,j} \setminus L_{n,k}$ has complexity $mn-(n-1)$.
	\item
	$L'_{m,j} \cap L_{n,k}$ has complexity $mn-(m+n-2)$.
	\end{enumerate}
\end{enumerate}

\end{theorem}

\begin{proof}
The remainder of this paper is devoted to the proof of this main theorem. 
The longer parts of the proof are separated into individual propositions and lemmas.

DFA $\mathcal{D}_{n,k}(a,b,-,-,-,d_2,e)$ is easily seen to be minimal.
Language $L_{n,k}(\Sigma)$ is $k$-proper by Propositions~\ref{prefixchar} and~\ref{convexclasses}.
\begin{enumerate}
\item
See Proposition~\ref{prop:propersyntactic}.
\item
If the initial state of $\mathcal{D}_{n,k}(a, b, -, -, -, d_2, e)$ is changed to $q \in E_{n,k}$, the new DFA accepts a quotient of $L_{n,k}$ and is still minimal; hence the complexity of that quotient is $n$. If the initial state is changed to $q \in F_{n,k}$ then states in $E_{n,k}$ are unreachable, but the DFA on $\{n-1-k, \dots, n-1\}$ is minimal; hence the complexity of that quotient is $k+1$. The remaining quotient is empty, and hence has complexity 1. By Proposition~\ref{prefixchar}, these are maximal for $k$-proper languages.
\item
See Propositions~\ref{prop:properreversebound} and~\ref{prop:properreverse} for the reverse. It was shown in~\cite{BrTa14} that the number of atoms is equal to the complexity of the reverse.
\item
See Proposition~\ref{prop:properatoms}.
\item
See Propositions~\ref{prop:properstarbound} and~\ref{prop:properstar}.
\item
See Propositions~\ref{prop:properproductbound} and~\ref{prop:properproduct}. The maximal value of  $m 2^{n-2}+1$ occurs only when $j = m-2$, since the complexity is increasing with $j$.
\item
By~\cite[Theorem 2]{Brz10}, all boolean operations on regular languages over the same alphabet have the upper bound $mn$. Since all proper languages have an empty quotient this gives the bound for {\bf \emph a}; the bounds for {\bf \emph b} and {\bf \emph c} follow from~\cite[Theorem 5]{Brz10}. See Proposition~\ref{prop:properboolean} for the proof that all these bounds are tight for  $L'_{m,j} \circ L_{n,k}$.
\end{enumerate}
\end{proof}

\begin{lemma}
\label{lem:properpermutations} Let $n \ge 1$ and $1 \le k \le n-2$. For any permutation $t$ of $Q_n$ such that $E_{n,k}t = E_{n,k}$, $F_{n,k}t = F_{n,k}$, and $(n-1)t = n-1$, there is a word $w \in \{a,b\}^*$ that induces $t$ on $\mathcal{D}_{n,k}$.
\end{lemma}

\begin{proof}
Only $a$ and $b$ induce permutations of $Q_n$; every other letter induces a properly injective map. Furthermore, $a$ and $b$ permute $E_{n,k}$ and $F_{n,k}$ separately, and both fix $n-1$. Hence every $w \in \{a,b\}^*$ induces a permutation  on $Q_n$ such that $E_{n,k}w = E_{n,k}$, $F_{n,k} w = F_{n,k}$, and $(n-1)w = n-1$. Each such permutation naturally corresponds to an element of $S_{n-1-k} \times S_k$, where $S_m$ denotes the symmetric group on $m$ elements. To be consistent with the DFA, assume $S_{n-1-k}$ contains permutations of $\{0, \dots, n-2-k\}$ and $S_k$ contains permutations of $\{n-1-k, \dots, n-2\}$. Let $s_a$ and $s_b$ denote the group elements corresponding to the transformations induced by $a$ and $b$ respectively. We show that $s_a$ and $s_b$ generate $S_{n-1-k} \times S_k$.

It is well known that $(0, \dots, m-1)$, and $(0,1)$ generate the symmetric group on $\{0, \dots, m-1\}$ for any $m \geq 2$. Note that $(1, \dots, m-1)$ and $(0, 1)$ are also generators, since $(0, 1)\ast(1, \dots, m-1) = (0, \dots, m-1)$.

If $n-1-k=1$ and $k=1$ then $S_{n-1-k} \times S_k$ is the trivial group. If $n-1-k = 1$ and $k \ge 2$, then $s_a = (\mathbbm{1}, (n-1-k,n-k))$ and $s_b$ is either $(\mathbbm{1}, (n-1-k, \dots, n-2))$ or $(\mathbbm{1}, (n-k, \dots, n-2))$, and either pair generates the group. There is a similar argument when $k = 1$.

Assume now $n-1-k \ge 2$ and $k \ge 2$. If $n-1-k$ is odd then $s_a = ((0,\dots, n-2-k), (n-1-k,n-k))$, and hence $s_a^{n-1-k} = ((0,\dots, n-2-k)^{n-1-k}, (n-1-k,n-k)^{n-1-k}) = (\mathbbm{1}, (n-1-k,n-k))$. Similarly if $n-1-k$ is even then $s_a = ((1,\dots, n-2-k), (n-1-k,n-k))$, and hence $s_a^{n-2-k}= (\mathbbm{1}, (n-1-k,n-k))$. Therefore $(\mathbbm{1}, (n-1-k,n-k))$ is always generated by $s_a$. By symmetry, $((0,1), \mathbbm{1})$ is always generated by $s_b$ regardless of the parity of $k$.

Since we can isolate the transposition component of $s_a$, we can isolate the other component as well: $(\mathbbm{1}, (n-1-k,n-k))s_a$ is either $((0, \dots, n-2-k),\mathbbm{1})$ or $((1, \dots, n-2-k),\mathbbm{1})$. Paired with $((0,1), \mathbbm{1})$, either element is sufficient to generate $S_{n-1-k} \times \{\mathbbm{1}\}$. Similarly, $s_a$ and $s_b$ generate $\{\mathbbm{1}\} \times S_k$. Therefore $s_a$ and $s_b$ generate $S_{n-1-k} \times S_k$. It follows that $a$ and $b$ generate all permutations $t$ of $Q_n$ such that $E_{n,k}t = E_{n,k}$, $F_{n,k}t = F_{n,k}$, and $(n-1)t = n-1$.
\end{proof}

\begin{proposition}[Syntactic Semigroup]\label{prop:propersyntactic}
The syntactic semigroup of $L_{n,k}(\Sigma)$ has cardinality $n^{n-1-k}(k+1)^k$, which is maximal for a $k$-proper language. Furthermore, seven letters are required to meet this bound. The maximum value $n(n-1)^{n-2}$ is reached only when $k=n-2$.
\end{proposition}
\begin{proof}
Let $L$ be a $k$-proper language of complexity $n$ and let $\mathcal{D}$ be a minimal DFA recognizing $L$.  By Lemma~\ref{idealorempty}, $\mathcal{D}$ has an empty state. By Proposition~\ref{prefixchar}, the only states that can be reached from one of the $k$ final states are either final or empty. Thus, a transformation in the transition semigroup of $\mathcal{D}$ may map each final state to one of $k+1$ possible states, while each non-final, non-empty state may be mapped to any of the $n$ states. Since the empty state can only be mapped to itself, we are left with $n^{n-1-k}(k+1)^k$ possible transformations in the transition semigroup. Therefore the syntactic semigroup of any $k$-proper language has size at most $n^{n-1-k}(k+1)^k$.

Now consider the transition semigroup of $\mathcal{D}_{n,k}(\Sigma)$. Every transformation $t$ in the semigroup must satisfy $F_{n,k}t \subseteq F_{n,k} \cup \{n-1\}$ and $(n-1)t = n-1$, since any other transformation would violate prefix-convexity. We show that the semigroup contains every such transformation, and hence the syntactic semigroup of $L_{n,k}(\Sigma)$ is maximal.

First, consider the transformations $t$ such that $E_{n,k}t \subseteq E_{n,k} \cup\{n-1\}$ and $qt =q$ for all $q \not\in E_{n,k}$.
By Lemma~\ref{lem:properpermutations}, $a$ and $b$ generate every permutation of $E_{n,k}$. When $t$ is not a permutation, we can use $c_1$ to combine any states $p$ and $q$: apply a permutation on $E_{n,k}$ so that $p \to 0$ and $q \to 1$, and then apply $c_1$ so that $1 \to 0$. Repeat this method to combine any set of states, and further apply permutations to induce the desired transformation while leaving the states of $F_{n,k} \cup \{n-1\}$ in place. The same idea applies with $d_1$; apply permutations and $d_1$ to send any states of $E_{n,k}$ to $n-1$.
Hence $a$, $b$, $c_1$, and $d_1$ generate every transformation $t$ such that $E_{n,k}t \subseteq E_{n,k}\cup\{n-1\}$ and $qt = q$ for all $q \in F_{n,k} \cup \{n-1\}$.

We can make the same argument for transformations that act only on $F_{n,k}$ and fix every other state. Since $c_2$ and $d_2$ act on $F_{n,k}$ exactly as $c_1$ and $d_1$ act on $E_{n,k}$, $a$, $b$, $c_1$, and $d_1$ generate every transformation $t$ such that $F_{n,k}t \subseteq F_{n,k} \cup \{n-1\}$ and $qt = q$ for all $q \in E_{n,k}\cup\{n-1\}$.
It follows that $a$, $b$, $c_1$, $c_2$, $d_1$, and $d_2$ generate every transformation $t$ such that $E_{n,k}t \subseteq E_{n,k} \cup \{n-1\}$, $F_{n,k}t \subseteq F_{n,k} \cup \{n-1\}$, and $(n-1)t = n-1$.

Note the similarity between this DFA restricted to the states $E_{n,k} \cup \{n-1\}$ (or $F_{n,k} \cup \{n-1\}$) and the most complex witness for right ideals.
The argument for the size of the syntactic semigroup of right ideals is similar to this; see~\cite{BrYe11}.

Finally, consider an arbitrary transformation $t$ such that $F_{n,k}t \subseteq F_{n,k} \cup \{n-1\}$ and $(n-1)t = n-1$. Let $j_t$ be the number of states $p \in E_{n,k}$ such that $pt \in F_{n,k}$. We show by induction on $j_t$ that $t$ is in the transition semigroup of $\mathcal{D}$. If $j_t = 0$, then $t$ is generated by $\Sigma \setminus \{e\}$.
If $j_t \ge 1$, there exist $p,q \in E_{n,k}$ such that $pt \in F_{n,k}$ and $q$ is not in the image of $t$.
Consider the transformations $s_1$ and $s_2$ defined by
$qs_1 = pt$ and $rs_1 = r$ for $r\not=q$, and
$ps_2 = q$ and $rs_2 = rt$ for $r\not=p$.
Notice that $t = s_1 \ast s_2$ and $j_{s_2} = j_t -1$. One can verify that $s_1 = (n-1-k,pt) \ast (0,q) \ast  (0 \to n-1-k) \ast (0,q)\ast  (n-1-k,pt)$. From this expression, we see that $s_1$ is the composition of transpositions induced by words in $\{a,b\}^*$ and the transformation $(0 \to n-1-k)$ induced by $e$, and hence $s_1$ is generated by $\Sigma$. Thus, $t = s_1 \ast s_2$ is in the transition semigroup. By induction on $j_t$, it follows that the syntactic semigroup of $L_{n,k}$ is maximal.

Now we show that seven letters are required to meet this bound. Two letters (like $a$ and $b$) that induce permutations are required to generate permutations, since clearly one letter is not sufficient. Every other letter will induce a properly injective map.
A letter (like $c_1$) that induces a properly injective map on $E_{n,k}$ and permutes $F_{n,k}$ is required.
Similarly, a letter (like $c_2$) that permutes $E_{n,k}$ and induces a properly injective map on $F_{n,k}$ is required.
A letter (like $d_1$) that sends a state in $E_{n,k}$ to $n-1$ and permutes $F_{n,k}$ is required.
Similarly, a letter (like $d_2$) that sends a state in $F_{n,k}$ to $n-1$ and permutes $E_{n,k}$ is required.
Finally, a letter (like $e$) that connects $E_{n,k}$ and $F_{n,k}$ is required.

For a fixed $n$, we may want to know which $k \in \{1 , \dots, n-2\}$ maximizes $s_n(k) = n^{n-1-k}(k+1)^k$; this corresponds to the largest syntactic semigroup of a proper prefix-convex language with $n$ quotients. We show that $s_n(k)$ is largest at $k = n-2$. Consider the ratio $\frac{s_n(k+1)}{s_n(k)} = \frac{(k+2)^{k+1}}{n(k+1)^k}$. Notice this ratio is increasing with $k$, and hence $s_n$ is a convex function on $\{1, \dots ,n-2\}$. It follows that the maximum value of $s_n$ must occur at one the endpoints, $1$ and $n-2$.

Now we show that $s_n(n-2) \ge s_n(1)$ for all $n \ge 3$. We can check this explicitly for $n = 3, 4, 5$.
When $n \ge 6$, $s_n(n-2)/s_n(1) = \frac{n}{2} \left(\frac{n-1}{n}\right)^{n-2} \ge 3 \left(1/e\right) > 1$; hence the largest syntactic semigroup of $L_{n,k}(\Sigma)$ occurs only at $k = n-2$ for all $n \ge 3$.
\end{proof}

\begin{proposition}[Reverse Bound]\label{prop:properreversebound}
For any regular language $L$ of complexity $n$ with an empty quotient, the reversal has complexity at most $2^{n-1}$.
\end{proposition}
\begin{proof}
Let $\mathcal{D} = (Q_n, \Sigma, \delta, 0, F)$ be a minimal DFA recognizing $L$. Construct an NFA $\mathcal N$ recognizing the reverse of $L$ from $\mathcal{D}$ by reversing each transition, letting the initial state $0$ be the unique final state, and letting the final states $F$ be the initial states. Notice that the empty state of $\mathcal{D}$ is not reachable in $\mathcal N$ since the transitions have been reversed; thus we may delete it. Applying the subset construction to $\mathcal N$ yields a DFA whose states are subsets of $Q_{n-1}$, and hence has at most $2^{n-1}$ states.
\end{proof}

\begin{proposition}[Reverse]\label{prop:properreverse}
The reverse of $L_{n,k}(a,b,-,-,-,d_2, e)$ has complexity $2^{n-1}$ for $n \ge 3$ and $1 \le k \le n-2$.
\end{proposition}

\begin{proof}
Let $\mathcal{D}_{n,k} = (Q_n, \{a,b,d_2,e\}, \delta_{n,k}, 0, F_{n,k})$ denote the DFA $\mathcal{D}_{n,k}(a,b,-,-,-,d_2,e)$ of Definition~\ref{def:proper} and let $L_{n,k} = L(D_{n,k})$. As before, construct an NFA $\mathcal N$ recognizing the reverse of $L_{n,k}$.
Applying the subset construction to $\mathcal N$ yields a DFA $\mathcal{D}^R$ whose states are subsets of $Q_{n-1}$, with initial state $F_{n,k}$ and final states $\{U \subseteq Q_{n-1} | 0 \in U\}$. We show that $\mathcal{D}^R$ is minimal, and hence the reverse of $L_{n,k}$ has complexity $2^{n-1}$. We check reachability and distinguishability.

Recall from Lemma~\ref{lem:properpermutations} that $a$ and $b$ generate all permutations of $E_{n,k}$ and $F_{n,k}$ in $\mathcal{D}_{n,k}$ and, although the transitions are reversed in $\mathcal{D}^R$, they still generate all such permutations. Let $u_1, u_2 \in \{a,b\}^*$ be such that $u_1$ induces $(0, \dots, n-2-k)$ and $u_2$ induces $(n-1-k, \dots, n-2)$ in $\mathcal{D}^R$.

Consider a state $U = \{q_1, \dots, q_h, n-1-k, \dots, n-2\}$ where $0 \le q_1 < q_2 < \dots < q_h \le n-2-k$. 
If $h = 0$, then $U$ is the initial state. When $h \ge 1$, $\{q_2 - q_1, q_3 - q_1, \dots, q_h - q_1, n-1-k, \dots, n-2\}eu_1^{q_1} = U$. By induction, all such states are reachable.

Now we show that any state $U = \{q_1, \dots, q_h, p_1, \dots, p_i\}$ where $0 \le q_1 < q_2 < \dots < q_h \le n-2-k$ and $n-1-k \le p_1 < p_2 < \dots < p_i \le n-2$ is reachable. 
If $i = k$, then $U = \{q_1, \dots, q_h, n-1-k, \dots, n-2\}$ is reachable by the argument above. 
When $0 \le i < k$, choose $p \in F_{n,k} \setminus U$ and see that $U$ is reached from $U \cup \{p\}$ by $u_2^{n-1-p}d_2u_2^{p - (n-1-k)}$. By induction, every state is reachable.

To prove distinguishability, consider distinct states $U$ and $V$.
Choose $q \in U \oplus V$. If $q \in E_{n,k}$, then $U$ and $V$ are distinguished by $u_1^{n-1-k-q}$. When $q \in F_{n,k}$, they are distinguished by $u_2^{n -1-q} e$. Hence, the states are pairwise distinguishable and $\mathcal{D}^R$ is minimal.
\end{proof}

\begin{proposition}[Atomic Complexity]\label{prop:properatoms}
For each atom $A_S$ of $L_{n,k}(\Sigma)$, write $S = X_1 \cup X_2$, where $X_1 \subseteq E_{n,k}$ and $X_2 \subseteq F_{n,k}$.
Let $\overline{X_1} = E_{n,k}\setminus X_1$ and $\overline{X_2} = F_{n,k}\setminus X_2$.
If $X_2 \not= \emptyset$, then $\kappa(A_S) =$
$$1 + \sum_{x_1=0}^{|X_1|}\sum_{x_2=1}^{|X_1| + |X_2| - x_1}\sum_{y_1=0}^{|\overline{X_1}|}\sum_{y_2=0}^{|\overline{X_1}|+|\overline{X_2}| -y_1}\binom{n-1-k}{x_1}\binom{k}{x_2}\binom{n-1-k-x_1}{y_1}\binom{k-x_2}{y_2}.$$
If $X_1 \not= \emptyset$ and $X_2 = \emptyset$, then $\kappa(A_S) =$
\begin{multline*}
1 + \sum_{x_1=0}^{|X_1|}\sum_{x_2=0}^{|X_1| - x_1}\sum_{y_1=0}^{|\overline{X_1}|}\sum_{y_2=0}^{k}\binom{n-1-k}{x_1}\binom{k}{x_2}\binom{n-1-k-x_1}{y_1}\binom{k-x_2}{y_2} \\
-2^k\sum_{y=0}^{|\overline{X_1}|}\binom{n-1-k}{y}.
\end{multline*}
Otherwise, $S = \emptyset$ and $\kappa(A_S) = 2^{n-1}$.
The atomic complexity is maximal for $k$-proper languages.
\end{proposition}

\begin{proof}
Let $L$ be a $k$-proper language with quotients $K_0, K_1, \dots, K_{n-1}$ where $K_0, \dots, K_{n-2-k}$ are non-final quotients, $K_{n-1-k}, \dots, K_{n-2}$ are final quotients, and $K_{n-1} = \emptyset$. For $S \subseteq Q_{n-1}$, we have $A_S = \bigcap_{i \in S} K_i \cap \bigcap_{i \in \overline{S}} \overline{K_i}$; note $n-1 \not\in S$ since $A_S$ must be non-empty.

The quotients are $w^{-1}A_S = \bigcap_{i \in S} w^{-1}K_i \cap \bigcap_{i \in \overline{S}} \overline{w^{-1}K_i}$. However $w^{-1}K_i$ is always another quotient $K_j$.
Thus $w^{-1}A_S$ has the form $J_{X,Y} = \bigcap_{i \in X} K_i \cap \bigcap_{i \in Y} \overline{K_i}$ where $X = \{i \mid K_i = w^{-1}K_j \text{ for some $j \in S$}\}$ and $Y = \{i \mid K_i = w^{-1}K_j \text{ for some $j \in \overline{S}$}\}$. For brevity, we write $S \xrightarrow{w} X$ and $\overline{S} \xrightarrow{w} Y$; this notation is in agreement with the action of $w$ on the states of $\mathcal{D}_{n,k}$ corresponding to $S$ and $\overline{S}$.

To establish the upper bound, we just count the number of possible distinct $J_{X,Y}$ for each value of $S$.
Notice $n-1 \in U$ and if $X \cap Y \not=\emptyset$ then $J_{X,Y}$ is the empty quotient.
Write $S = X_1 \cup X_2$ where $X_1 \subseteq E_{n,k}$ and $X_2 \subseteq F_{n,k}$, and let $\overline{X_1} = E_{n,k}\setminus X_1$ and $\overline{X_2} = F_{n,k}\setminus X_2$.
By Proposition~\ref{prefixchar} every word $w$ maps $X_1$ to a subset of $Q_n$ and $X_2$ to a subset of $F_{n,k}\cup \{n-1\}$.
Similarly, $w$ maps $\overline{X_1}$ to a subset of $Q_n$, $\overline{X_2}$ to a subset of $F_{n,k} \cup \{n-1\}$, and $n-1$ to itself.

The number of disjoint sets $X, Y \subseteq Q_n$ that are reached from $S$ and $\overline{S}$ respectively by some word $w$ is an upper bound on the number of non-empty quotients of $A_S$. Specifically, we require $n-1 \in Y$, $|X| \le |S|$, $|Y| \le |\overline{S}|$, $|X \cap E_{n,k}| \le |X_1|$, $|Y \cap E_{n,k}| \le |\overline{X_1}|$, and $X \cap Y= \emptyset$.
Thus we have the initial estimate
$$\sum_{x_1=0}^{|X_1|}\sum_{x_2=0}^{|X_1| - x_1}\sum_{y_1=0}^{|\overline{X_1}|}\sum_{y_2=0}^{k}\binom{n-1-k}{x_1}\binom{k}{x_2}\binom{n-1-k-x_1}{y_1}\binom{k-x_2}{y_2},$$
 where $x_1$ counts $|X \cap E_{n,k}|$, $x_2$ counts $|X \cap F_{n,k}|$, $y_1$ counts $|Y \cap E_{n,k}|$, and $y_2$ counts $|Y \cap F_{n,k}|$.
With some refinements, this estimate leads to the three cases in the statement.
 
Note if $S \not= \emptyset$ then $X \not= \emptyset$. Also, if $X_2 \not= \emptyset$, then any non-empty quotient $J_{X,Y}$ must have $X \cap F_{n,k} \not= \emptyset$ since $X_2$ cannot be mapped to $n-1$; this has the effect that $x_2$ cannot be 0 in our initial estimate.
Since the estimate only counts non-empty quotients, we add 1 to represent the empty quotient achieved when $X$ and $Y$ intersect, yielding the expression for the case $X_2 \not= \emptyset$.

If $X_1 \not= \emptyset$ and $X_2 = \emptyset$, then we cannot have $x_1 = x_2 = 0$ since that would correspond to $X = \emptyset$; the subtracted term in the statement is the value of the estimate when $x_1=x_2=0$. As before, add 1 for the empty quotient.

Finally, if $S=\emptyset$, then $X =\emptyset$ and $Y \subseteq Q_n$ with $n-1 \in Y$. There are $2^{n-1}$ possible values of $Y$, hence $\kappa(A_S) \le 2^{n-1}$. 
Since $X$ and $Y$ cannot intersect in this case, there is no need to add 1 for an empty quotient, as in the previous two cases.
Thus we have the three cases in the statement.

It remains to prove that $L_{n,k}(\Sigma)$ of Definition~\ref{def:proper} meets this upper bound.
Let the quotient $K_q$ of $L_{n,k}$ be the language accepted by state $q$ in $\mathcal{D}_{n,k}$.
We must show that every $J_{X,Y}$ can be reached from $A_S$ by some word in $\Sigma^*$, and that every non-empty $J_{X,Y}$ is distinct from $J_{X',Y'}$ whenever $(X,Y) \not= (X', Y')$.
By Proposition~\ref{prop:propersyntactic}, the syntactic semigroup is as large as possible for $k$-proper languages.
Hence, whenever $n-1 \in Y$, $|X| \le |S|$, $|Y| \le |\overline{S}|$, $|X \cap E_{n,k}| \le |X_1|$, and $|Y \cap E_{n,k}| \le |\overline{X_1}|$, there is a word $w \in \Sigma^*$ such that $S \xrightarrow{w} X$ and $\overline{S} \xrightarrow{w} Y$.
Thus each quotient $J_{X,Y}$ counted by the upper bound is reachable in $A_S$.

Consider $J_{X,Y}$ where $X \cap Y = \emptyset$ and $n-1 \in Y$. If $X \not= \emptyset$ then there exists $w\in \Sigma^*$ such that $X \xrightarrow{w} \{n-2\}$ and $Y \xrightarrow{w} \{n-1\}$; hence $w \in J_{X,Y}$ since $\varepsilon \in K_{n-2}$. If $X = \emptyset$ choose $w$ such that $Y \xrightarrow{w} \{n-1\}$; hence $w \in J_{X,Y}$. Thus $J_{X,Y}$ is non-empty.

Now take $J_{X',Y'}$ where $(X,Y) \not= (X', Y')$, $X' \cap Y' = \emptyset$ and $n-1 \in Y'$.
We must show that $J_{X,Y}$ and $J_{X',Y'}$ are distinct.
Note $w^{-1}J_{X,Y} = J_{Xw,Yw}$ where $X \xrightarrow{w} Xw$ and $Y \xrightarrow{w} Yw$.
If $r \in X' \setminus X$, then choose $w$ that maps $r \to n-1$ in $\mathcal{D}_{n,k}$; $J_{Xw,Yw}$ is non-empty, since $Xw \cap Yw = \emptyset$, and $J_{X'w,Y'w} = \emptyset$ since $n-1 \in X'w$.
Similarly, if $X = X'$ and $r \in Y' \setminus Y$, then choose $w$ that maps $X \cup \{r\} \to \{n-2\}$ and $Q_n \setminus (X \cup \{r\}) \to \{n-1\}$. Then $J_{Xw,Yw} = J_{\{n-2\}, \{n-1\}}$ is non-empty and $J_{X'w,Y'w} = J_{\{n-2\}, \{n-2, n-1\}} = \emptyset$. Finally, if $X = X' = \emptyset$ and $r \in Y' \setminus Y$, then distinguish $J_{X,Y}$ and $J_{X',Y'}$ by a word that sends $r \to n-2$ and $Q_n\setminus\{r\} \to \{n-1\}$. Hence, $J_{X,Y}$ and $J_{X',Y'}$ are distinct in all cases. Therefore, the quotients of $A_S$ counted in the upper bound are pairwise distinct and $L_{n,k}$ has maximal atomic complexity.
\end{proof}

\begin{proposition}[Star Bound]\label{prop:properstarbound}
Suppose $L$ is a language with $n$ quotients including an empty quotient. Let $k \ge 1$ be the number of final quotients. Then $\kappa(L^*) \le 2^{n-2} + 2^{n-2-k} + 1$. This bound also holds for all prefix-convex languages when $n \ge 3$.
\end{proposition}
\begin{proof}
If $n=1$, then $L = \emptyset$ has an empty quotient, and obeys the upper bound. Assume $n \ge 2$.  Let $\mathcal D = (Q_n, \{a,b, -, -,d_1,d_2,e\}, \delta, 0, F)$ be a minimal DFA recognizing $L$. Since $L$ has an empty quotient, let $n-1$ be the empty state of $\mathcal{D}$. To obtain an $\varepsilon$-NFA for $L^*$, we add a new initial state $0'$ which is final and has the same transitions as $0$. We then add an $\varepsilon$-transition from every state in $F$ to $0$. Applying the subset construction to this $\varepsilon$-NFA yields a DFA $\mathcal{D}' = (Q', \{a,b,d_1,d_2,e\}, \delta', \{0'\}, F')$ recognizing $L^*$, in which the states are non-empty subsets of $Q_n \cup \{0'\}$.

Many of these states are unreachable or indistinguishable from other states.
First, since there is no transition in the $\varepsilon$-NFA to $0'$, the only reachable state in $Q'$ containing $0'$ is $\{0'\}$.
Second, because of the $\varepsilon$-transitions, any reachable final state $U \not= \{0'\}$ must have $0 \in U$.

Finally, for any $U \in Q'$, we have $U \cup \{n-1\} \in F'$ if and only if $U \in F'$. Furthermore, $\delta'(U \cup \{n-1\}, w) = \delta'(U, w) \cup \{n-1\}$ for all $w \in \Sigma^*$. Therefore, the states $U$ and $U \cup \{n-1\}$ are equivalent in $D'$. This allows us to deal with subsets of $Q_{n-1}$ instead of $Q_n$. Combined with the first two reductions, this proves that $\mathcal D'$ is equivalent to a reduced DFA with the states $\{\{0'\}\} \cup \{U \subseteq Q_{n-1} | U \cap F = \emptyset\} \cup \{U \subseteq Q_{n-1} | \text{$0 \in U$ and $U \cap F \not= \emptyset$}\}$. This DFA has $1 + 2^{n-k-1} + (2^{n-2} - 2^{n-k-2}) = 2^{n-2} + 2^{n-2-k} + 1$ states. Thus, $\kappa(L^*) \le 2^{n-2} + 2^{n-2-k} + 1$.

This bound also applies when $L$ is any prefix-convex language and $n \ge 3$: by Lemma~\ref{idealorempty}, $L$ is either a right ideal or has an empty state. If $L$ is a right ideal, then $\kappa(L^*) \le n+1$, which is at most $2^{n-2}+2^{n-2-k}+1$ for $n \ge 3$.
\end{proof}

\begin{proposition}[Star]\label{prop:properstar}
The star of $L_{n,k}(a,b,-,-,d_1,d_2, e)$ has complexity $2^{n-2} + 2^{n-2-k} + 1$ for $n \ge 3$ and $1 \le k \le n-2$.
\end{proposition}
\begin{proof}
Let $\mathcal{D}_{n,k} = (Q_n, \Sigma, \delta_{n,k}, 0, F_{n,k})$ denote the DFA $\mathcal{D}_{n,k}(a,b,-,-,d_1,d_2,e)$ of Definition~\ref{def:proper} and let $L_{n,k} = L(D_{n,k})$. We apply the same construction and reduction as in Proposition~\ref{prop:properstarbound} to obtain a DFA $\mathcal{D}_{n,k}'$ recognizing $L_{n,k}^*$ with states $Q' = \{\{0'\}\} \cup \{U \subseteq E_{n,k}\} \cup \{U \subseteq Q_{n-1} | 0 \in U \text{ and } U \cap F_{n,k} \not= \emptyset\}$. It is clear that $|Q'| = 2^{n-2} + 2^{n-2-k} + 1$. We must prove that $\mathcal D_{n,k}'$ is minimal. To do so, we show that every state in $Q'$ is reachable and that all distinct states are distinguishable.

By Lemma~\ref{lem:properpermutations}, $a$ and $b$ generate all permutations of $E_{n,k}$ and $F_{n,k}$ in $\mathcal{D}_{n,k}$. Let $u_1, u_2 \in \{a,b\}^*$ be such that $u_1$ induces $(0, \dots, n-2-k)$ and $u_2$ induces $(n-1-k, \dots, n-2)$ in $\mathcal{D}_{n,k}$.

For reachability, we consider three cases.
\begin{enumerate}
\item State $\{0'\}$ is reachable by $\varepsilon$.
\item Let $U \subseteq E_{n,k}$. For any $q \in E_{n,k}$, we can reach $U \setminus \{q\}$ by $u_1^{n-2-k-q}d_1u_1^q$; hence if $U$ is reachable, then every subset of $U$ is reachable. Observe that state $E_{n,k}$ is reachable by $eu_1^{n-2-k}d_2^k$, and we can reach any subset of this state. Therefore, all non-final states are reachable.
\item If $U \cap F_{n,k} \not= \emptyset$, then $U = \{0, q_1, q_2, \dots, q_h, r_1, \dots, r_i\}$ where $0 < q_1 < \dots < q_h \le n-2-k$ and $n-1-k \le  r_1 < \dots < r_i < n-1$ and $i \ge 1$. 
We prove that $U$ is reachable by induction on $i$. If $i = 0$, then $U$ is reachable by {\bf 2}.
For any $i \ge 1$, we can reach $U$ from $\{0, q_1, \dots, q_h, r_2 - (r_1 - (n-1-k)), \dots, r_i - (r_1 - (n-1-k)) \}$ by $eu_2^{r_1 - (n-1-k)}$. Therefore, all states of this form are reachable.
\end{enumerate}

Now we show that the states are pairwise distinguishable.
\begin{enumerate}
\item The initial state $\{0'\}$ is distinguishable from any other final state $U$ since $\{0'\}u_1$ is non-final and $Uu_1$ is final.
\item Suppose $U$ and $V$ are distinct subsets of $E_{n,k}$. There is some $q \in U \oplus V$. We distinguish $U$ and $V$ by $u_1^{n-1-k-q}e$.
\item Suppose $U$ and $V$ are distinct and final and neither one is $\{0'\}$. 
There is some $q \in U \oplus V$. If $q \in E_{n,k}$, then $Ud_2^k = U \setminus F_{n,k}$ and $Vd_2^k = V \setminus F_{n,k}$ are distinct, non-final states as in {\bf 2}. Otherwise, $q \in F_{n,k}$ and we distinguish $U$ and $V$ by $u_2^{n-1-q}d_2^{k-1}$.
\end{enumerate}
Therefore, $\mathcal D_{n,k}'$ is minimal and $\kappa(L_{n,k}^*) = 2^{n-2} + 2^{n-2-k} + 1$.
\end{proof}

\begin{proposition}[Product Bound]\label{prop:properproductbound}
Let $L'$ and $L$ be proper prefix-convex languages with $m \ge 2$ and $n \ge 2$ quotients respectively. Let $j$ be the number of final quotients in $L'$ and $k$ be the number of final quotients in $L$. The product of $L'$ and $L$ has complexity at most $m - 1- j + j2^{n-2}+2^{n-1}$.
\end{proposition}

\begin{proof}
Let $\mathcal{D}'$ and $\mathcal{D}$ be minimal DFAs for $L'$ and $L$ respectively. Let $Q_n$ be the state set of $\mathcal{D}$. To avoid confusing states in the product, take $\mathcal{D}'$ to have the states $Q_m' = \{0', 1', \dots, (m-1)'\}$.

Applying the standard DFA construction for the product $L'L$, we obtain states of the form $\{p'\} \cup S$ for $p' \in Q_m'$ and $S \subseteq Q_n$. We prove the upper bound by determining reachability and distinguishability for these states. By Lemma~\ref{idealorempty}, each of $L'$ and $L$ has an empty quotient.

Assume $(m-1)'$ is the empty state in $\mathcal{D}'$ and $n-1$ is the empty state in $\mathcal{D}$. Take $\{0', \dots, (m-2-j)'\}$ to be the non-final, non-empty states and $\{(m-1-j)', \dots, (m-2)'\}$ to be the final states in $\mathcal{D}'$. By Proposition~\ref{prefixchar}, no non-final, non-empty state can be reached from a final state in $\mathcal{D}'$. Therefore, every reachable state in the product is either $\{p'\}$ for $0 \le p \le m-2-j$, $\{p',0\} \cup S$ for $m-1-j \le p \le m-2$ and $S \subseteq Q_{n} \setminus \{0\}$, or  $\{(m-1)'\} \cup S$ for $S \subseteq Q_{n}$.
This gives the preliminary estimate of $m-1-j + j2^{n-1} + 2^n$ states. For $m-1-j \le p \le m-1$, $\{p'\} \cup S$ is equivalent to $\{p', n-1\} \cup S$ for any $S \subseteq Q_n$. This removes $j2^{n-2} + 2^{n-1}$ states from our estimate, leaving $m-1-j + j2^{n-2} + 2^{n-1}$ states.
\end{proof}

\begin{proposition}[Product]\label{prop:properproduct}
For $m,n \ge 3$, $1 \le j \le m-2$, and $1 \le k \le n-2$, the product of $L'_{m,j}(a,b,c_1,-,d_1, d_2, e)$ and $L_{n,k}(a,d_2, c_1, -, d_1, b, e)$ has complexity $m - 1 - j + j 2^{n-2} + 2^{n-1}$.
\end{proposition}
\begin{proof}
Let $\mathcal{D}_{m,j}'$ and $\mathcal{D}_{n,k}$ be the DFAs of Definition~\ref{def:proper} for languages $L'_{m,j}(a,b,c_1,-, d_1, d_2, e)$ and $L_{n,k}(a,d_2, c_1, -, d_1, b, e)$ respectively. As before, take $\mathcal{D}_{m,j}'$ to have the states $Q_m' = \{0', 1', \dots, (m-1)'\}$ and let $E_{n,k}' = \{0', \dots, (m-2-j)'\}$. Using the standard construction of the $\varepsilon$-NFA $\mathcal N$ for the product, we delete the empty state $n-1$, change the final states of $\mathcal{D}_{m,j}'$ to non-final states, and add $\varepsilon$-transitions from each final state of $\mathcal{D}_{m,j}'$ to the initial state of $\mathcal{D}_{n,k}$.

The subset construction on $\mathcal N$ yields states of the form $\{p'\} \cup S$, where $p' \in Q_m'$ and $S \subseteq Q_{n-1}$. However, some of these sets are not reachable in the product: if $p' \in E_{m,j}'$ then we must have $S = \emptyset$, and if $p' \in F_{m,j}'$ then $0 \in S$ because of the $\varepsilon$-transitions in $\mathcal N$.

Thus, we have the states $\{p'\}$ for $p' \in E_{m,j}'$, $\{p', 0\} \cup S$ for $p' \in F_{m,j}'$ and $S \subseteq Q_{n-1} \setminus \{0\}$, and $\{(m-1)'\} \cup S$ for $S \subseteq Q_{n-1}$. This totals to $(m-1-j) + (j2^{n-2}) + (2^{n-1}) = m-1-j+j2^{n-2}+2^{n-1}$ different states. We show that they are reachable and pairwise distinguishable.

State $\{p'\}$ is reached by $d_1^p$ for all $p' \in E_{m,j}'$. State $\{(m-1-j)', 0\}$ is reached by $e$. For $m-j \le p \le m-1$ we have $\{(m-1-j)',0\}\xrightarrow{d_2^{p-(m-1-j)}} \begin{cases} \{p', 0, 1\} &\text{if }n-1-k \ge 2 \\ \{p', 0\} &\text{if } n-1-k = 1 \end{cases} \xrightarrow{c_1} \{p', 0\}$.

Now consider states of the form $\{p', 0\} \cup T$ where $p' \in F_{m,j}'$ and $T \subseteq F_{n,k}$. These states are reachable when $T = \emptyset$. Inductively assume the states are reachable when $|T| < i$ for some $i \ge 1$. Let $T_i = \{r_1, r_2, \dots, r_i\}$ where $n-1-k \le r_1 < r_2 < \dots < r_i \le n-2$, and let $T_{i-1} = \{r_2-(r_1-(n-1-k)), \dots, r_i-(r_1-(n-1-k))\}$.
Then $\{0\} \cup T_{i-1} \xrightarrow{e} \{n-1-k\} \cup T_{i-1} \xrightarrow{b^{r_1-(n-1-k)}} T_i$. Notice $b$ induces a permutation on $\mathcal{D}_{m,j}'$, so for any $p' \in F_{m,j}'$ there is a state $q' \in F_{m,j}'$ such that $q' \xrightarrow{eb^{r_1-(n-1-k)}} p'$. Thus, $\{p', 0\} \cup T_i$ is reachable from $\{q',0\} \cup T_{i-1}$.

Extend this to states of the form $\{p', 0\} \cup S \cup T$, where $p' \in  F_{m,j}'$, $S \subseteq E_{n,k} \setminus \{0\}$, and $T \subseteq F_{n,k}$. These states are reachable when $S = \emptyset$. Inductively assume the states are reachable when $|S| < h$ for some $h \ge 1$. Let $S_h = \{q_1, q_2, \dots, q_h\}$ where $n-1-k \le q_1 < q_2 < \dots < q_i \le n-2$, and let $S_{h-1} = \{q_2-q_1, \dots, q_h-q_1\}$. Then $\{p', 0\} \cup S_{h-1} \cup T \xrightarrow{d_1} \{p', 0, 1\} \cup (S_{h-1}+1) \cup T \xrightarrow{(d_1c_1)^{q_1-1}} \{p',0, q_1\} \cup (S_{h-1}+q_1) \cup T = \{p',0\} \cup S_h \cup T$. In the last derivation, $S+c$ denotes the set $\{q+c \mid q \in S\}$.

State $\{(m-1)', 0\} \cup S \cup T$ is reachable from $\{(m-2)', 0\} \cup S \cup T$ by $d_2^\ell$, where $\ell > 0$ is the order of $d_2$ in $\mathcal{D}_{n,k}$ (i.e. $d_2^\ell$ induces the identity transformation on $\mathcal{D}_{n,k}$).

Finally, state $\{(m-1)'\} \cup S \cup T$ is reachable from $\{(m-1)', 0\} \cup S \cup T$: by Lemma~\ref{lem:properpermutations}, the permutation $(0, 1, \dots, n-2-k)$ of $\mathcal{D}_{n,k}$ is generated by some $u_1 \in\{a, d_2\}^*$, and $\{(m-1)', 0\} \cup S \cup T \xrightarrow{u_1^{n-2-k}} \{(m-1)', n-2-k\} \cup (S-1) \cup T \xrightarrow{d_1} \{(m-1)'\} \cup S \cup T$. Thus all states are reachable.

We now check distinguishability in cases. Using Lemma~\ref{lem:properpermutations}, take words $u_1, u_2 \in \{a, d_2\}^*$ such that $u_1$ induces $(0, 1, \dots, n-2-k)$ and $u_2$ induces $(n-1-k, n-k, \dots, n-2)$ on $\mathcal{D}_{n,k}$. Note $u_1$ and $u_2$ act on $\mathcal{D}_{m,j}'$ as well. 
\begin{enumerate}
\item Let $U = \{(m-1)'\}$ and let $V$ be any other state. Notice $U$ is the empty state. We show that $V$ is non-empty.
\begin{enumerate}
\item If $q \in V \cap Q_{n-1}$ then by the minimality of $\mathcal{D}_{n,k}$ there is a word $w$ such that $qw \in F_{n,k}$; hence $Vw$ is final.
\item Otherwise $V = \{p'\}$ for some $p' \in E_{m,j}'$. There is a word $w$ such that $p'w \in F_{m,j}'$; hence $0 \in Vw$ and this reduces to {\bf \emph a}.
\end{enumerate}
\item Let $U = \{p'\}$ and $V = \{q'\}$ where $p',q' \in E_{m,j}'$ and $p < q$. Then $Vd_1^{m-1-j-q} = \{(m-1)'\}$ and $Ud_1^{m-1-j-q}$ is non-empty by {\bf 1}.
\item Let $U = \{p'\}$ and $V = \{q', 0\} \cup S$ where $p' \in E_{m,j}'$, $q' \in F_{m,j}'$, and $S \subseteq Q_{n-1} \setminus \{0\}$. Then $U$ and $V$ are distinguished by $e$.
\item Let $U = \{p'\}$ and $V = \{(m-1)'\} \cup S$ where $p' \in E_{m,j}'$ and $S \subseteq Q_{n-1}$. If $S = \emptyset$ this reduces to {\bf 1}. If $S \cap F_{n,k} \not= \emptyset$ then $V$ is final. Otherwise there is some $r \in S$, and $Vu_1^{n-1-k-r}e$ is final. Notice $Uu_1^{n-1-k-r}e$ is non-final because $u_1 \in \{a,d_2\}^*$.
\item Let $U = \{(m-1)'\} \cup S$ and $V = \{(m-1)'\} \cup T$ where $S \not= T \subseteq Q_{n-1}$; pick $r \in S \oplus T$. Without loss of generality, say $r \in S \setminus T$.
\begin{enumerate}
\item If $r = 0$, then $U \xrightarrow{b^k} U \setminus F_{n,k} \xrightarrow{e} U \setminus (\{0\} \cup F_{n,k}) \cup \{n-1-k\}$ and $V \xrightarrow{b^k} V \setminus F_{n,k}  \xrightarrow{e} V \setminus F_{n,k}$.
\item If $r \in E_{n,k}$, then we reduce to {\bf \emph a} by applying $u_1^{n-1-k-r}$.
\item If $r = n-1-k$, then $Ub^{k-1}$ is final and $Vb^{k-1}$ is non-final.
\item If $r \in F_{n,k}$, then we reduce to {\bf \emph c} by applying $u_2^{n-1-r}$.
\end{enumerate}
\item Let $U = \{p', 0\} \cup S$ and $V = \{(m-1)'\} \cup T$ where $p' \in F_{m,j}'$, and $S, T \subseteq Q_{n-1}$. Notice $Ud_1^{n-1-k}b^k$ is non-empty since $p'$ is not mapped to $(m-1)'$, but $V \xrightarrow{d_1^{n-1-k}} \{(m-1)'\} \cup T \setminus E_{n,k} \xrightarrow{b^k} \{(m-1)'\}$; this reduces to {\bf 1}.
\item Let $U = \{p', 0\} \cup S$ and $V = \{q', 0\} \cup T$ where $p', q' \in F_{m,j}'$, $p < q$, and $S, T \subseteq Q_{n-1}$. Reduce to {\bf 6} by applying $d_2^{m-1-q}$.
\item Let $U = \{p', 0\} \cup S$ and $V = \{p', 0\} \cup T$ where $p' \in F_{m,j}'$ and $S \not= T \subseteq Q_{n-1}$. Pick $r \in S \oplus T$ and assume without loss of generality that $r \in S \setminus T$.
\begin{enumerate}
\item If $r \ge 2$, then $d_2^{m-1-p}$ fixes $r$ and maps $p'$ to $(m-1)'$; hence this reduces to {\bf 5}.
\item If $p = m-2$, then apply $d_2$ to reduce to {\bf 5}. Notice $Sd_2$ and $Td_2$ are distinct since $d_2$ induces a permutation on $\mathcal{D}_{n,k}$.
\item If $r = 1$ and $n-1-k \ge 2$, then applying $d_1$ reduces to {\bf \emph a}.
\item If $r = 1$ and $n-1-k = 2$, then observe that $a$ and $b$ both fix $1$ in $\mathcal{D}_{n,k}$. By Lemma~\ref{lem:properpermutations}, there is a word $w \in \{a,b\}^*$ such that $p'w = (m-2)'$. Since $n-1-k=2$, $a$ and $b$ do not alter $E_{n,k}$. Hence $1 \in Sw$ and $1 \not\in Tw$, so this reduces to {\bf \emph b}.
\end{enumerate}
\end{enumerate}
\end{proof}

\begin{proposition}[Boolean Operations]\label{prop:properboolean} 
For $m,n\ge 3$, $1 \le j \le m-2$, and $1 \le k \le n-2$, let $L'_{m,j} = L'_{m,j}(a, b, c_1, -, d_1, d_2, e)$ and $L_{n,k} = L_{n,k}(a, b, e, -, d_2, d_1, c_1)$ of Definition~\ref{def:proper}. For any proper binary boolean function $\circ$, the complexity of $L'_{m,j} \circ L_{n,k}$ is maximal. In particular,
	\begin{enumerate}
	\item
	$\kappa(L'_{m,j} \cup L_{n,k}) = \kappa (L'_{m,j} \oplus L_{n,k}) = mn$.
	\item
	$\kappa(L'_{m,j} \setminus L_{n,k}) = mn-(n-1)$.
	\item
	$\kappa(L'_{m,j} \cap L_{n,k}) = mn-(m+n-2)$.
	\end{enumerate}
\end{proposition}

\begin{proof}
Let $\mathcal{D}_{m,j}'$ and $\mathcal{D}_{n,k}$ be the DFAs of Definition~\ref{def:proper} for languages $L'_{m,j}(a, b, c_1, -, d_1, d_2, e)$ and $L_{n,k}(a, b, e, -, d_2, d_1, c_1)$ respectively. As before, take $\mathcal{D}_{m,j}'$ to have the states $Q_m' = \{0', 1', \dots, (m-1)'\}$.
There is a standard construction for $L'_{m,j} \circ L_{n,k}$ for any boolean set operation $\circ$ in terms of the direct product. The direct product of $\mathcal{D}_{m,j}'$ and $\mathcal{D}_{n,k}$ has states $Q_m' \times Q_n$, initial state $(0',0)$, and transition function $\delta$ such that $\delta((p',q),w) = (\delta_{m,j}'(p',w), \delta_{n,k}(q,w))$. If we set the final states to be $(F_{m,j}' \times Q_n) \circ (Q_m' \times F_{n,k})$, it is a DFA recognizing $L'_{m,j} \circ L_{n,k}$.
For each $\circ \in \{\cup, \oplus, \setminus, \cap\}$, we construct the DFA $\mathcal{D}_\circ$ to recognize $L'_{m,j} \circ L_{n,k}$. All four DFAs have the same states and transitions as the direct product and will only differ in the set of final states.
The DFA $\mathcal D_\oplus$ for symmetric difference is shown in Figure~\ref{fig:properboolean}.

We can usefully partition the states of the direct product. Let $W = E_{m,j}' \times E_{n,k}$, $X = E_{m,j}' \times F_{n,k}$, $Y = F_{m,j}' \times E_{n,k}$, $Z = F_{m,j}' \times F_{n,k}$, and $S = W \cup X \cup Y \cup Z$. Let $R = \{(m-1)'\} \times Q_n$ and $C = Q_m' \times \{n-1\}$.

\begin{figure}[ht]
\unitlength 8.5pt
\begin{center}\begin{picture}(35,23)(0,-6)
\gasset{Nh=2.6,Nw=2.6,Nmr=1.2,ELdist=0.3,loopdiam=1.2}
	{\scriptsize
\node(0'0)(2,15){$0',0$}\imark(0'0)
\node(1'0)(2,10){$1',0$}
\node(2'0)(2,5){$2',0$}\rmark(2'0)
\node(3'0)(2,0){$3',0$}\rmark(3'0)
\node(4'0)(2,-5){$4',0$}

\node(0'1)(10,15){$0',1$}
\node(1'1)(10,10){$1',1$}
\node(2'1)(10,5){$2',1$}\rmark(2'1)
\node(3'1)(10,0){$3',1$}\rmark(3'1)
\node(4'1)(10,-5){$4',1$}

\node(0'2)(18,15){$0',2$}\rmark(0'2)
\node(1'2)(18,10){$1',2$}\rmark(1'2)
\node(2'2)(18,5){$2',2$}
\node(3'2)(18,0){$3',2$}
\node(4'2)(18,-5){$4',2$}\rmark(4'2)

\node(0'3)(26,15){$0',3$}\rmark(0'3)
\node(1'3)(26,10){$1',3$}\rmark(1'3)
\node(2'3)(26,5){$2',3$}
\node(3'3)(26,0){$3',3$}
\node(4'3)(26,-5){$4',3$}\rmark(4'3)

\node(0'4)(34,15){$0',4$}
\node(1'4)(34,10){$1',4$}
\node(2'4)(34,5){$2',4$}\rmark(2'4)
\node(3'4)(34,0){$3',4$}\rmark(3'4)
\node(4'4)(34,-5){$4',4$}
	}

\drawedge[curvedepth=-2,ELdist=-1](0'0,2'0){$e$}
\drawedge[curvedepth=2,ELdist=0.2](0'0,0'2){$c_1$}

\drawedge(0'0,0'1){$d_2$}
\drawedge(1'0,1'1){$d_2$}
\drawedge(0'0,1'0){$d_1$}
\drawedge(0'1,1'1){$d_1$}

\drawedge(0'3,1'4){$d_1$}
\drawedge(3'0,4'1){$d_2$}

\drawedge[curvedepth=-2,ELdist=-1](1'0,4'0){$d_1$}
\drawedge[curvedepth=-2,ELdist=-1](1'1,4'1){$d_1$}
\drawedge[curvedepth=2,ELdist=0.2](0'1,0'4){$d_2$}
\drawedge[curvedepth=2,ELdist=0.2](1'1,1'4){$d_2$}

\drawedge(2'3,2'4){$d_1$}
\drawedge(3'3,3'4){$d_1$}

\drawedge(0'4,1'4){$d_1$}
\drawedge[curvedepth=2,ELdist=0.2](1'4,4'4){$d_1$}
\drawedge[ELdist=-1.2](2'4,3'4){$d_2$}
\drawedge[ELdist=-1.2](3'4,4'4){$d_2$}

\drawedge(3'2,4'2){$d_2$}
\drawedge(3'3,4'3){$d_2$}

\drawedge(4'0,4'1){$d_2$}
\drawedge[curvedepth=-2,ELdist=-1](4'1,4'4){$d_2$}
\drawedge(4'2,4'3){$d_1$}
\drawedge(4'3,4'4){$d_1$}

\drawedge(2'2,3'2){$d_2$}
\drawedge(2'3,3'3){$d_2$}
\drawedge(2'2,2'3){$d_1$}
\drawedge(3'2,3'3){$d_1$}

\drawedge[curvedepth=-2,ELdist=-1](0'2,2'2){$e$}
\drawedge[curvedepth=-2,ELdist=-1](0'3,2'3){$e$}
\drawedge[curvedepth=2,ELdist=0.2](2'0,2'2){$c_1$}
\drawedge[curvedepth=2,ELdist=0.2](3'0,3'2){$c_1$}
\end{picture}\end{center}
\caption{DFA $\mathcal D_\oplus$ for symmetric difference of proper languages with DFAs $\mathcal D_{5,2}'(a,b,c_1,-, d_1, d_2, e)$ and $\mathcal D_{5,2}(a,b,e,-, d_2, d_1, c_1)$ shown partially.}
\label{fig:properboolean}
\end{figure}
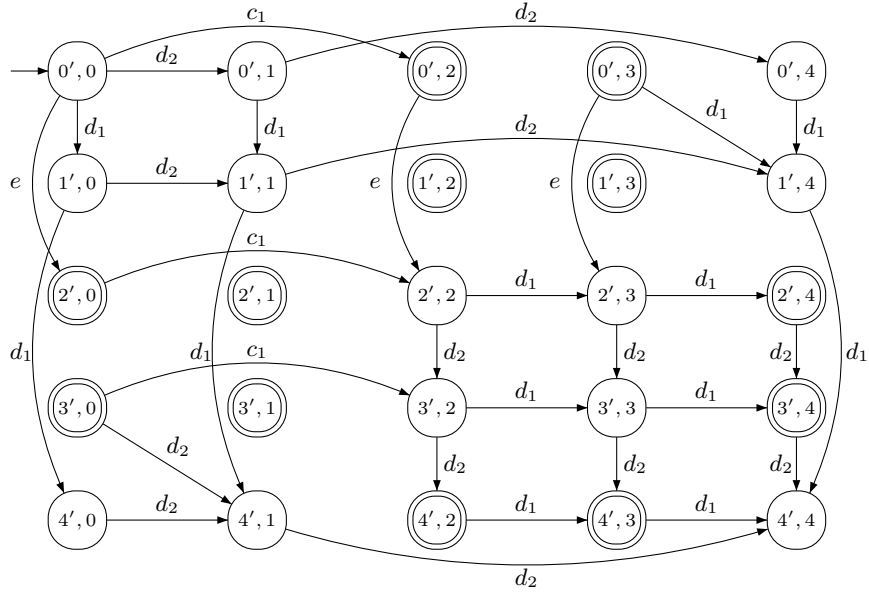

We check that every state in the direct product is reachable. Since $\mathcal{D}_\cup$, $\mathcal{D}_\oplus$, $\mathcal{D}_\setminus$, and $\mathcal{D}_\cap$ have the same structure as the direct product, this argument will apply to them as well. By Lemma~\ref{lem:properpermutations} there exist $u_1, u_2 \in \{a,b\}^*$ such that $u_1$ induces $(0', \dots, (m-2-j)')$ and $u_2$ induces $((m-1-j)', \dots, (m-1)')$ in $\mathcal D_{m,j}'$. Note that $u_1$ and $u_2$ permute $E_{n,k}$ and $F_{n,k}$ in $\mathcal D_{n,k}$.
Similarly, there exist $v_1, v_2 \in \{a,b\}^*$ such that $v_1$ induces $(0, \dots, n-2-k)$ and $v_2$ induces $(n-1-k, \dots, n-1)$ in $\mathcal D_{n,k}$, and they permute $E_{m,j}'$ and $F_{m,j}'$ in $\mathcal D_{m,j}'$.

\begin{enumerate}
\item State $(p', q) \in W$ is reachable since $(0',0) \xrightarrow{d_1^p} (p',0) \xrightarrow{d_2^q} (p', q)$.
\item State $(p', 0) \in Y$ is reachable since $(0',0) \xrightarrow{e} ((m-1-j)', 0) \xrightarrow{(d_2e)^{p-(m-1-j)}} (p',0)$. 
 An arbitrary $(p', q) \in Y$ is then reached by $v_1^q$ from some $(r',0)$ where $r' \in F_{m,j}'$ is chosen so that $r' \xrightarrow{v_1^q} p'$ in $\mathcal D_{m,j}'$.
\item State $(p', q) \in X$ is reachable by symmetry with {\bf 2}.
\item State $(p', q) \in Z$ is reachable since $(0',0) \xrightarrow{ec_1} ((m-1-j)', n-1-k) \xrightarrow{d_2^{p-(m-1-j)}} (p',n-1-k) \xrightarrow{d_1^{q-(n-1-k)}} (p', q)$.
\item State $(p', n-1) \in C$ is reachable since $(0',0) \xrightarrow{d_2^{n-1-k}} (0', n-1)$, and $p'$ is reachable in $\mathcal D_{m,j}'$.
\item State $((m-1)', q) \in R$ is reachable by symmetry with {\bf 5}.
\end{enumerate}
Hence all states are reachable.

As a tool for distinguishability, we show that the states of $S$ are distinguishable with respect to $R \cup C$; that is, for any pair of distinct states in $S$, we show that there is a word that sends one state to $R \cup C$ and leaves the other state in $S$. We check this fact in cases. Note that $d_2$ fixes the states of $X$ and $d_1$ fixes the states of $Y$.
\begin{enumerate}
\item States of $W$ and $X$ are distinguished by words in $d_2^*$.
\item States of $W$ and $Y$ are distinguished by words in $d_1^*$.
\item States of $X$ and $Y$ are distinguished by words in $d_1^*$.
\item States of $X$ and $Z$ are distinguished by words in $d_2^*$.
\item States of $Y$ and $Z$ are distinguished by words in $d_1^*$.
\item To distinguish states of $W$ and $Z$, we reduce to {\bf 5} by a word in $u_1^*e$.
\item Any two states of $W$ are distinguished by a word in $d_1^*$ if they differ in the first coordinate, or by a word in $d_2^*$ if they differ in the second coordinate.
\item Any two states of $Z$ are distinguished by a word in $d_2^*$ if they differ in the first coordinate, or by a word in $d_1^*$ if they differ in the second coordinate.
\item To distinguish two states of $X$, reduce to {\bf 4} by a word in $u_1^*e$ if they differ in the first coordinate, or reduce to {\bf 8} by a word in $u_1^*e$ if the first coordinate is the same.
\item Any two states of $Y$ are distinguishable by symmetry with {\bf 9}.
\end{enumerate}
Now we determine which states are pairwise distinguishable with respect to the final states of $\mathcal D_\circ$ for each $\circ \in \{\cup, \oplus, \setminus, \cap\}$.
Let $w = (u_1e)^{m-1-j}(v_1c_1)^{n-1-k}$; observe that $w$ maps every state of $S$ to a state of $Z$.

$\cup$, $\oplus$: In $\mathcal D_\cup$, $(p',q)$ is final if $p' \in F_{m,j}'$ or $q \in F_{n,k}$.
In $\mathcal D_\oplus$, $(p',q)$ is final if $p' \in F_{m,j}'$ and $q \not\in F_{n,k}$ or $p' \not\in F_{m,j}'$ and $q \in F_{n,k}$.
We show that all $mn$ states are pairwise distinguishable in both cases.

The states of $R$ are pairwise distinguishable by the minimality of $\mathcal D_{n,k}$.
Similarly, the states of $C$ are pairwise distinguishable by the minimality of $\mathcal D_{m,j}'$.
The states of $R$ and $C$ are distinguishable by $wd_1^k$, since $R \setminus \{((m-1)', n-1)\} \xrightarrow{w} \{(m-1)'\} \times F_{n,k} \xrightarrow{d_1^k} \{(m-1)',n-1\}$ and $C \setminus \{((m-1)', n-1)\} \xrightarrow{w} F_{m,j}' \times \{n-1\} \xrightarrow{d_1^k} F_{m,j}' \times \{n-1\}$.
The states of $C$ and $S$ are distinguishable since $S \xrightarrow{w} Z \xrightarrow {d_2^j} \{(m-1)'\} \times F_{n,k} \subseteq R$, and we can distinguish states of $R$ and $C$.
The states of $R$ and $S$ are similarly distinguishable.
Finally, states of $S$ are pairwise distinguishable because they can be distinguished with respect to $R \cup C$, and we can distinguish states of $S$ and $R \cup C$.

$\setminus$: In $\mathcal D_\setminus$, $(p',q)$ is final if $p' \in F_{m,j}'$ and $q \not\in F_{n,k}$. The states of $R$ are all empty, and the remaining states are pairwise distinguishable for a total of $mn - (n - 1)$ distinguishable states.

The states of $C$ are pairwise distinguishable by the minimality of $\mathcal D_{m,j}'$.
The states of $C$ and $S$ are distinguishable since $S \xrightarrow{w} Z \xrightarrow {d_2^j} \{(m-1)'\} \times F_{n,k} \subseteq R$, and every state in $R$ is empty.
Finally, states of $S$ are pairwise distinguishable because they can be distinguished with respect to $R \cup C$, and we can distinguish states of $S$ and $R \cup C$.

$\cap$: In $\mathcal D_\cap$ the final state set is $Z$. The states of $R \cup C$ are all empty, leaving $mn - (m + n + 2)$ distinguishable states. The states of $S$ are non-empty since $S \xrightarrow{w} Z$. We can distinguish the states of $S$ with respect to $R \cup C$; hence they are pairwise distinguishable. \qed
\end{proof}

\section{Conclusions}\label{sec:conclusions}
 
 Our results are summarized in Table~\ref{tab:summary}. The largest bounds are shown in boldface type, and they are reached  in the classes of ideal languages, closed languages, and proper languages.
Recall that for regular languages we have the following results: 
semigroup: $n^n$; 
reverse: $2^n$; 
star: $2^{n-1}+2^{n-2}$; 
restricted product:  $(m-1)2^n+2^{n-1}$;
unrestricted product: $m2^n+2^{n-1}$;  
restricted $\cup$ and $\oplus$: $mn$; 
unrestricted $\cup$ and $\oplus$: $(m+1)(n+1)$;
restricted $\setminus$: $mn$;
unrestricted $\setminus$: $mn+m$;
restricted $\cap$: $mn$; 
unrestricted $\cap$: $mn$. 
In Table~\ref{tab:summary}, $R$ and $U$ stand for \emph{restricted} and \emph{unrestricted}, respectively.
 \renewcommand{\arraystretch}{1.3}
\begin{table}[ht]
\caption{Complexities of prefix-convex languages}
\label{tab:summary}
\scriptsize
\begin{center}
$
 \begin{array}{|c||c|c|c|c||}    
\hline
 & \ Ideal \ & \ Closed  \ &  \ Free\  &\  Proper \     \\
\hline \hline
\ Semigroup \   	
 &\ \mathbf{n^{n-1}} \	&\mathbf{n^{n-1}} &n^{n-2} & n^{n-1-k}(k+1)^k      \\
\hline
\ Reverse \   
 &\ \mathbf{2^{n-1}} \	& \mathbf{2^{n-1}} &\ 2^{n-2} +1 \ &\mathbf{2^{n-1} } \\
\hline
\ Star \   
&\ n+1 \	& \ 2^{n-2}+1 \ &\ n \ & \mathbf{2^{n-2} +2^{n-2-k}+1}  \\
\hline
\ Product\;  R\   &
\ m+2^{n-2} \	& \ \mathbf{(m+1)2^{n-2}} \ &\ m+n-2 \ & m-1-j + j2^{n-2} +2^{n-1}  \\
\hline
\ Product\; U\   &
\ m+2^{n-1}+2^{n-2}+1 \	& \ \mathbf{(m+1)2^{n-2}} \ &\ m+n-2 \ & m-1-j + j2^{n-2} +2^{n-1}    \\

\hline
\ \cup  \; R \   &	
mn-(m+n-2) 	& \ \mathbf{mn} \ &\ mn-2 \ & \mathbf{mn}   \\
\hline
\ \cup \;   U \   &	
 \mathbf{(m+1)(n+1)} 	& \ mn \ &\ mn-2 \ & mn   \\
 \hline
\ \oplus\;  R \   &	
\mathbf{mn} 	& \ \mathbf{mn} \ &\ mn-2 \ & \mathbf{mn}   \\
 \hline
\ \oplus\;   U \   &	
 \mathbf{(m+1)(n+1)} 	& \ mn \ &\ mn-2 \ & mn   \\
 \hline
\setminus\; R \   &	
\mathbf{mn-(m-1)} \	& \ \mathbf{mn-(n-1)} \ & mn-(m+2n-4) & \mathbf{mn-(n-1)}  \\
\hline
\setminus\; U \   &	
 \mathbf{mn+m} \	& \ mn-(n-1) \ & mn-(m+2n-4) & mn-(n-1)   \\
\hline
\cap\; R \text{ and } U \   &	
\ \mathbf{mn} \	&\ mn-(m+n-2) \ & mn-2(m+n-3)  & mn -(m+n-2)  \\
\hline
\end{array} 
$
\end{center}
\end{table}

\providecommand{\noopsort}[1]{}


\begin{thebibliography}{10}

\bibitem{AnBr09}
T.~Ang and J.~Brzozowski.
\newblock Languages convex with respect to binary relations, and their closure
  properties.
\newblock {\em Acta Cybernet.}, 19(2):445--464, 2009.

\bibitem{BBMR14}
J.~Bell, J.~Brzozowski, N.~Moreira, and R.~Reis.
\newblock Symmetric groups and quotient complexity of boolean operations.
\newblock In J.~Esparza and et~al., editors, {\em ICALP 2014}, volume 8573 of
  {\em LNCS}, pages 1--12. Springer, 2014.

\bibitem{BPR10}
J.~Berstel, D.~Perrin, and C.~Reutenauer.
\newblock {\em Codes and Automata (Encyclopedia of Mathematics and its
  Applications)}.
\newblock Cambridge University Press, 2010.

\bibitem{Brz10}
J.~Brzozowski.
\newblock Quotient complexity of regular languages.
\newblock {\em J. Autom. Lang. Comb.}, 15(1/2):71--89, 2010.

\bibitem{Brz13}
J.~Brzozowski.
\newblock In search of the most complex regular languages.
\newblock {\em Int. J. Found. Comput. Sci,}, 24(6):691--708, 2013.

\bibitem{Brz16}
J.~Brzozowski.
\newblock Unrestricted state complexity of binary operations on regular
  languages.
\newblock In {\em DCFS 2016}, LNCS. Springer, 2016.
\newblock To appear. Also {\tt http://arxiv.org/abs/1602.01387}.

\bibitem{BrDa14}
J.~Brzozowski and G.~Davies.
\newblock Maximally atomic languages.
\newblock In Z.~\"Esik and Z~F\"ulop, editors, {\em Automata and Formal
  Languages (AFL 2014)}, pages 151--161. EPTCS, 2014.

\bibitem{BrDa15}
J.~Brzozowski and S.~Davies.
\newblock Quotient complexities of atoms in regular ideal languages.
\newblock {\em Acta Cybernet.}, 22(2):293--311, 2015.

\bibitem{BDL15}
J.~Brzozowski, S.~Davies, and B.~Y.~V. Liu.
\newblock Most complex regular ideal languages.
\newblock {\small\tt http://arxiv.org/abs/1511.00157}.

\bibitem{BJL13}
J.~Brzozowski, G.~Jir{\'a}skov{\'a}, and B.~Li.
\newblock Quotient complexity of ideal languages.
\newblock {\em Theoret. Comput. Sci.}, 470:36--52, 2013.

\bibitem{BJZ14}
J.~Brzozowski, G.~Jir{\'a}skov{\'a}, and C.~Zou.
\newblock Quotient complexity of closed languages.
\newblock {\em Theory Comput. Syst.}, 54:277--292, 2014.

\bibitem{BLY12}
J.~Brzozowski, B.~Li, and Y.~Ye.
\newblock Syntactic complexity of prefix-, suffix-, bifix-, and factor-free
  regular languages.
\newblock {\em Theoret. Comput. Sci.}, 449:37--53, 2012.

\bibitem{BrLi12}
J.~Brzozowski and B.~Liu.
\newblock Quotient complexity of star-free languages.
\newblock {\em Int. J. Found. Comput. Sci.}, 23(6):1261--1276, 2012.

\bibitem{BSY15}
J.~Brzozowski, M.~Szyku{\l}a, and Y.~Ye.
\newblock Syntactic complexity of regular ideals.
\newblock {\small\tt http://arxiv.org/abs/1509.06032}.

\bibitem{BrTa13}
J.~Brzozowski and H.~Tamm.
\newblock Quotient complexities of atoms of regular languages.
\newblock {\em Int. J. Found. Comput. Sci.}, 24(7):1009--1027, 2013.

\bibitem{BrTa14}
J.~Brzozowski and H.~Tamm.
\newblock Theory of \'atomata.
\newblock {\em Theoret. Comput. Sci.}, 539:13--27, 2014.

\bibitem{BrYe11}
J.~Brzozowski and Y.~Ye.
\newblock Syntactic complexity of ideal and closed languages.
\newblock In G.~Mauri and A.~Leporati, editors, {\em DLT 2011}, volume 6795 of
  {\em LNCS}, pages 117--128. Springer Berlin / Heidelberg, 2011.

\bibitem{HSW09}
Y.-S. Han, K.~Salomaa, and D.~Wood.
\newblock Operational state complexity of prefix-free regular languages.
\newblock In Z.~\'Esik and Z~F\"ul\"op, editors, {\em Automata, Formal
  Languages, and Related Topics}, pages 99--115. Institute of Informatics,
  University of Szeged, Hungary, 2009.

\bibitem{HoKo04}
M.~Holzer and B.~K\"{o}nig.
\newblock On deterministic finite automata and syntactic monoid size.
\newblock {\em Theoret. Comput. Sci.}, 327(3):319--347, 2004.

\bibitem{Iva16}
S.~Iv\'an.
\newblock Complexity of atoms, combinatorially.
\newblock {\em Inform. Process. Lett.}, 116(5):356--360, 2016.

\bibitem{JiKr10}
G.~Jir\'askov\'a and M.~Krausov\'a.
\newblock Complexity in prefix-free regular languages.
\newblock In I.~McQuillan, G.~Pighizzini, and B.~Trost, editors, {\em
  Proceedings of the 12th International Workshop on Descriptional Complexity of
  Formal Systems $($DCFS\/$)$}, pages 236--244. University of Saskatchewan,
  2010.

\bibitem{Kra11}
M.~Krausov{\'a}.
\newblock Prefix-free regular languages: Closure properties, difference, and
  left quotient.
\newblock In Z.~Kot{\'a}sek, J.~Bouda, I.~Cern{\'a}, L.~Sekanina, T.~Vojnar,
  and D.~Antos, editors, {\em MEMICS}, volume 7119 of {\em Lecture Notes in
  Computer Science}, pages 114--122. Springer, 2011.

\bibitem{KLS05}
B.~Krawetz, J.~Lawrence, and J.~Shallit.
\newblock State complexity and the monoid of transformations of a finite set.
\newblock In M.~Domaratzki, A.~Okhotin, K.~Salomaa, and S.~Yu, editors, {\em
  Proceedings of the Implementation and Application of Automata, $($CIAA\/$)$},
  volume 3317 of {\em LNCS}, pages 213--224. Springer, 2005.

\bibitem{Myh57}
J.~Myhill.
\newblock Finite automata and representation of events.
\newblock {\em Wright Air Development Center Technical Report}, 57--624, 1957.

\bibitem{Pin97}
J.-E. Pin.
\newblock Syntactic semigroups.
\newblock In {\em Handbook of Formal Languages, vol.~1: Word, Language,
  Grammar}, pages 679--746. Springer, New York, NY, USA, 1997.

\bibitem{SWY04}
A.~Salomaa, D.~Wood, and S.~Yu.
\newblock On the state complexity of reversals of regular languages.
\newblock {\em Theoret. Comput. Sci.}, 320:315--329, 2004.

\bibitem{Thi73}
G.~Thierrin.
\newblock Convex languages.
\newblock In M.~Nivat, editor, {\em Automata, Languages and Programming}, pages
  481--492. North-Holland, 1973.

\bibitem{Yu01}
S.~Yu.
\newblock State complexity of regular languages.
\newblock {\em J. Autom. Lang. Comb.}, 6:221--234, 2001.

\end{thebibliography}

\end{document}